\documentclass[11pt]{article}

\usepackage{amsthm}
\usepackage{graphicx} % support the \includegraphics command and options
\usepackage{array} % for better arrays (eg matrices) in maths

\usepackage{amsmath, amssymb, amsfonts, verbatim}
\usepackage{hyphenat,epsfig,subcaption,multirow,mathtools}
\usepackage{nicefrac}
\usepackage{paralist}
\usepackage{thm-restate}
\usepackage{mleftright}
\usepackage{wrapfig}

\usepackage[font=small, labelfont=bf]{caption}

\usepackage[usenames,dvipsnames]{xcolor}
\usepackage[ruled]{algorithm2e}

\DeclareFontFamily{U}{mathx}{\hyphenchar\font45}
\DeclareFontShape{U}{mathx}{m}{n}{
	<5> <6> <7> <8> <9> <10>
	<10.95> <12> <14.4> <17.28> <20.74> <24.88>
	mathx10
}{}
\DeclareSymbolFont{mathx}{U}{mathx}{m}{n}
\DeclareMathSymbol{\bigtimes}{1}{mathx}{"91}

\usepackage{tcolorbox}
\tcbuselibrary{skins,breakable}
\tcbset{enhanced jigsaw}

\usepackage[normalem]{ulem}
\usepackage[compact]{titlesec}

\definecolor{DarkRed}{rgb}{0.5,0.1,0.1}
\definecolor{DarkBlue}{rgb}{0.1,0.1,0.5}

\usepackage{nameref}
\definecolor{ForestGreen}{rgb}{0.1333,0.5451,0.1333}
\definecolor{Red}{rgb}{0.9,0,0}
\usepackage[linktocpage=true,
pagebackref=true,colorlinks,
linkcolor=DarkRed,citecolor=ForestGreen,
bookmarks,bookmarksopen,bookmarksnumbered]
{hyperref}
\usepackage[noabbrev,nameinlink]{cleveref}
\crefname{property}{property}{Property}
\creflabelformat{property}{(#1)#2#3}
\crefname{equation}{eq}{Eq}
\creflabelformat{equation}{(#1)#2#3}
\usepackage{bm}
\usepackage{url}
\usepackage{xspace}
\usepackage[mathscr]{euscript}

\usepackage{tikz}
\usetikzlibrary{arrows}
\usetikzlibrary{arrows.meta}
\usetikzlibrary{shapes}
\usetikzlibrary{backgrounds}
\usetikzlibrary{positioning}
\usetikzlibrary{decorations.markings}
\usetikzlibrary{patterns}
\usetikzlibrary{calc}
\usetikzlibrary{fit}
\usetikzlibrary{decorations}

\usepackage[framemethod=tikz]{mdframed}

\usepackage[noend]{algpseudocode}
\makeatletter
\def\BState{\State\hskip-\ALG@thistlm}
\makeatother

\usepackage{cite}
\usepackage{enumitem}

\usepackage[margin=1in]{geometry}

\newtheorem{theorem}{Theorem}
\newtheorem{lemma}{Lemma}[section]
\newtheorem{proposition}[lemma]{Proposition}

\newtheorem{claim}[lemma]{Claim}
\newtheorem{fact}[lemma]{Fact}

\newtheorem{definition}[lemma]{Definition}

\newtheorem*{claim*}{Claim}
\newtheorem*{assumption*}{Assumption}
\newtheorem*{proposition*}{Proposition}
\newtheorem*{lemma*}{Lemma}
\newtheorem*{problem5*}{Problem}

\crefname{lemma}{Lemma}{Lemmas}
\crefname{figure}{Figure}{Figures}
\crefname{claim}{claim}{claims}

\newtheorem{mdresult}{Result}
\newenvironment{result}{\begin{mdframed}[backgroundcolor=lightgray!40,topline=false,rightline=false,leftline=false,bottomline=false,innertopmargin=2pt]\begin{mdresult}}{\end{mdresult}\end{mdframed}}

\newtheorem{mddistribution}{Distribution}
\newenvironment{Distribution}{\begin{mdframed}[roundcorner=10pt, hidealllines=false,innerleftmargin=5pt,backgroundcolor=white!10,innertopmargin=2pt]\begin{mddistribution}}{\end{mddistribution}\end{mdframed}}

\newtheorem{mdalgorithm}{Algorithm}
\newenvironment{Algorithm}{\begin{mdframed}[hidealllines=false,innerleftmargin=10pt,backgroundcolor=white!10,innertopmargin=5pt,innerbottommargin=5pt,roundcorner=10pt]\begin{mdalgorithm}}{\end{mdalgorithm}\end{mdframed}}

\newtheorem{observation}[lemma]{Observation}

\newtheoremstyle{restate}{}{}{\itshape}{}{\bfseries}{~(restated).}{.5em}{\thmnote{#3}}
\theoremstyle{restate}
\newtheorem*{restate}{}

\theoremstyle{definition}
\newtheorem{mdinvariant}[lemma]{Definition}
\newenvironment{Definition}{\begin{mdframed}[roundcorner=10pt, hidealllines=false,innerleftmargin=5pt,backgroundcolor=white!10,innertopmargin=2pt]\begin{mdinvariant}}{\end{mdinvariant}\end{mdframed}}

\crefname{mdinvariant}{Definition}{Definitions}

\allowdisplaybreaks

\renewcommand{\geq}{\geqslant}
\renewcommand{\leq}{\leqslant}
\renewcommand{\ge}{\geq}
\renewcommand{\le}{\leq}

\newcommand{\Leq}[1]{\ensuremath{\underset{\textnormal{#1}}\leq}}

\newcommand{\Eq}[1]{\ensuremath{\underset{\textnormal{#1}}=}}

\newcommand{\tvd}[2]{\ensuremath{\norm{#1 - #2}_{\mathrm{tvd}}}}
\newcommand{\Ot}{\ensuremath{\widetilde{O}}}
\newcommand{\eps}{\ensuremath{\varepsilon}}

\newcommand{\Bracket}[1]{\Big[#1\Big]}
\newcommand{\bracket}[1]{\left[#1\right]}
\newcommand{\paren}[1]{\ensuremath{\left(#1\right)}\xspace}
\newcommand{\card}[1]{\left\vert{#1}\right\vert}

\newcommand{\norm}[1]{\ensuremath{\|#1\|}}

\newcommand{\set}[1]{\ensuremath{\left\{ #1 \right\}}}
\newcommand{\poly}{\mbox{\rm poly}}
\newcommand{\polylog}{\mbox{\rm  polylog}}

\newcommand{\alg}{\ensuremath{\mathcal{A}}\xspace}
\newcommand{\dist}{\textnormal{\text{dist}}}

\DeclareMathOperator*{\Exp}{\ensuremath{{\mathbb{E}}}}
\DeclareMathOperator*{\Prob}{\ensuremath{\textnormal{Pr}}}
\renewcommand{\Pr}{\Prob}

\newcommand{\Ex}{\Exp}

\newenvironment{tbox}{\begin{tcolorbox}[
		enlarge top by=5pt,
		enlarge bottom by=5pt,
		 breakable,
		 boxsep=0pt,
                  left=4pt,
                  right=4pt,
                  top=10pt,
                  arc=0pt,
                  boxrule=1pt,toprule=1pt,
                  colback=white
                  ]%%
	}
{\end{tcolorbox}}

\newcommand{\event}{\ensuremath{\mathcal{E}}}

\newcommand{\rv}[1]{\ensuremath{{\mathsf{#1}}}\xspace}
\newcommand{\rA}{\rv{A}}
\newcommand{\rB}{\rv{B}}
\newcommand{\rC}{\rv{C}}
\newcommand{\rD}{\rv{D}}

\newcommand{\cX}{\mathcal{X}}
\newcommand{\cY}{\mathcal{Y}}

\newcommand{\supp}[1]{\ensuremath{\textnormal{\text{supp}}(#1)}}
\newcommand{\distribution}[1]{\ensuremath{\textnormal{dist}(#1)}\xspace}

\newcommand{\kl}[2]{\ensuremath{\mathbb{D}(#1~||~#2)}}
\newcommand{\II}{\ensuremath{\mathbb{I}}}
\newcommand{\HH}{\ensuremath{\mathbb{H}}}
\newcommand{\mi}[2]{\ensuremath{\def\mione{#1}\def\mitwo{#2}\mireal}}
\newcommand{\mireal}[1][]{
  \ifx\relax#1\relax%
    \II(\mione \,; \mitwo)%
  \else%
    \II(\mione \,; \mitwo\mid #1)%
  \fi
}
\newcommand{\en}[1]{\ensuremath{\HH(#1)}}

\newcommand{\itfacts}[1]{\Cref{fact:it-facts}-(\ref{part:#1})\xspace}

\newcommand{\OR}{\ensuremath{\textnormal{OR}}}

\newcommand{\ORPPC}{\ensuremath{\textnormal{\texttt{OR-PPC}}}}

\newcommand{\rX}{\ensuremath{\rv{X}}}
\newcommand{\rY}{\ensuremath{\rv{Y}}}

\newcommand{\cc}[1]{\ensuremath{\textsc{cc}(#1)}}
\newcommand{\prot}{\ensuremath{\pi}}
\newcommand{\Prot}{\ensuremath{\Pi}}

\newcommand{\cL}{\ensuremath{\mathcal{L}}}
\newcommand{\point}[2]{\ensuremath{\textsc{last}_{#1}(#2)}}

\newcommand{\PC}{\ensuremath{\textnormal{\texttt{PC}}}}

\newcommand{\PPC}{\ensuremath{\textnormal{\texttt{PPC}}}}

\newcommand{\product}[2]{\ensuremath{#1 \times #2}}
\newcommand{\join}[2]{\ensuremath{#1 + #2}}

\newcommand{\Gprod}{\ensuremath{G^{\times}}}
\newcommand{\Vprod}{\ensuremath{V^{\times}}}
\newcommand{\Eprod}{\ensuremath{E^{\times}}}
\newcommand{\Mprod}{\ensuremath{M^{\times}}}

\newcommand{\Vjoin}{\ensuremath{V^{+}}}
\newcommand{\Ejoin}{\ensuremath{E^{+}}}

\newcommand{\mutilPC}{\ensuremath{\tilde{\mu}_{\PC}}}

\newcommand{\muPC}{\ensuremath{\mu_{\PC}}}
\newcommand{\nuPC}{\ensuremath{\nu_{\PC}}}

\newcommand{\muPPC}{\ensuremath{\mu_{\PPC}}}

\newcommand{\Stil}{\ensuremath{\tilde{S}}}
\newcommand{\Ttil}{\ensuremath{\tilde{T}}}

\newcommand{\pt}[1]{\ensuremath{P(#1)}}

\newcommand{\rGam}{\ensuremath{\bm{\Gamma}}}
\newcommand{\Gam}{\ensuremath{{\Gamma}}}
\newcommand{\rProt}{\ensuremath{\bm{\Prot}}}

\newcommand{\rpt}[1]{\ensuremath{\rv{P}(#1)}}

\newcommand{\cUell}{\ensuremath{\mathcal{U}_{\ell}}}
\newcommand{\rM}{\ensuremath{\rv{M}}}

\newcommand{\randvert}{\ensuremath{\rv{v}}}

\newcommand{\rT}{\ensuremath{\rv{T}}}

\newcommand{\rProtAl}{\ensuremath{\rProt^{\textnormal{A}}}}
\newcommand{\rProtBob}{\ensuremath{\rProt^{\textnormal{B}}}}

\newcommand{\ents}[2]{\ensuremath{\mathbb{H}^{(#1)}(#2)}}
\newcommand{\rModd}{\ensuremath{\rM^{\textnormal{odd}}}}
\newcommand{\rMeven}{\ensuremath{\rM^{\textnormal{even}}}}
\newcommand{\Meven}{\ensuremath{M^{\textnormal{even}}}}

\newcommand{\outB}{\ensuremath{\textsc{out}}}

\newcommand{\ver}[1]{\ensuremath{v_{#1}}}

\newcommand{\rb}{\ensuremath{\rv{b}}}

\newcommand{\istar}{\ensuremath{i^{\star}}}

\newcommand{\protOR}{\ensuremath{\prot_{\OR}}}

\newcommand{\cOR}{\ensuremath{c_{\OR}}}

\newcommand{\protPPC}{\ensuremath{\prot_{\PPC}}}

\renewcommand{\eps}{\epsilon}

\title{Better Bounds for Semi-Streaming Single-Source Shortest Paths}
\author{
Sepehr Assadi\footnote{(sepehr@assadi.info) School of Computer Science, University of Waterloo. 
Supported in part by a Sloan Research Fellowship, an NSERC
Discovery Grant, and a Faculty of Math Research Chair grant. \smallskip} 
\\ {\small University of Waterloo} \and
Gary Hoppenworth\footnote{(garytho@umich.edu) Computer Science Department, University of Michigan. Supported in part by National Science Foundation grant NSF:AF 2153680. \smallskip}
 \\ {\small University of Michigan} \and
Janani Sundaresan\footnote{(jsundaresan@uwaterloo.ca) School of Computer Science, University of Waterloo. Supported in part by a Cheriton Scholarship from School of Computer Science, a Faculty of Math Graduate Research Excellence Award, and Sepehr Assadi's NSERC Discovery Grant. \smallskip} \\ {\small University of Waterloo}
}
\date{}

\begin{document}
	\maketitle
	
	\pagenumbering{roman}
	
	% !TeX root = main.tex 
%!TEX root = main.tex
\begin{abstract}

\bigskip

In the semi-streaming model, an algorithm must process any $n$-vertex graph by making one or few passes over a stream of its edges, use $\Ot(n) := O(n \cdot \poly\!\log\!{(n)})$ words of space, and at the end of the last pass, output a solution to the problem at hand. Approximating (single-source) shortest paths on \emph{undirected} graphs is a longstanding open question in this model. 

\medskip

In this work, we make progress
on this question from both upper and lower bound fronts: 

\begin{itemize}
	\item We present a simple randomized algorithm that for any $\eps > 0$, with high probability computes $(1+\eps)$-approximate shortest paths from a given source vertex in 
	\[
	O\left(\frac{1}{\eps} \cdot n \log^3n \right)~\text{space} \quad \text{and} \quad O\left(\frac{1}{\eps} \cdot \paren{\frac{\log{n}}{\log\log{n}}}^2\right) ~\text{passes}.
	\] 
	The algorithm can also be derandomized and made  to work on dynamic streams  at a cost of some extra $\poly(\log{n},1/\eps)$ factors only in the space. 
	
	Previously, the best known algorithms for this problem required $1/\eps \cdot \log^{c}(n)$ passes, for an    unspecified large constant $c$. 
	
	\medskip
	
	\item  We prove that any semi-streaming algorithm that with large constant probability outputs any constant approximation to shortest paths from a given source vertex (even to a single fixed target vertex
	and only the distance, not necessarily the path)
	 requires 
	\[
		\Omega\left(\frac{\log{n}}{\log\log{n}}\right) ~\text{passes}. 
	\]
	We emphasize that our lower bound holds for \textit{any} constant-factor approximation of shortest paths. 
	Previously, only constant-pass lower bounds were known and only for small approximation ratios below two.  
	\end{itemize}

\medskip
\noindent
Our results collectively reduce the gap in the pass complexity of approximating single-source shortest paths in the semi-streaming model from $\poly\!\log{\!(n)}$ vs $\omega(1)$ to only a quadratic gap. 

\end{abstract}
	
	\clearpage
	
	\setcounter{tocdepth}{3}
	\tableofcontents
	
	\clearpage
	
	\pagenumbering{arabic}
	\setcounter{page}{1}
	
	% !TEX root = main.tex
\section{Introduction}

In the semi-streaming model for processing massive graphs---introduced by~\cite{FeigenbaumKMSZ04}---the edges of an $n$-vertex graph $G= (V,E)$ are
presented to an algorithm one-by-one in an arbitrarily ordered stream. A semi-streaming algorithm then is allowed to make one or few passes over this stream and use $\Ot(n) := O(n \cdot \poly\!\log{\!(n)})$ space to solve a given problem. 
We study the \textbf{single-source shortest path} problem in \textbf{undirected} (and possibly \textbf{weighted}) graphs in this model. %\footnote{We discuss the state-of-the-art for directed graphs briefly later in the introduction.}.

Alongside the introduction of the semi-streaming model,~\cite{FeigenbaumKMSZ04} presented a single-pass semi-streaming algorithm for $O\left(\frac{\log{n}}{\log\log{n}}\right)$-approximation of single-source shortest paths, 
which was shown to be optimal soon after in~\cite{FeigenbaumKMSZ05}. This led~\cite{FeigenbaumKMSZ04} to pose the following open question for undirected shortest paths: 
\begin{quote}
	\cite{FeigenbaumKMSZ04}: \emph{``Is there a multipass algorithm that improves the approximation ratio?''}
\end{quote}

Since then, there has been a long line of research on this question using a wide range of techniques, primarily focusing on proving upper and lower bounds on the number of passes needed for obtaining a constant or $(1+\eps)$-approximation~\cite{ElkinZ04,GuruswamiO13,HenzingerKN16,ElkinN18,ElkinN19,AssadiR20,BeckerFKL21,ChenKPSSY21,ElkinT22,AssadiJJST22,AssadiS23}. 
%%\footnote{While we only work with undirected graphs in this paper, it is worth quickly mentioning that
%%for \emph{directed} graphs, we still do not even know better than $\sqrt{n}$ pass algorithms for $s$-$t$ reachability~\cite{AssadiJJST22}, and hence, any multiplicative approximation ratio to shortest paths.}.

\paragraph{Upper bounds.} The first of such algorithms appeared in the work of~\cite{ElkinZ04} who obtained a semi-streaming $O_{\eps}(n^{\delta})$-pass algorithm for any small \emph{constant} $\delta > 0$. 
This algorithm is based on the idea of \emph{hybrid spanners} (see~\cite{ElkinP04}). 
The next big improvements came from using \emph{hopsets} to reduce the passes to $(1/\eps)^{O(\sqrt{\log{n}\cdot\log\log{n}})}$~\cite{HenzingerKN16,ElkinN19}, which is $n^{o(1)}$ for constant $\eps > 0$. 
Finally,~\cite{BeckerFKL21}, using \emph{continuous optimization} tools, 
reduced the number of passes to $\poly\!\log{(n)}/\eps^2$, and~\cite{AssadiJJST22} further improved the $\eps$-dependency in passes to $1/\eps$; the dependency
on $\log{n}$ in both these algorithms is some unspecified large polynomial (which appears to be at least $(\log{n})^7$ in~\cite{BeckerFKL21} and $(\log{n})^{10}$ in~\cite{AssadiJJST22}). 

\paragraph{Lower bounds.} On the lower bound front,~\cite{FeigenbaumKMSZ05} first proved that finding $(1+\eps)$-approximate shortest paths 
in the semi-streaming model requires $\Omega(1/\eps)$ passes as long as $\eps \geq (\frac{\log\!\log{n}}{\log{n}})$. This lower bound was extended to the (algorithmically easier) problem of 
$(1+\eps)$-approximation of the shortest path \emph{distance} (and not necessarily finding a path) between a fixed source and target in~\cite{GuruswamiO13}. The space complexity in the latter result has been improved to $n^{2-o(1)}$ space in a series of work~\cite{AssadiR20,ChenKPSSY21,AssadiS23}, although they do not imply new pass lower bounds over~\cite{GuruswamiO13} for semi-streaming algorithms. 
%It is worth emphasizing that for (small) constant values of $\eps > 0$, these work only imply $\Omega(1)$ pass lower bounds. 

\medskip 
\noindent 
The current state-of-affairs for constant-factor approximation of shortest paths is thus the following: 
\begin{itemize}
	\item $\polylog{(n)}$-pass upper bounds for arbitrarily small constant approximation ratios, and
	\item $\omega(1)$-pass lower bounds for sufficiently small constant approximation ratios. 
	%and $\Omega(\frac{\log{\!(n)}}{\log\log{\!(n)}})$ passes for exact algorithms. 
\end{itemize}
%See~\Cref{table:prior} for a more comprehensive overview of prior work and our results. 
%%Further attempts in simplifying these work have been done in~\cite{Zuzic} which 
%%can be used to o

%%
%%obtained an $\tilde{O_{\eps}}(n^{1+\rho})$ space algorith m in $(\frac{1}{\eps \cdot \rho})^{O(1/\rho)}$ passes (see also~\cite{ElkinN05}); for semi-streaming
%%algorithms, this amounts to an $O_{\eps}(n^{\alpha})$ pass algorithm for any constant $\alpha > 0$ (which can be reduced by increasing the space of the algorithm by logarithmic factors). 
%%Another improvement here was 

\subsection{Our Contributions} 

We make progress on this problem by establishing new upper and lower bounds.

Our upper bound holds for the  general problem of computing  $(1+\eps)$-approximate single-source shortest paths on weighted graphs undergoing edge insertions and deletions. On the other hand, our lower bound holds for the simplest version of the problem 
where we only need to compute a  large constant approximation of the distance  between two fixed vertices in an unweighted graph (without having to find a shortest path). 
Our results can be stated more generally in terms of space-pass tradeoffs beyond the semi-streaming space restrictions. Our upper bound is as follows. 

%%In the following, we state both our upper and lower bounds specifically for $(1+\eps)$-approximation of shortest 
%%paths in unweighted graphs via semi-streaming algorithms in insertion-only streams, which we find to be the main problem of interest and the key contribution of our work. We then mention further extensions of our results separately. 

\begin{result}\label{res:upper}
	For any integer $k \ge 1$ and parameter $\eps \in (0, 1)$, there is a randomized streaming algorithm for weighted graphs that, with (exponentially) high probability, computes
	$(1+\eps)$-approximate  shortest paths from a given source to all other vertices in an $n$-vertex graph using 
	\[
	O\!\left(\frac{1}{\eps} \cdot kn^{1+1/k} \log n \right)~\text{space} \quad \text{and} \quad O\!\left(\frac{k^2}{\eps}\right) ~\text{passes.} 
	\]
	Moreover, this algorithm can be extended to dynamic streams that include both edge insertions and deletions. 
	It can also be derandomized on insertion-only streams by increasing the space with some fixed $\poly(\log{\!(n)},1/\eps)$ factors. 
\end{result}

Plugging in $k = \log{n}/\!\log\log{n}$ to  \Cref{res:upper} implies a randomized semi-streaming algorithm for $(1+\eps)$-approximate single-source shortest paths with 
\[
O\left(\frac{1}{\eps} \cdot n \log^3 n \right)~\text{space} \quad \text{and} \quad O\left(\frac{1}{\eps} \cdot \paren{\frac{\log{n}}{\log\log{n}}}^2\right) ~\text{passes}.\footnote{We  measure space using $O(\log n)$-bit machine words.}
\] 
This improves upon prior algorithms for this problem by reducing the pass complexity from some unspecified polylogarithmic number of  passes to only $\approx \log^2\!{(n)}$ passes (and keeping $\eps$-dependency the same). 
Another benefit of this algorithm, which we will elaborate on further when we discuss our techniques, is that it is arguably much simpler than all existing approaches for this problem, especially the state-of-the-art methods that utilize convex optimization tools~\cite{BeckerFKL21,AssadiJJST22}.

For the most general version of the problem on dynamic streams, the current state-of-the-art algorithms (to the best of our knowledge) are due to~\cite{ElkinT22}. These algorithms, operating in a similar space of $\Ot(n^{1+1/k})$, require $\left(\frac{k}{\eps}\right)^{k}$ passes for unweighted graphs and $\left(\frac{k \cdot \log{n}}{\eps}\right)^{k}$ passes for weighted graphs.
\Cref{res:upper} reduces this pass complexity to only $O(k^2/\eps)$,  even for weighted graphs.

Switching to lower bounds, we can prove that the pass complexity of~\Cref{res:upper} is within at most a quadratic factor of the optimal bounds for any constant-factor approximation of the problem. 
\begin{result}\label{res:lower}
	For any $\alpha \ge 1$ and any integer $k$ sufficiently smaller than $\log{\!(n)}/\alpha$, any streaming algorithm  on unweighted $n$-vertex graphs that, with constant probability of success, can output an $\alpha$-approximation to the distance between a 	fixed pair of vertices requires 
	\[
	\Omega\left(\frac{1}{\poly(k)} \cdot n^{1+\frac{1}{\alpha \cdot k}}\right)~\text{space} \quad \text{or} \quad \Omega(k) ~\text{passes}.
	\]
\end{result}
Again, for semi-streaming algorithms in insertion-only streams and on unweighted graphs, our \Cref{res:lower} implies that \emph{any} constant approximation of shortest paths requires
\[
\Omega\paren{\frac{\log{n}}{\log\log{n}}} ~\text{passes.}
\] 
Previously, only single-pass lower bounds were known
for approximation ratios beyond two~\cite{FeigenbaumKMSZ05}, and only constant pass lower bounds for less-than-two approximation factors~\cite{GuruswamiO13,AssadiR20,ChenKPSSY21} (even on weighted graphs, for finding all single-source shortest paths,
and on dynamic streams). 

\medskip

Our results collectively show that the pass-complexity of (constant factor) approximation of single-source shortest paths in undirected graphs is between $\approx \log{\!(n)}$ to $\approx \log^2{\!(n)}$ passes. 

\subsection{Our Techniques}\label{sec:tech}

We briefly discuss the techniques behind our work in this subsection and postpone a more elaborate
discussion to our technical overview in~\Cref{sec:overview}. 

\paragraph{Upper bounds.} Our upper bound uses the recently developed framework of~\cite{Assadi24} for designing multi-pass semi-streaming algorithms for approximate matching via 
a direct application of the multiplicative weight update method (as opposed to prior linear programming based approaches in streaming, e.g.,~\cite{AhnG11,IndykMRUVY17,AhnG18}, based on the Plotkin-Shmoys-Tardos framework~\cite{PlotkinST91}). 
Our algorithm is simply as follows: sample $\Ot(n/\eps)$ edges from the graph, find a shortest path tree of the sampled edges, increase the ``importance'' of any edge of the original graph 
that violates triangle inequality with respect to this tree, and repeat while adjusting sampling probabilities to be proportional to importances of edges. A key additional aspect we need to handle in this algorithm, compared to~\cite{Assadi24} (and its follow up in~\cite{Quanrud24}), 
is to ensure that even edges that violate the triangle inequality cannot do so ``too dramatically''.\footnote{This issue does not appear in any form in~\cite{Assadi24} as the corresponding step is to consider edges not covered by a vertex cover, which has a Boolean nature as opposed to the integer value that an edge can violate triangle inequality in our setting. Note that while prior work of~\cite{Assadi24,Quanrud24} applied this framework to a covering/packing problem, our approach
	to shortest path is not based on a covering or packing formulation of the problem.} We handle this by additionally maintaining a \emph{spanner} of the input graph at all times, and we compute the shortest path tree of the sampled  edges together with the spanner edges. Our analysis then incorporates this additional aspect into these frameworks as well (in addition to  problem-specific steps such as proving a ``sampling lemma'' for shortest paths 
and recovering our final solution from the intermediate ones). 

Many of the recent advances in shortest path algorithms in streaming and related models have involved solving the more general \emph{transshipment} 
problem~\cite{AndoniSZ20,Li20,BeckerFKL21,AssadiJJST22,RozhonGHZL22,Zuzic23}. 
This is often attributed to the fact that ``[unlike shortest paths,] crude approximations to transshipment can be boosted to near-optimal approximate solutions''~\cite{Zuzic23}. 
However, these approaches lead to considerably complicated algorithms with various heavy subroutines (e.g. requiring approximate dual maximizers and not only primal ones~\cite{BeckerFKL21,Zuzic23})  and indirect steps (e.g. even recovering approximate shortest path trees from approximate transshipment solutions is a non-trivial step with its own complications; see, e.g.~\cite{AndoniSZ20,Li20,BeckerFKL21}). Given this, we believe our direct approach for approximating shortest paths 
without relying on detours to the transshipment problem can be of its own independent interest, even beyond the streaming model.

\paragraph{Lower bounds.} Our lower bound follows the work of~\cite{GuruswamiO13} (itself inspired by~\cite{FeigenbaumKMSZ05}) by reducing the problem to \emph{many} instances of the 
classical pointer chasing problem on $n$-vertex graphs (see, e.g.~\cite{NisanW91,Yehudayoff20} and~\Cref{sec:overview-lower}). The key difference is in the analysis of the underlying pointer chasing instances. The lower bound of~\cite{GuruswamiO13} 
works on a \emph{product} distribution of inputs which inherently only corresponds to finding the shortest path distance to within an \emph{additive} factor of two. On the other hand, we introduce and analyze a 
\emph{correlated} input distribution that allows us to prove a lower bound for constant factor \emph{multiplicative} approximation of shortest path distance. This part is the key technical contribution of our lower bound and is 
proven using a combination of probabilistic and combinatorial arguments to address the introduced correlation in the inputs. 

Finally, similar to~\cite{GuruswamiO13}, we also  extend this lower bound to ``OR of many pointer chasing instances'' using information complexity~\cite{ChakrabartiSWY01,Bar-YossefJKS02,BarakBCR10}. But, this part of our argument also differs almost entirely from~\cite{GuruswamiO13} since our underlying hard distribution is no longer a product distribution. We implement
this step by relying on the vast body of work on message compression and information complexity in a blackbox manner---instead of the adhoc arguments of~\cite{GuruswamiO13}---and prove the lower bound even for two-player
communication complexity.

\clearpage
	
	\clearpage

	% !TEX root = main.tex

\section{Technical Overview}\label{sec:overview}

In this section, we give a technical overview of both the upper bound and the lower bound. The techniques for our two main results are mostly disjoint, and therefore we cover them separately. 
The rest of the paper is self-contained and thus the reader can skip this section and proceed directly to the technical material beginning in \Cref{sec:prelim}.

\subsection{Upper Bound} \label{sec:overview-upper}

Since the entire proof of our upper bound in \Cref{sec:upper_rand} is quite short, we refrain from an extensive technical overview and instead focus on highlighting its main ideas and comparing it to prior work.

Our upper bound follows the ``sample-and-solve'' paradigm commonly used in graph streaming algorithms dating back to~\cite{LattanziMSV11,KumarMVV13}. 
In this framework, we  sample a sparse representation of the input graph from the stream. Once we have this sparse representation stored locally, we are able to process it without using additional resources, and then
use this information to update our sampling strategy for the next pass, until we eventually solve the problem. 

In the following, we focus specifically on the semi-streaming implementation of our algorithm in \Cref{res:upper} (for the case of $k = O(\log n)$, and with a slightly worse pass complexity up to a $(\log\log{n})^2$ factor). 
Given an input graph $G$ and a source vertex $s \in V(G)$, the high-level outline of our semi-streaming algorithm for approximate single-source shortest paths can be summarized as follows. 
\begin{Algorithm} \normalfont
A high-level description of our algorithm on an input $G=(V,E,w)$. 
\begin{enumerate}
	\item Compute an $O(\log n)$-spanner $H$ of $G$. Assign each edge in $G$ an ``importance value''.\footnote{All edges are equally important to start with.}
	\item  In one pass over the stream, sample a set $F \subseteq E(G)$ of  $\widetilde{O}(n/\eps)$ edges of $G$ according to their importance values, and  compute a shortest path tree $T$  of $H \cup F$ rooted at $s$. 
	\item For each edge $(u, v) \in E(G)$, we say that $(u, v)$ violates the triangle inequality with respect to tree $T$ iff  
	\[
	\card{\dist_T(s, u) - \dist_T(s, v)} > w(u, v).
	\]
	If  $(u, v)$ violates the triangle inequality with respect to $T$, 
	double its importance value. 
	\item Repeat this sampling procedure $R = \Theta(\log^2(n)/\eps)$ times. Let $H^*$ be the union of all  the shortest path trees $T$ computed  in each of the $R$ sampling passes, and output a shortest path tree of $H^*$. 
\end{enumerate}
\end{Algorithm}

Our edge importance update rule increases the importance of every edge $(u, v) \in E(G)$ such that adding  $(u, v)$ to graph $H \cup F$ would decrease distances from $s$. This update rule is reminiscent of edge relaxation updates in the Bellman-Ford algorithm and serves a similar purpose. If an edge $(u, v)$ has its importance value increased frequently, then it will be sampled into set $F$ with relatively higher probability. In this way, edge importance values track which edges in $G$ are relevant for approximating $\{s\} \times V(G)$ distances.

We prove the correctness of our algorithm by arguing that if we choose the number of rounds $R$ to be sufficiently large, then for each edge $(u, v) \in E(G)$, almost every shortest path tree $T$ we computed satisfies the triangle inequality 
\begin{equation}
	\card{\dist_T(s, u) - \dist_T(s, v)} \le w(u, v).
	\label{eq:triangle-ineq}
\end{equation}
It is not hard to see that \emph{if} one of our trees $T$ satisfied  \Cref{eq:triangle-ineq} for every edge in $E(G)$ \textit{simultaneously}, then $T$ would be a shortest path tree of $s$ in $G$.  In general, this condition does not hold; however, we can ensure that after $R$ rounds,  \Cref{eq:triangle-ineq} is violated by at most a $(1+\eps)$-factor \textit{on average} for each edge $e \in E(G)$. We show that this is sufficient to obtain a $(1+\eps)$-approximate shortest  path tree of $s$ in $G$. 
Finally, to prove the earlier guarantee regarding~\Cref{eq:triangle-ineq} (on average), we use a potential function argument of a similar flavor as the standard MWU analysis of \cite{AHK12}, and specifically following the approach of~\cite{Assadi24}. 

The  $O(\log n)$-spanner $H$ of $G$ that we compute in our algorithm is essential for this correctness argument. Since each tree $T$ is a shortest path tree in $H \cup F$, it follows that  \Cref{eq:triangle-ineq} is violated by at most an $O(\log n)$-factor for any edge $(u, v) \in E(G)$. This in turn allows us to argue we only need $R \approx \log^2 n$ rounds for our algorithm to converge on a $(1+\eps)$-approximate shortest path tree (even when using  spanner $H$,  \Cref{eq:triangle-ineq} can be violated by an $\Omega(\log n)$-factor in general; this  causes our algorithm to suffer an additional logarithmic factor in the pass complexity  compared to \cite{Assadi24}).

\paragraph{Comparison with \cite{Assadi24}.} As we  mentioned earlier, our upper bound uses the  framework of \cite{Assadi24}, which obtained a  multi-pass streaming algorithm for a  $(1-\eps)$-approximation of  maximum matching by directly applying the multiplicative weight update method. The streaming algorithm of \cite{Assadi24} samples edges from the input graph; computes  a maximum matching of the sampled edges; increases the edge ``importance'' values of  edges that would have led to a larger maximum matching (via a primal-dual view, e.g., in bipartite graphs, by using edges of the original graph that violate a vertex cover of the sampled subgraph); and repeats this process until it obtains a $(1-\eps)$-approximate maximum matching. Adapting this strategy to single-source shortest paths requires several problem-specific modifications. 

The first and perhaps most important step is to come up with our notion of which edges to prioritize, which unlike the packing/covering framework of~\cite{Assadi24} (and~\cite{Quanrud24}) and its primal dual view, is based on the triangle inequality
and the ``Bellman-Ford view'' of the problem.\footnote{While it is possible to cast $s$-$t$ shortest path as a covering problem, its dual packing problem (as required by~\cite{Assadi24}) does not correspond to triangle inequality. On the other hand, 
one can see triangle-inequality as a dual for another LP relaxation of $s$-$t$ shortest path but this will no longer be a covering/packing approach. Finally, these approaches both seem to be restricted to $s$-$t$ shortest path 
unlike $s \times V(G)$ shortest paths we have for our solution.} We then need to establish our own ``sampling lemma'' (\Cref{lem:sampling-lemma}),  which shows that if we randomly sample a sufficiently large edge set $F \subseteq E(G)$, then not too many edges $e \in E(G)$ will violate the triangle inequality of \Cref{eq:triangle-ineq} with respect to the shortest path tree $T$ of $H \cup F$. Finally, edges $(u, v) \in E(G)$ can violate the triangle inequality with respect to tree $T$ (i.e., violate \Cref{eq:triangle-ineq}) by an arbitrarily large factor; as we mentioned earlier, we need to compute an $O(\log n)$-spanner $H$ in order to guarantee that \Cref{eq:triangle-ineq} is violated by at most an $O(\log n)$-factor. In contrast, this issue can be avoided in  \cite{Assadi24} using the duality of maximum matchings and minimum vertex covers in bipartite graphs (and the  corresponding duality in general graphs).

\paragraph{Comparison with  \cite{BeckerFKL21}, \cite{AssadiJJST22}, and \cite{Zuzic23}.} We also find it useful to compare our streaming algorithm to a line of work studying the transshipment problem in the streaming and related models. 
In the transshipment problem, we are given an uncapacitated graph $G$ and a feasible demand vector  over the vertex set of $G$, and we wish to find a minimum-cost flow that routes the demand. 
The work of \cite{BeckerFKL21} gave the first nontrivial $(1+\eps)$-approximation semi-streaming algorithm for transshipment using  $O(\eps^{-2} \cdot \polylog(n))$ passes.
Subsequently, a semi-streaming algorithm with $O(\eps^{-1} \cdot \polylog(n))$ pass complexity was given in \cite{AssadiJJST22}, improving the dependency on $\eps$. 
Transshipment is a generalization of the single-source shortest paths problem, and so the authors achieve similar bounds for $(1+\eps)$-approximate single-source shortest paths as well. One advantage of working with the transshipment problem is that $\alpha$-approximate solutions to transshipment can be boosted to $(1+\eps)$-approximations at the cost of a $\poly(\alpha)$-factor in the pass complexity. 
This approach has the advantage of extending to the distributed and parallel models (in certain cases; see also \cite{Zuzic23}). 
In contrast, in our streaming algorithm, we need to solve single-source shortest paths on the sampled subgraph exactly or at the very least, find a solution wherein at most $\eps$-fraction of sampled edges 
violate the triangle inequality. This is where we rely on the power of the streaming model that once we are working with sparse enough sample, 
we can perform ``heavy computations'' over them for free. 

However, directly solving single-source shortest paths has its own advantages. First, we obtain a much simpler algorithm. Approximately solving transshipment in \cite{BeckerFKL21} requires performing projected gradient descent on a modified version of transshipment, with many technical considerations. 
The subsequent work of \cite{Zuzic23} simplifies  \cite{BeckerFKL21} by solving only the dual of transshipment as a subroutine. Still, this  approach leads to new complications, like  difficulties obtaining a feasible primal solution, and an additional binary search to guess the optimal solution value. Likewise, the work of \cite{AssadiJJST22} obtains an approximation algorithm for transshipment by implementing a primal-dual algorithm inspired by \cite{Sherman} and gradient methods in the semi-streaming model, 
which leads to considerably complicated algorithms. 

Another advantage of solving 
 single-source shortest paths directly is that it leads to an improved pass complexity for our algorithm. Many of the technical details of \cite{BeckerFKL21}, \cite{AssadiJJST22}, and \cite{Zuzic23} discussed above introduce additional $\polylog(n)$ factors in  pass complexity that are  sometimes even hard to track. On top of this, a $(1+\eps)$-approximate solution to transshipment does not even directly yield a $(1+\eps)$-approximate shortest path tree. 
This obstacle can be overcome using ``an intricate
randomized rounding  scheme''\cite{BeckerFKL21}  that yields a  tree whose shortest paths are $(1+\eps)$-approximations on average. This in turn requires an additional logarithmic factor in passes to obtain a true approximate shortest path tree. In contrast, our streaming algorithm computes $(1+\eps)$-approximate $(s, t)$-shortest paths for all $t \in V(G)$ simultaneously, at no additional cost.

% !TEX root = main.tex

\newcommand{\NPC}{\ensuremath{\textnormal{PC}}}

\subsection{Lower Bound}\label{sec:overview-lower}
In this subsection, we give an overview of our lower bound of $\tilde{\Omega}(\log n)$ passes for any constant-factor approximation of shortest path estimation.\footnote{We emphasize that when we say $\tilde{\Omega}(f)$ or $\tilde{O}(f)$, we mean $\Omega(f/\poly \log f)$ and $O(f \cdot \poly \log f)$ respectively. Here,  $\tilde{\Omega}(\log n)$ is hiding a single $\log \log n$ factor.} Our work generalizes the result of \cite{GuruswamiO13}, which proves the same lower bound for finding the exact length of the shortest path. 

\subsubsection*{Prior Work by \cite{GuruswamiO13}}
We begin by briefly talking about the key components of the lower bound of \cite{GuruswamiO13} for testing if the shortest path between two given vertices $u, v$ is of length at most $k$ for some integer $k \geq 0$. 

\noindent
\textbf{Equality of Pointer Chasing.}
Let us start with a simple layered graph with $k+1$ layers and $n/(k+1)$  vertices in each layer.  The edges between each layer are chosen to be a matching.\footnote{Technically speaking, the construction of~\cite{GuruswamiO13} does not use matchings, but edges associated with a random function between the layers. We will talk about matchings here for  simplicity.} Let $u$ and $v$ be arbitrary vertices in the first layer and last layer of the graph, respectively. It is easy to see that vertex $u$ is connected to a unique vertex of the last layer, and there is a path from $u$ to $v$ if and only if $v$ is this unique vertex. 
 Refer to \Cref{fig:pc} for an illustration.
\begin{figure}[h!]
	\centering
	\begin{subfigure}[b]{0.45\textwidth}
		\centering
		\includegraphics[scale=0.20]{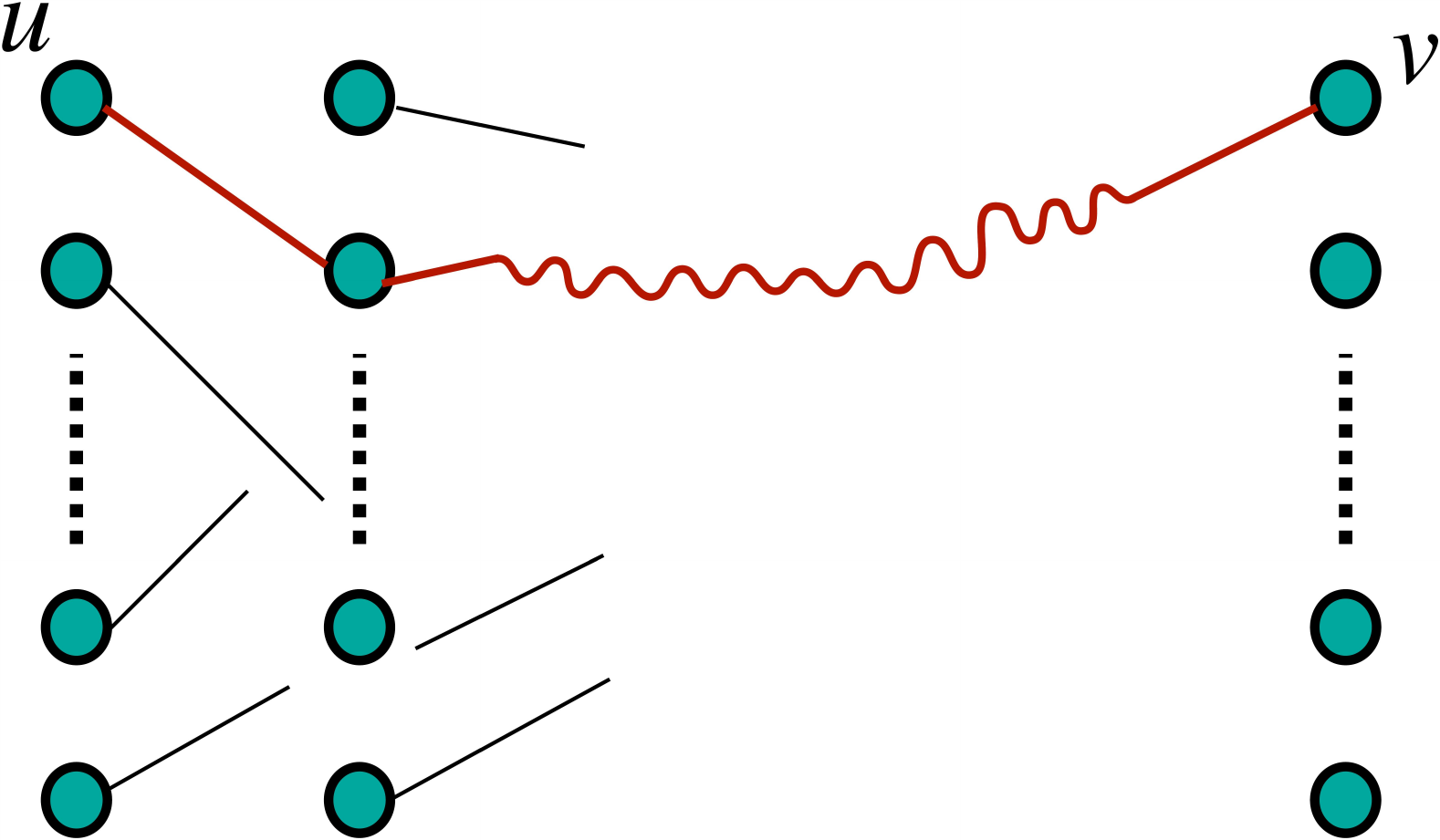}
		\caption{A graph where the shortest $(u, v)$-path is of length exactly $k$.}
	\end{subfigure}
	\hfill
	\begin{subfigure}[b]{0.45\textwidth}
		\centering
		\includegraphics[scale=0.20]{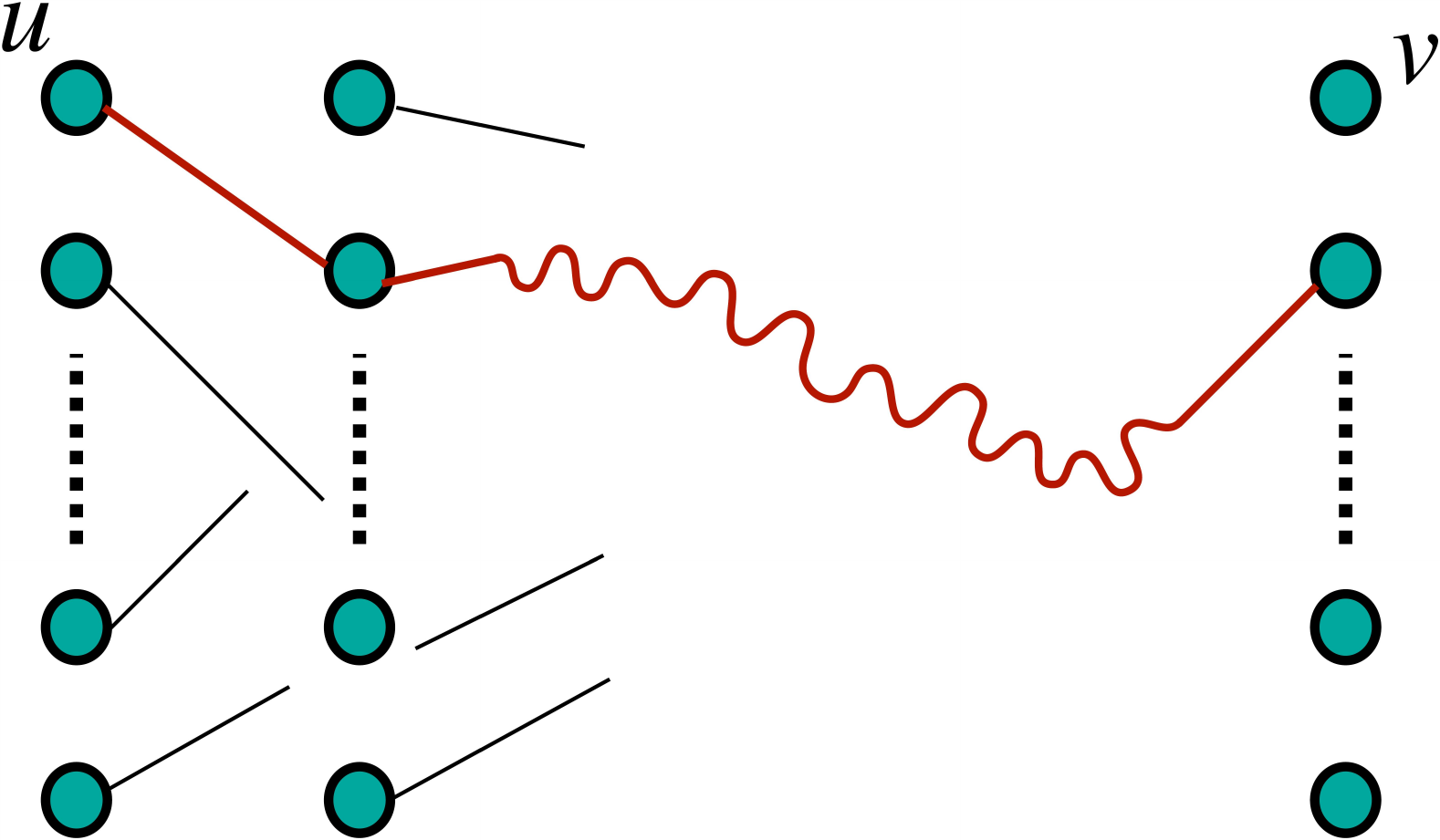}
		\caption{A graph where  $u$ and $v$ are not connected, as $v$ is not the vertex in the last layer connected to $u$.}
	\end{subfigure}
	\caption{Pointer Chasing  modeled as a layered graph. The path from vertex $u$ is shown using the wiggly line (brown).}
	\label{fig:pc}
\end{figure}

The search version of the problem, which requires outputting the unique vertex in the last layer connected to $u$, is closely related to the standard pointer chasing ($\NPC$) problem in communication complexity introduced by \cite{PapadimitriouS84}. 
This problem has since then been widely studied in various settings \cite{NisanW91,PonzioRV99,AssadiCK19b,Yehudayoff20,AssadiKSY20}. This search version has in fact been used previously in~\cite{FeigenbaumKMSZ05} to prove
semi-streaming lower bounds for \emph{finding} exact shortest paths.

The work of \cite{GuruswamiO13}, on the other hand, looks at the \emph{decision} version of the problem, where there are two $\NPC$ instances.
There are a total of $2k$ players, with $k$ players from each instance. 
The first instance has $k$ players $P_1, P_2, \ldots, P_k$, and the $k$ different matchings are distributed among them. The second instance has players $P_{k+1}, P_{k+2}, \ldots P_{2k}$, and these players have the $k$ matchings from the second instance. In each round, the players speak in the order of $P_1, P_2, \ldots$, until the last player $P_{2k}$. The final player outputs $1$ if the final vertex of both $\PC$ instances is the same and outputs $0$ otherwise. 
We refer to this problem as the Equality of Pointer Chasing; see  \Cref{fig:eq-pc} for an illustration. 

\begin{figure}[h!]
	\centering
	\includegraphics[scale=0.25]{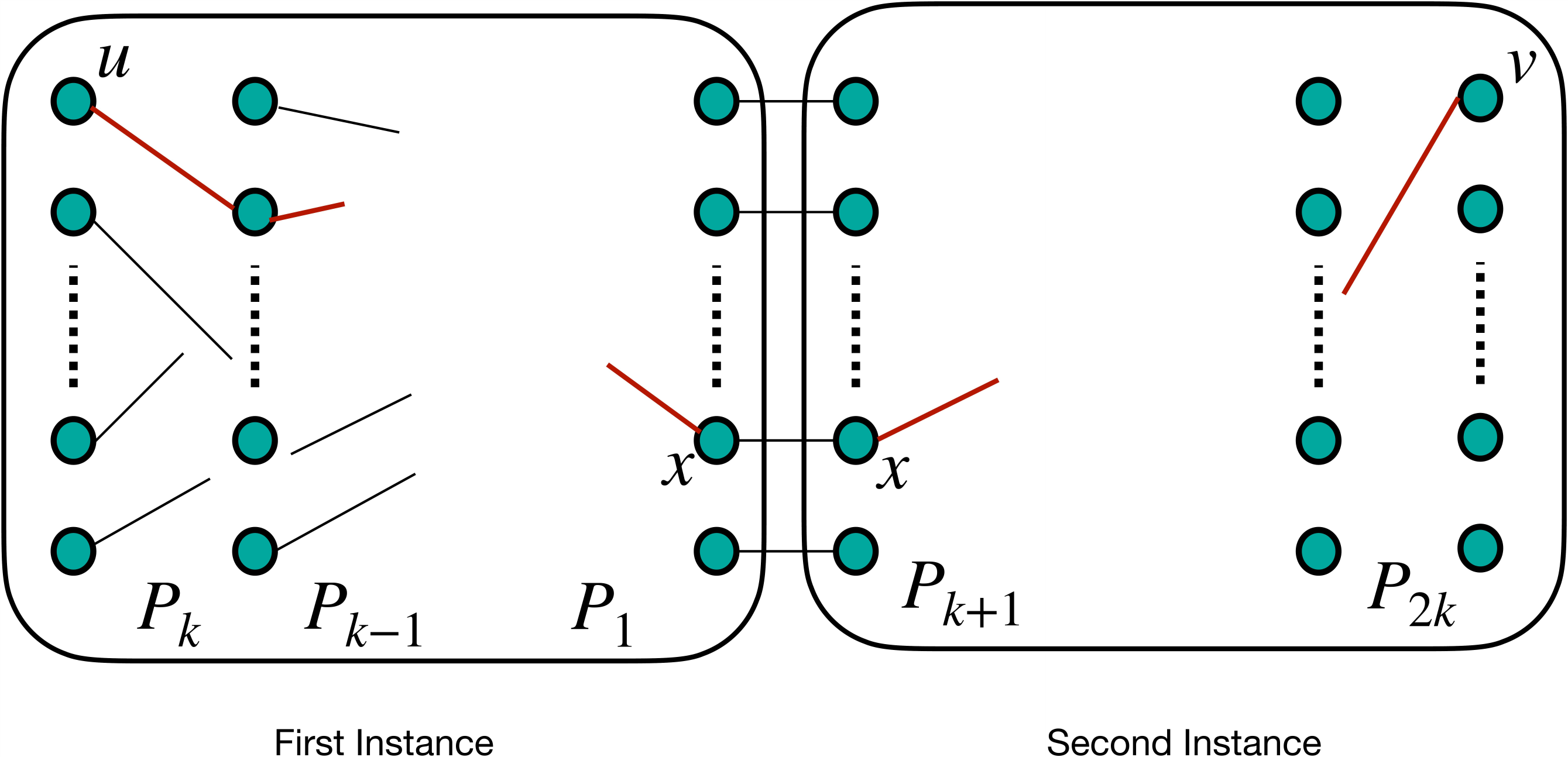}
	\caption{Equality of Pointer Chasing. When the final vertex is $x$ for both the instances, there is an $(s, t)$-path of length $2k+1$. Otherwise, the two vertices are not connected. The players are $P_1, P_2 \ldots, P_{2k}$, and which matching is given to each player is specified below the matching. }\label{fig:eq-pc}
\end{figure}

The authors in \cite{GuruswamiO13} prove a lower bound of $\Omega(n/\poly(k))$ on the communication complexity of Equality of Pointer Chasing when the number of rounds is at most $k-1$.\footnote{Observe that there is a protocol that uses $O(\log n)$ communication by simply sending the next vertex on the paths from $s$ and $t$ if there are $k$ rounds.} In fact, they prove this lower bound on the information complexity\cite{CoverT06} (see \Cref{app:info} for formal definitions) of the problem over the uniform distribution of matchings, and this in turn gives a communication lower bound by the well-known connection between communication complexity and information complexity (see \Cref{prop:ic-cc}). 

\noindent
\textbf{$\OR$ of  Equality of Pointer Chasing.} 
The current graph in \Cref{fig:eq-pc}, although requiring $\Omega(n/\poly(k))$ communication with limited number of rounds, is easy to solve in the semi-streaming model with $\tilde{O}(n)$ space because it only has $O(n)$ edges. 
The work of \cite{GuruswamiO13} gets around this issue by embedding many equality of $\NPC$ instances in the same graph. They then define a new problem, called $\OR$ of Equality of Pointer Chasing, where the players must output $1$ if the solution of at least one of the equality of $\NPC$ instances is $1$, and output $0$ otherwise. 
%To guarantee that there is no interference between the different equality of $\NPC$ instances, the work of  \cite{GuruswamiO13} randomly permutes all the middle layers of the graph, while guaranteeing that the solution of each $\NPC$ instance is unchanged.
To guarantee that there is no interference between the different equality of $\NPC$ instances, the work of \cite{GuruswamiO13} randomly  scrambles the connections between layers in each $\NPC$ instance, while ensuring that the solution of each $\NPC$ instance is unchanged.

More concretely, to prove a lower bound against $p$-pass streaming algorithms, first they take $t := n^{\Theta(1/p)}$ many $\NPC$ instances, then they scramble each $\NPC$ instance by randomly permuting the vertices in the middle layers, and finally they take the union of these graphs on the same vertex set. Refer to \Cref{fig:or-eq-pc} for an illustration.

\begin{figure}[h!]
	\centering
	\includegraphics[scale=0.25]{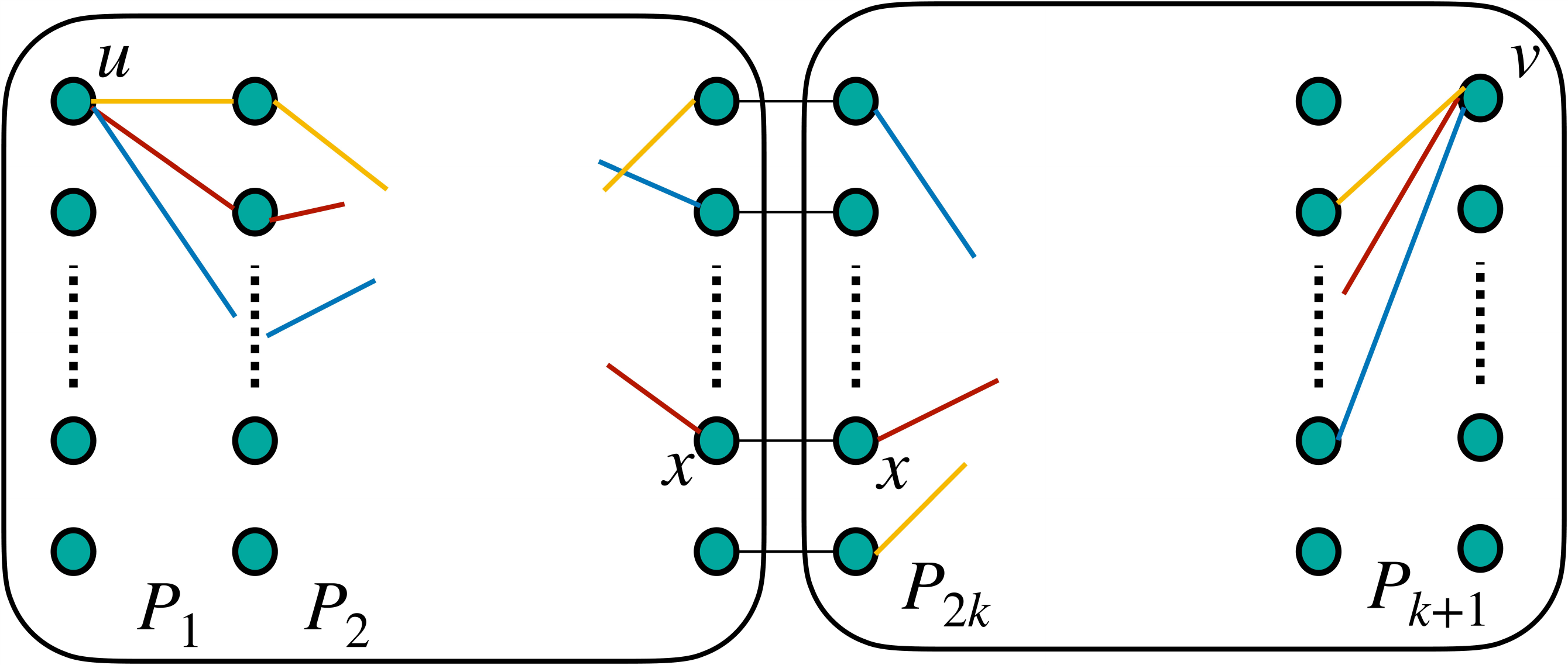}
	\caption{Multiple equality of $\NPC$ instances embedded in the same graph. Each color corresponds to one instance. There may be paths of multi-colored edges from $u$ and $v$ to the same vertex in the middle layer, so that the shortest path length between $u,v$ is $2k+1$, but the parameters are chosen so that this is unlikely.}\label{fig:or-eq-pc}
\end{figure}

It is easy to see that if the solution to the $\OR$ of equality of $\NPC$ instance is $1$, then the distance from $u$ to $v$ in the graph is $2k+1$. When the solution is $0$, none of the $t$ input instances have equal pointers, and it is not clear immediately what we can conclude about the distance between $u$ and $v$. However, here we can use the fact that \cite{GuruswamiO13} randomly permutes the vertices in the middle layers of each $\NPC$ instance before unioning them. As these permutations are sampled independently of each other across different equality of $\NPC$ instances, roughly $t^{2p+1}$ vertices in the last layer become reachable from vertex $u$ through paths from multiple $\NPC$ instances. As long as $t^{2p+1} \ll (n/k)$, the shortest path distance will be at least $2k+2$ with high probability. Hence, a streaming algorithm for \textit{exact} shortest paths can effectively distinguish between these two cases.

To conclude the lower bound, \cite{GuruswamiO13} uses ad hoc direct sum techniques from \cite{Bar-YossefJKS02} and the independence properties of the uniform distribution to say that the $\OR$ of $t$-many instances of equality of $\NPC$ requires $\Omega(t \cdot n/\poly(k))$ communication. By setting $t  = n^{\Theta(1/p)}$ and $k  = p+1$,  they get a lower bound of $n^{1+\Omega(1/p)}/\poly(p, k)$ on the space required for $p$-pass algorithms to decide  whether two given vertices $u, v$ are at most $2k+1$ distance apart. 

\subsubsection*{Our Work}
We now discuss our result and the changes from \cite{GuruswamiO13} that give us the lower bound.

\noindent
\textbf{Challenge in Extending to Constant Approximation.}
In the reduction from exact shortest paths to the $\OR$ of Equality of $\NPC$ problem, \cite{GuruswamiO13} crucially uses the fact that as long as the shortest path distance between $u$ and $v$ is more than $2k+1$, the players can output that the answer to $\OR$ of Equality of $\NPC$ is $0$. 
The interference between many equality of $\NPC$ instances happens with a probability of $o(1)$ if the parameter $t$ is chosen appropriately.

However, this argument does immediately give a reduction from approximate shortest paths to the $\OR$ of Equality of $\NPC$ problem. Constant approximation algorithms cannot distinguish between the case where the distance from $u$ to $v$ is exactly $2k+1$ (in which case the players should output $1$) and the distance from  $u$ to $v$ is greater than $2k+1$ but still $O(k)$ (in which case the players should output $0$). Specifically, it is possible there is an $(u, v)$-path of length, say, $6k$ in the graph, via a path from $u \rightsquigarrow x \rightsquigarrow y \rightsquigarrow v$ with $u \rightsquigarrow x$, and $x \rightsquigarrow y$ and $y \rightsquigarrow v$ each coming from different equality of $\NPC$ instances (see \Cref{fig:overlap-bad}). Moreover, this issue cannot be avoided by randomly scrambling the individual $\NPC$ instances. 
%Therefore, there is no hope for an $\alpha$-approximation algorithm of shortest path to be able to distinguish between 
Therefore, we cannot hope to use a constant approximation shortest paths algorithm to solve the exact $\OR$ of equality of $\NPC$ instances.  
There is some hope, however, to solve the problem distributionally for some apt distribution and prove a lower bound using Yao's minimax principle.

\begin{wrapfigure}{r}{0.5\textwidth}
	\centering
	\includegraphics[scale=0.25]{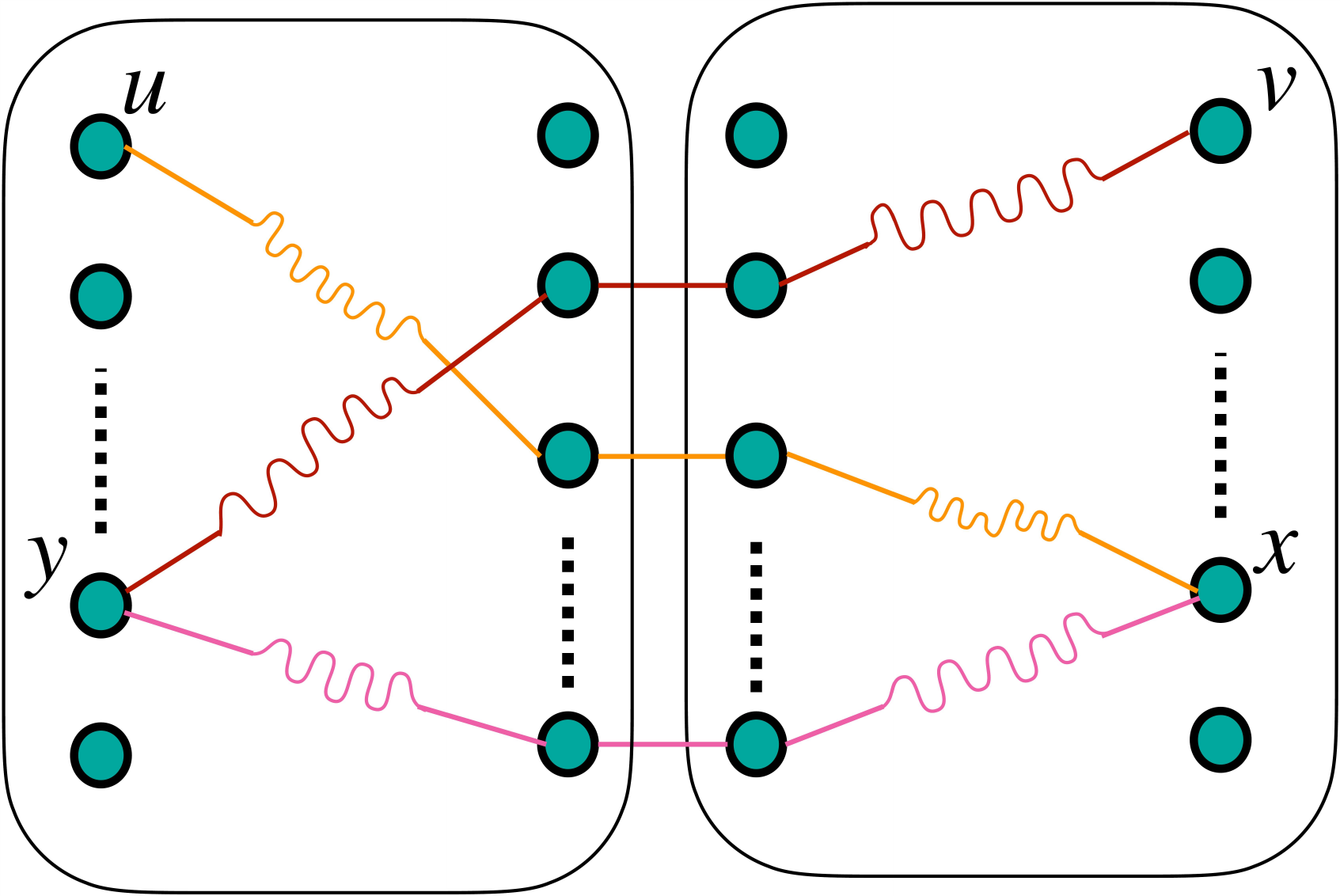}
	\caption{Three paths from three equality of $\NPC$ instances, with different colors used to represent edges from different instances.  Together, these three paths form an $(u, v)$-path of length $O(k)$. \vspace{-12mm} }\label{fig:overlap-bad} 
\end{wrapfigure}

% Whatever threshold we use on $\dtilde$ to solve $\OR$ of equality of $\PC$,  the algorithm cannot differentiate between a shortest path length of $(2k+1)$ (when the answer is one) or a distance larger than $(2k+1)$ but at most $\alpha \cdot (2k+1)$, when the $\OR$ is zero. 

The underlying hard distribution from \cite{GuruswamiO13} is simply the uniform distribution for all matchings.\footnote{This choice is necessary for their communication lower bounds as they heavily exploit the independence
between the input of players, both for proving the communication lower bound for the equality of pointer chasing, and then also for extending it to multiple instances via a direct sum argument.} Let us argue why this distribution is insufficient for our purpose. 
Here, we have:
\begin{itemize}
	\item When the $\OR$ is one, which happens with $\approx t/n$ probability, the streaming algorithm can always give the correct answer because the shortest path distance will be $2k+1$. 
	\item When the $\OR $ is zero, which happens with $\approx (1-t/n)$ probability, the streaming algorithm succeeds whenever the instances do not interfere with each other. The instances interfere with a probability of at least $\Omega(t^{2p} \cdot k/n)$. 
\end{itemize}

The error probability when using the approximate shortest path distance is at least $\Omega(t^{2p} \cdot k/n)$, and this error probability carries over to one instance of equality of $\NPC$ also with standard direct sum techniques. But, solving one instance of equality of $\NPC$ with error probability $> k/n$ on the uniform distribution is simple: we just have to output zero all the time. 
Instead, we need a different distribution so that the streaming algorithm has a substantial advantage with \emph{even an approximate} shortest path distance. Moreover, this distribution should remain hard to solve. 

We now describe how we achieve this by introducing some correlation in the distribution.

\noindent
\textbf{Correlated Input Distribution.}
We create a distributional communication problem called \textbf{paired pointer chasing}, closely related to the $\OR$ of equality of $\NPC$ problem. We define this problem informally in the overview; for a precise definition, see \Cref{def:paired-pc}.

There is an answer bit $\rb \in \{0,1\}$ sampled uniformly at random that the players have to find. The input given to them is sampled from a distribution $\mu_{\rb}$:
\begin{itemize}
	\item Distribution $\mu_0$: all the matchings in both instances of $\NPC$ are chosen uniformly at random. 
	\item Distribution $\mu_1$: all the matchings in both instances of $\NPC$ are chosen uniformly at random, \textbf{conditioned on the pointers being equal}. 
\end{itemize}

The two $\NPC$ instances are possibly paired together based on the value of $\rb$. There are only two players, Alice and Bob, who receive alternate matchings from both the instances. They have to output the value of $\rb$ at the end of all messages, with Bob speaking first. Refer to  \Cref{fig:ppc} for an illustration.

\begin{figure}[h!]
	\centering
	\begin{subfigure}[b]{0.45\textwidth}
		\centering
		\includegraphics[width=\linewidth]{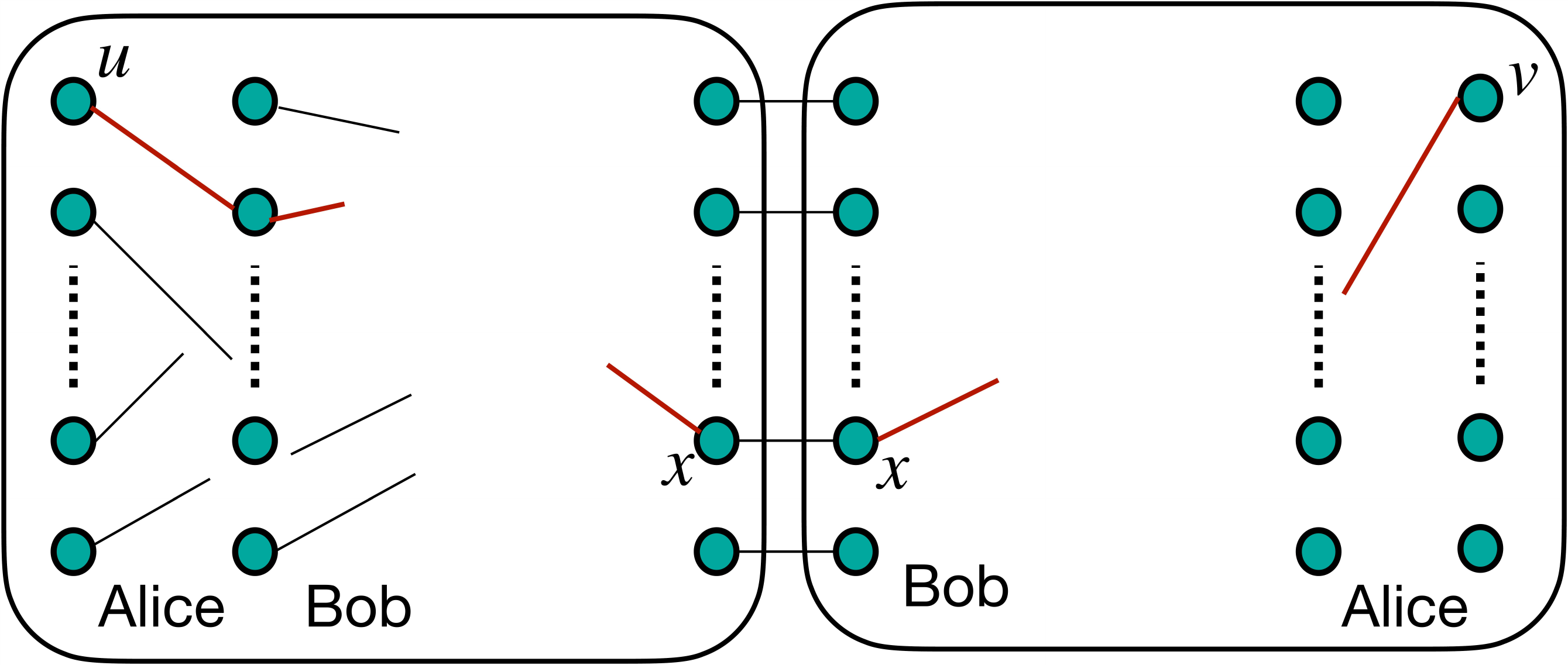}
		\caption{A paired pointer chasing instance sampled when $\rb = 1$. The two pointers are equal to some vertex $x$ in the middle layer. }
	\end{subfigure}
	\hfill
	\begin{subfigure}[b]{0.45\textwidth}
		\centering
		\includegraphics[width=\linewidth]{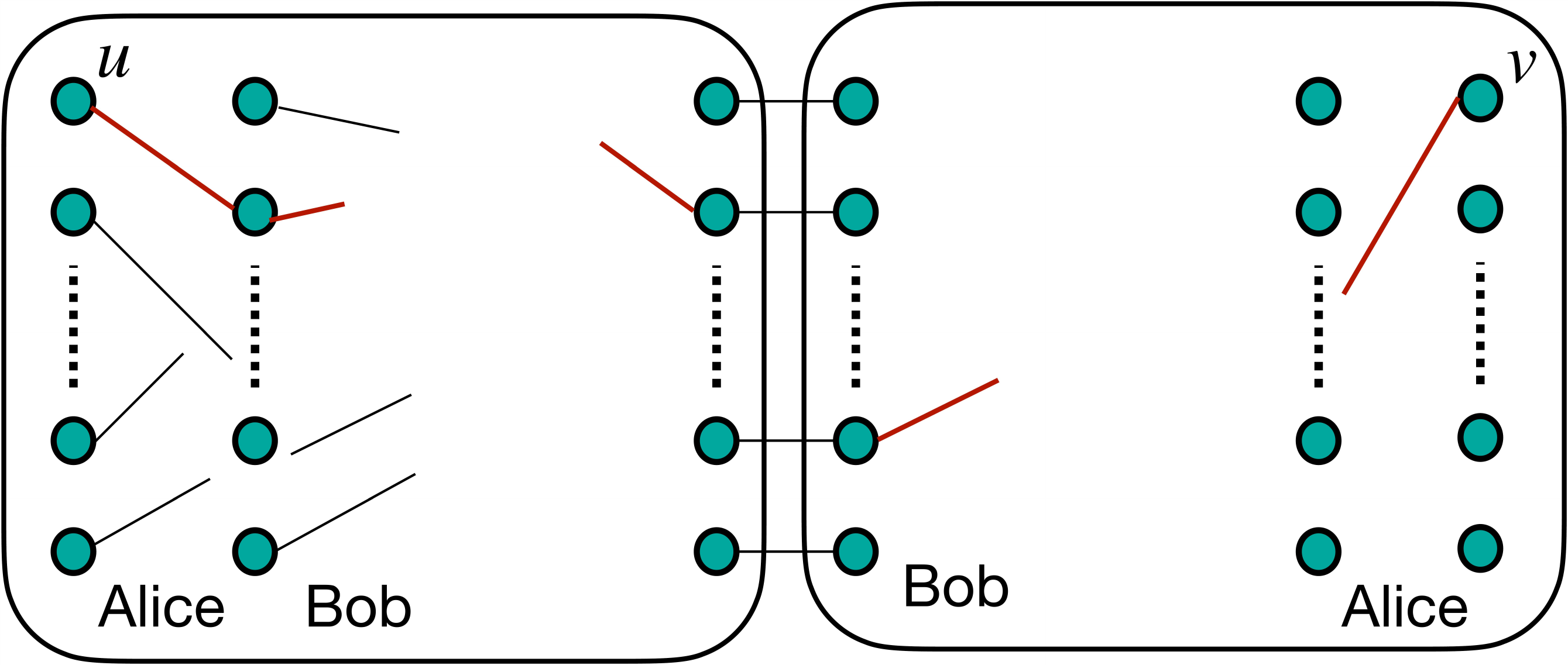}
		\caption{A paired pointer chasing instance sampled when $\rb = 0$. All the matchings are chosen at random. The two pointers may or may not be equal.}
	\end{subfigure}
	\caption{Paired pointer chasing instances. The player whom the matching is given to is specified below the edges. \textbf{A reminder that here, Bob speaks first.}}
	\label{fig:ppc}
\end{figure}

Our distribution can be viewed as artificially boosting the probability that the two $\NPC$ instances are equal so that the answer $\rb$ is uniformly random. 
The approximate shortest path distance still has a non-trivial advantage in guessing the value of $\rb$, and trivially outputting either $0$ or $1$ for $\rb$ can only succeed with probability half. We prove that outputting the value of $\rb$ with a high enough constant advantage requires $\Omega(n / \poly(k))$ communication, which is the main technical part of our lower bound. The details will be covered later in the overview.

\textbf{$\OR $ of Paired Pointer Chasing.}
The final hard instance is created by $\OR$ of $t$ many instances of paired pointer chasing instead. The interference probability when all instances are sampled from $\mu_0$ is $O(t^{2p+1} \cdot k/n) = o(1)$, and the streaming algorithm will have $\Omega(1)$ advantage in finding the value of $\rb$ based on the approximate shortest path distance. Our lower bound will be $\Omega(t \cdot n/\poly(k))$ for the $\OR $ of paired pointer chasing. As the input distribution for our problem is correlated, this part is not a straightforward application of techniques. Our process has roughly three steps:
\begin{enumerate}[label=$(\roman*)$]
	\item A direct sum part, converting a protocol $\protOR$ for the $\OR$ problem into a protocol $\prot$ for one instance of paired pointer chasing. The information cost when the instance is sampled from $\mu_0$ is $t$ times smaller than the communication cost of $\prot$. This uses standard direct sum techniques, e.g., as in~\cite{BarakBCR10}.
	\item A message compression part, converting $\prot$ into a protocol with low communication in expectation when the input distribution is $\mu_0$. This part uses the external information compression arguments of \cite{HarshaJMR07}, more specifically, the presentation from \cite{RaoY20}.
	\item A final protocol that either has a low communication over both of $\mu_1$ and $\mu_0$ (recall that part $(ii)$ was only for $\mu_0$) or can distinguish between the two distributions
	using a straightforward ``communication odometer'' (see also the (much) stronger notion of an \emph{information odometer}~\cite{BravermanW15} and odometer arguments, e.g., in~\cite{GoosJP017,Assadi17sc}). 
\end{enumerate}

With $t$ set as $n^{\Theta(1/(\alpha \cdot p))}$, we get a lower bound of $n^{1+\Omega(1/(\alpha \cdot p))}$ on the space used by $p$-pass streaming algorithms for $\alpha$-approximation of estimating shortest path distance.

\noindent
\textbf{Lower Bound Paired Pointer Chasing.}
We briefly outline how we get the lower bound of $\Omega(n/\poly(k))$ for paired pointer chasing, which is the main technical contribution of our lower bound. 
Firstly, we can cast our problem as an $\NPC$ problem on $n^2/(k+1)^2$ vertices in each layer, simply by taking the cartesian product of the corresponding vertex sets in each layer. 
An edge exists between $(w_1, w_2)$ and $(w_1', w_2')$ in the next layer if and only if edge $(w_1, w_1')$ exists in the first instance, and $(w_2, w_2')$ exists in the second instance. 
See \Cref{fig:cart-prod} for an illustration.

\begin{figure}[h!]
	\centering
	\begin{subfigure}[b]{0.45\textwidth}
		\centering
		\includegraphics[scale=0.30]{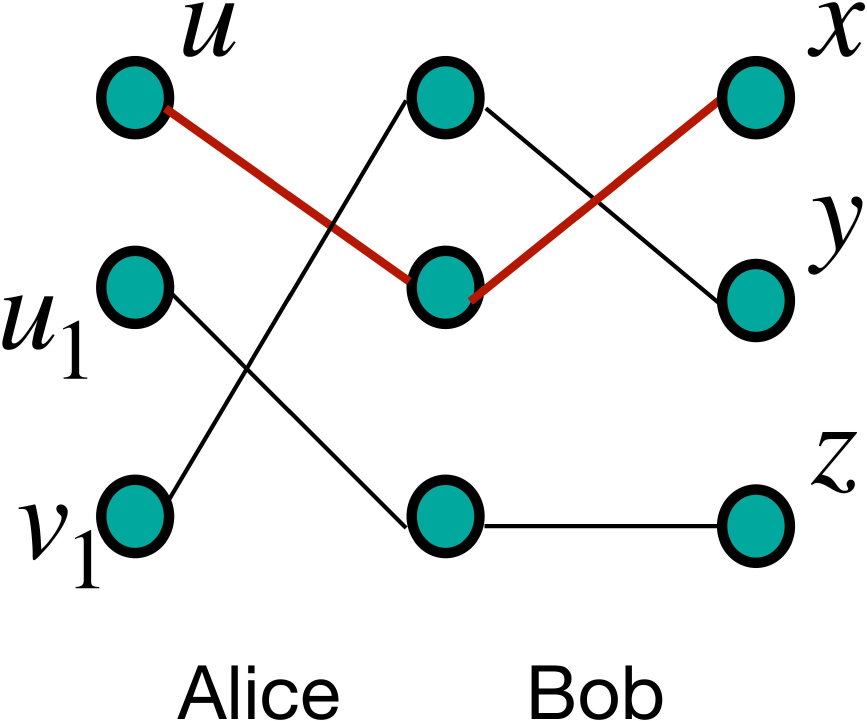}
		\caption{One instance of $\NPC$ from paired pointer chasing.}
	\end{subfigure}
	\hfill
	\begin{subfigure}[b]{0.45\textwidth}
		\centering
		\includegraphics[scale=0.30]{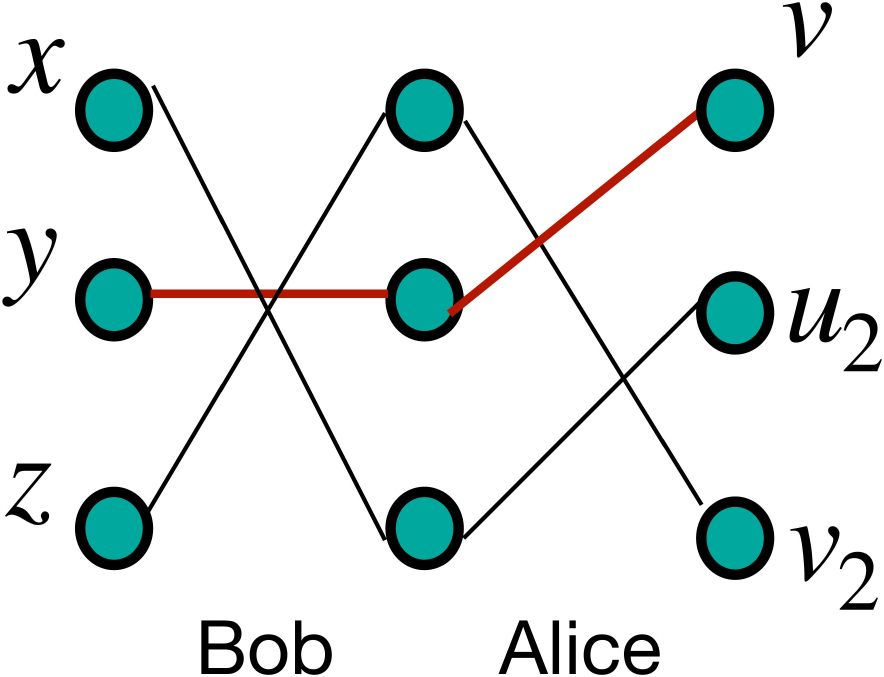}
		\caption{Other instance of $\NPC$ from paired pointer chasing.}
	\end{subfigure}
	
	\begin{subfigure}[b]{\textwidth}
		\centering
		\includegraphics[scale=0.25]{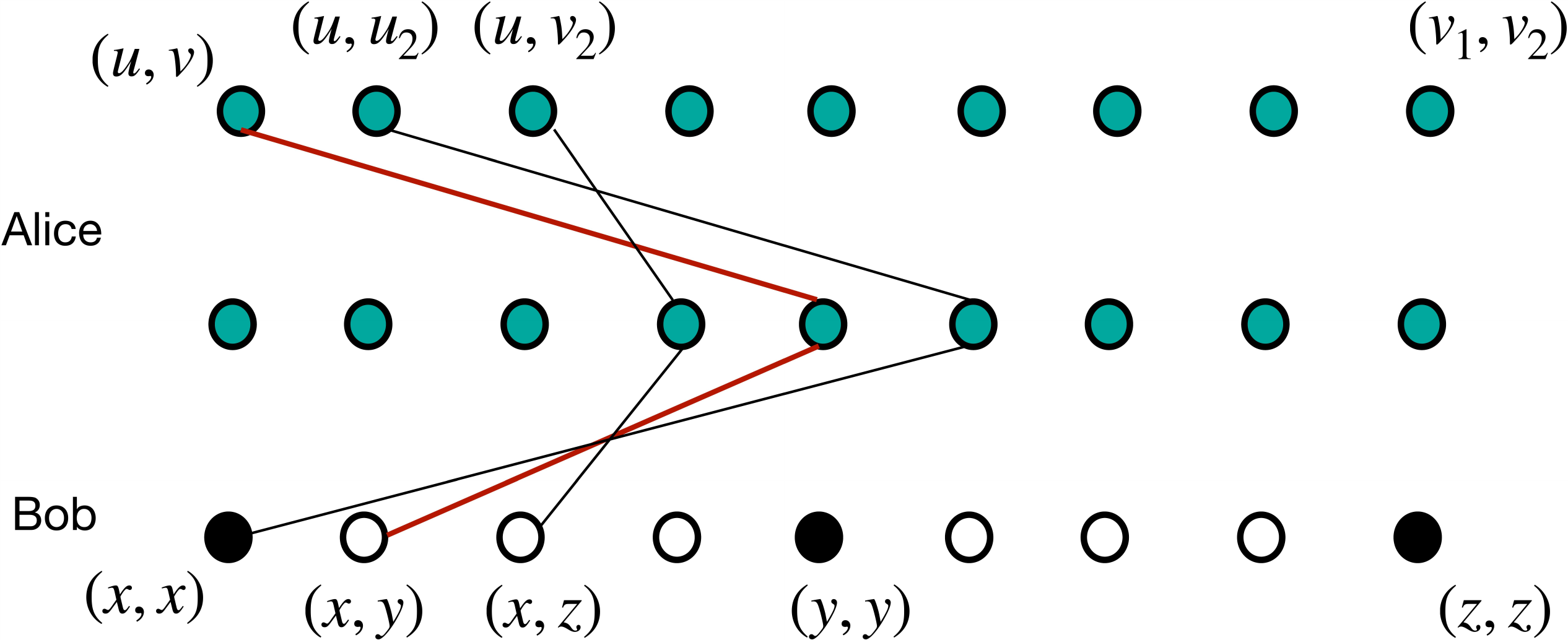}
		\caption{Product of the two instances. Pink vertices in the last layer highlight the end points when the two instances are equal. Some edges and labels of vertices are omitted.}\label{fig:cart-prod-3}
	\end{subfigure}
	\caption{Paired pointer chasing instances with $n = 9, k = 2$. 
		All four matchings are chosen uniformly at random when $\rb = 0$, and all four are chosen uniformly at random conditioned on the vertex reachable from $(u, v)$ being blank when $\rb = 1$.}
	\label{fig:cart-prod}
\end{figure}
%\janani{Figure out changes in the from green pink to blank and black}

There are $n/(k+1)$ vertices in the last layer in the new product graph of the form $(x,x)$ for some vertex $x$; these vertices correspond to the case when the two pointers are equal. 
We can reduce the problem to finding whether the vertex $(s,t)$ in the first layer has a path to one of these $n/(k+1)$ different vertices in the last layer. If the distribution of the edges in the product graph were uniform as well, we would be done now through existing techniques for analyzing sparse pointer chasing or sparse index problems \cite{AssadiKL16,GhaziKS16,Saglam19}.  Alas, this is not true as the $n^2/(k+1)^2$ edges between each layer originate only from two matchings of size $n/(k+1)$.

We carefully choose a subset of vertices in the last layer, termed a \textbf{rook set} (see \Cref{def:rook-set}), so that, conditioned on the final pointer of $(u, v)$ being inside a rook set, the matchings are indeed uniformly random. 
Let us call the set of vertices in the last layer that correspond to the instances having the same pointer as \textbf{one-vertices}, i.e., vertices of the form $(w,w)$. The rest of the vertices are called \textbf{zero-vertices}. In the last layer of \Cref{fig:cart-prod-3}, one-vertices are black, and others are blank.

\begin{wrapfigure}{r}{0.4\textwidth}
	\centering
	\includegraphics[scale=0.25]{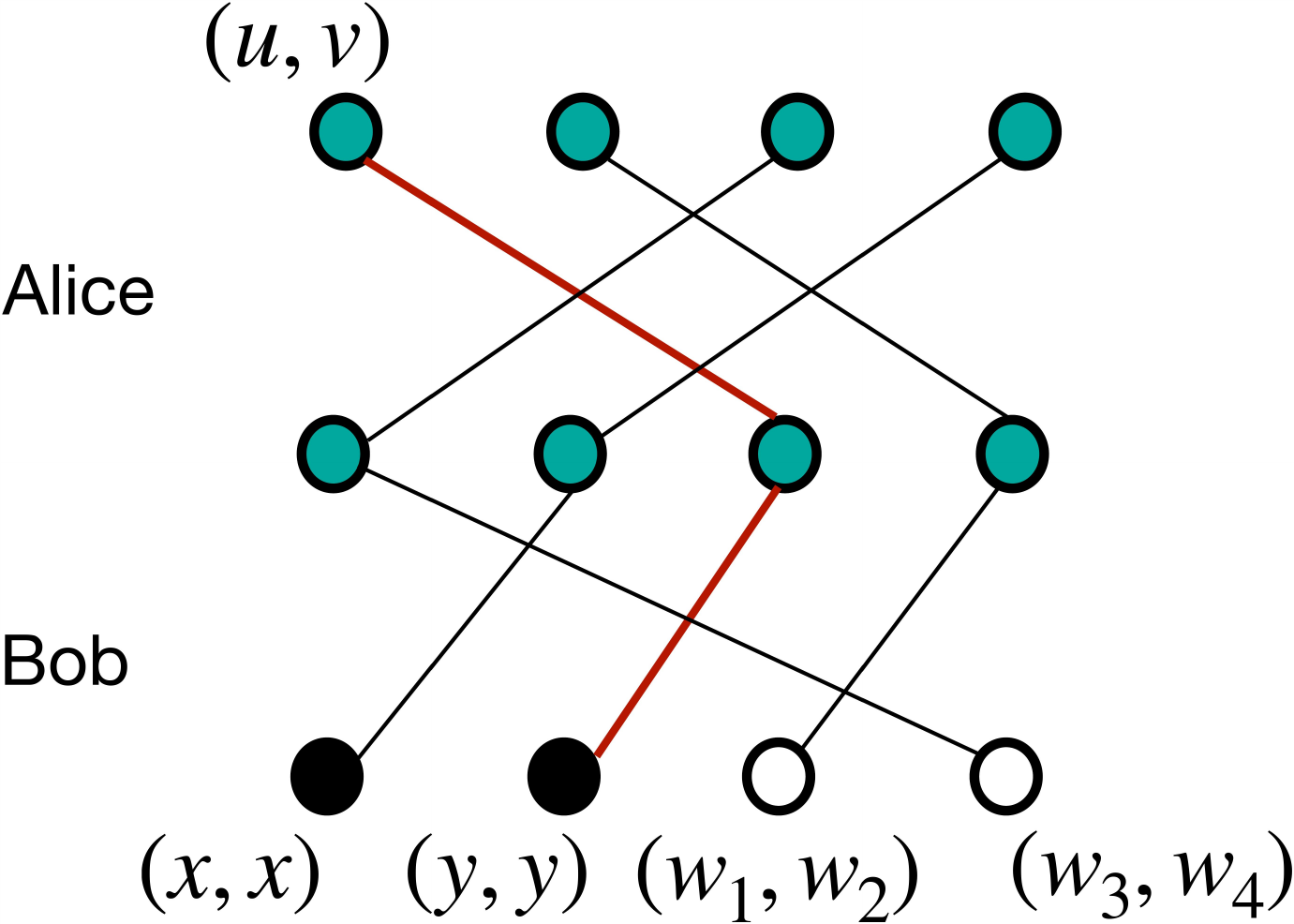}
	\caption{A rook set of size four. The vertex from $(u, v)$ in the last layer is a one vertex. The labels of $x, y, w_1, w_3$ are distinct. The labels of $x, y, w_2$, and $w_4$ are also distinct. Other vertices and edges in the graph are omitted. \vspace{-25mm}}\label{fig:rook-set}
\end{wrapfigure}

A rook set has the following properties:
\begin{enumerate}[label=$(\roman*)$]
	\item Size of $\Omega(n/k)$, ensuring that the final pointer from $(u, v)$ lands in the rook set with $\Omega(1)$ probability. 
	
	\item Equal number of one-vertices and zero-vertices, so that conditioned on the pointer from $(u, v)$ landing in the set,  the value of $\rb$ is still almost uniform. 
	\item For each vertex $x$ in the last layer of the first $\NPC$ instance, at most one vertex $(x,w)$ is present in the rook set. Similarly, for each vertex $y$ in the last layer of the second $\NPC$ instance, at most one $(w,y)$ is present in the set. This is to ensure that the edges do not interfere with each other.  In fact, this property gives us the name rook set as well, see\\ \Cref{foot:rook-set} for more details.
\end{enumerate}

\vspace{2mm}

We show that such a set exists with some basic combinatorial tools. Refer to \Cref{fig:rook-set} for a simple illustration. The final piece of the proof is to show that $\NPC$ requires $\Omega(n/\poly(k))$ communication when the matchings are chosen uniformly at random. This is given in \Cref{lem:unif-pc} and \Cref{app:unif-pc}, and the proof follows almost directly from \cite{AssadiN21}.\footnote{The only difference is that \cite{AssadiN21} proves the lower bound for multiple players, whereas we only have two players. }

\subsubsection*{Comparison with \cite{GuruswamiO13}.}
We conclude our overview by giving a schematic comparison of our work with \cite{GuruswamiO13}, which also serves as a summary of the main steps. 
\renewcommand{\arraystretch}{1.4}
\begin{center}
	\begin{tabular}{|p{7cm}|p{7cm}|}
		\hline
		\textbf{Work of \cite{GuruswamiO13}} & \textbf{Our Work} \\
		\hline
		Reduction from exact shortest path distance to $\OR$ of equality of $\NPC$&   Similar reduction from estimation of shortest path distance to distributional $\OR$ of paired pointer chasing\\
		\hline
		Ad-hoc techniques with direct sum and message compression to reduce $\OR$ to one equality of $\NPC$ with a focus on the uniform distribution & Direct sum and message compression to reduce $\OR$ to one paired pointer chasing on a non-product distribution\\
		\hline
		Lower bound on the uniform product distribution in the equality of $\NPC$  & Lower bound on a correlated distribution for paired pointer chasing \\
		\hline
		$\NPC$ input split between many players & Only two players, Alice and Bob  \\
		\hline
	\end{tabular}
\end{center}

\clearpage
	
	\section{Preliminaries}\label{sec:prelim}

\paragraph{Notation.}
We use $[n]$ denote the set of numbers $\{1, \dots, n\}$. We use \textsf{sans-serif} font to denote random variables throughout this paper. 
 We use $\log(\cdot)$ to denote the  binary logarithm and $\ln(\cdot)$ to denote the natural logarithm. We use $\exp(\cdot)$ to denote the exponential function,  $\exp(x) = e^x$.

Given two graphs $G_1 = (V, E_1)$ and $G_2 = (V, E_2)$ over the same vertex set $V$, we use $G_1 \cup G_2$ to denote the edge union of graphs $G_1$ and $G_2$, i.e., $G_1 \cup G_2 = (V, E_1, \cup E_2)$. Likewise, we define $G_1 \cap G_2$ as $(V, E_1 \cap E_2)$.  Given a (possibly weighted) graph $G$ containing vertices $u$ and $v$, we use $\dist_G(u, v)$ to denote the weighted shortest path distance between $u$ and $v$ in $G$.

Unless stated otherwise, we measure  space complexity in our algorithms and lower bounds in terms of machine words each having $O(\log n)$ bits (for our algorithms on weighted graphs, we 
assume the weights can fit a single word, either by limiting weights to be $\poly(n)$-bounded or letting words be larger to fit a single edge weight). 
We say that an event holds with \emph{exponentially high probability} if it holds with probability at least $1-\exp(-cn)$ for some constant $c > 0$.

\subsection*{Graph Primitives}

We will need a standard theorem from graph theory that counts the number of distinct spanning trees of a complete graph with labeled vertices. 

\begin{fact}[Cayley's Tree Formula \cite{borchardt1861interpolationsformel}] \label{fact:treeform}
	The number of spanning trees of a complete graph with labeled vertices is $n^{n-2}$. 
\end{fact}

Additionally, our upper bounds will use graph spanners, which are sparse subgraphs that approximately preserve distances in the original graph. More formally, given a (possibly weighted) graph $G = (V, E)$ and an integer $k \ge 1$, a subgraph $H \subseteq G$ is a $k$-spanner of $G$ if $\dist_H(s, t) \le k \cdot \dist_G(s, t)$ for all $s, t \in V$. We will rely on existing space-efficient  streaming algorithms for computing high-quality graph spanners in a single pass.

\begin{proposition}[\!\!\!\cite{FeigenbaumKMSZ05,jin2024streaming}] \label{prop:spanner}
	For any $k \ge 1$, there is a one-pass streaming algorithm for computing a $(2k-1+\eps)$-spanner of size $O(\eps^{-1}n^{1+1/k}\log n)$ of a weighted graph using  $O(\eps^{-1}n^{1+1/k}\log n)$ words. 
\end{proposition}

\subsection*{Model of Communication}
We work with the standard two party model of communication with players Alice and Bob. 
There is a function $f:\cX \times \cY \rightarrow \{0,1\}$ that Alice and Bob want to compute. Alice is give $X \in \cX$ and Bob is given $Y \in \cY$. Alice and Bob have a shared tape of random bits they can access called public randomness. They also have their own private sources of randomness.

The communication happens in rounds, and there is one message sent by Alice or Bob to the other player in each round. The players alternate who sends the message over the rounds, and \textbf{Bob speaks first}. More specifically, Bob sends a message in the first round, and in all subsequent odd numbered rounds. Alice sends a message in all even numbered rounds. At the end of all communication, the player who receives the last message is required to output $f(X, Y)$.

A protocol $\prot$ is said to succeed with probability $p$ if, for all $(X,Y) \in \cX \times \cY$, the output of $\prot$ is the same as $f(X, Y)$ with probability at least $p$. The communication cost of protocol $\prot$ is the maximum length of any message sent over all the rounds in the protocol. 

We use information-theoretic techniques extensively in the proof of our lower bound. Basic definitions and standard results from information theory can be found in \Cref{app:info}.

%\subsection{Concentration Inequalities}
%
%We use the standard form of Chernoff bound and its extension to negatively correlated variables in~\cite{PanconesiS97}. 
%
%\begin{proposition}[Chernoff Bound; cf.~\cite{DubhashiP09}]\label{prop:chernoff}
%	Let $X_1,\ldots,X_n$ be $n$ independent random variables in $[0,1]$ and $X := \sum_{i=1}^{n} X_i$. For any $\delta > 0$ and $\mu_{min} \leq \expect{X} \leq \mu_{max}$, 
%	\begin{align*}
%		&\Pr\paren{X \geq (1+\delta) \cdot \mu_{max}} \leq \exp\paren{-\frac{\delta^2 \cdot \mu_{max}}{2+\delta}} \quad , \quad \Pr\paren{X \leq (1-\delta) \cdot \mu_{min}} \leq \exp\paren{-\frac{\delta^2 \cdot \mu_{min}}{2}}.
%	\end{align*} 	
%	Moreover, the upper tail bound continues to hold as long $X_1,\ldots,X_n$ are negatively correlated, i.e., for every $S \subseteq [n]$, 
%	\[
%	\expect{\prod_{i \in S} X_i} \leq \prod_{i \in S} \expect{X_i}. 
%	\]
%\end{proposition}

\clearpage
	
%	\clearpage
	
	% !TEX root = main.tex

\newcommand{\pr}[1]{p^{(#1)}}
\newcommand{\Fr}[1]{F^{(#1)}}
\newcommand{\Tr}[1]{T^{(#1)}}

\newcommand{\qr}[1]{q^{(#1)}}
\newcommand{\Qr}[1]{Q^{(#1)}}
\newcommand{\Hr}[1]{H^{(#1)}}

\newcommand{\TT}{\ensuremath{\mathcal{T}}}

\section{The Upper Bound}\label{sec:upper_rand}

We now present our streaming algorithm for single-source shortest paths. This section focuses on our randomized approach for insertion-only streams, while we defer the (standard) extensions to deterministic algorithms and dynamic streams to~\Cref{app:extension}. The following theorem formalizes the first part of~\Cref{res:upper}.

\begin{theorem}\label{thm:upper-main}
	Let $G = (V, E, w)$  be an 
	 $n$-vertex graph  with non-negative edge weights presented in 
	an insertion-only stream.
	Given a source vertex 
	  $s \in V$ and  parameters $k \in [\ln{n}]$ and $\eps \in (0,1)$, 	there is a randomized streaming algorithm that 
	 outputs a $(1+\eps)$-approximate shortest path tree rooted at source $s$ using 
	\[
		O\paren{\frac{k}{\eps} \cdot n^{1+1/k} \cdot \log{n}}~\text{space} \quad \text{and} \quad O\left(\frac{k^2}{\eps}\right) ~\text{passes},
	\]
	with probability at least $1-n^{-\Omega(n)}$.
\end{theorem}

%We remark that if we plug in $k = \log n / \log \log n$ to \Cref{thm:upper-main}, we obtain a semi-streaming algorithm with $O(\eps^{-1} \cdot n \log^3 n)$ space and $O\left(\left( \frac{\log n}{\log \log n} \right)^2 \right)$ passes.
%We note that the assumption of integer weights and $\poly{(n)}$ bounds on them in~\Cref{thm:upper-main} is for simplicity of exposition and can be easily removed.

We first present our algorithm and its analysis independently of any computational model, and then verify that it can be implemented in the streaming model efficiently in \Cref{lem:upper-stream}.

At a high level, the algorithm begins by computing an $O(k)$-spanner $H$ of the input graph $G$. Then, it goes through $R=O(k^2/\eps)$ rounds and in each round $r \in [R]$, 
it samples each edge $e \in E$ with some prescribed probability $\pr{r}_e$ into a set $\Fr{r} \subseteq E$, which we use to locally compute an exact shortest path tree $\Tr{r}$ of $H \cup \Fr{r}$. 
Finally, we increase the sampling probabilities of any edge in $E$ that violates triangle inequality with respect to the computed tree $\Tr{r}$ for the next round and continue. 
At the end, we compute another shortest path tree on all existing trees $\Tr{r}$ for $r \in [R]$.

\begin{Algorithm}\label{alg:s_t} \normalfont
	An algorithm for approximating shortest paths in a graph $G=(V,E,w)$ from a source $s \in V$ (with given integer $1 \leq k \leq \ln{n}$ and real $\eps \in (0,1)$).

	\begin{enumerate}
		\item Compute a $2k$-spanner $H = (V, E_H)$ of $G$. % with size $|E_H| = O(n^{1+1/k} \log n)$.
		\item Assign each edge $e \in E$ an importance attribute $\qr{1}_e := 1$,  and let $\Qr{1} := \sum_{e \in E} \qr{1}_e$.
		\item For rounds $r=1$ to $R := 10  k^2/\eps$:
		\begin{enumerate}
			\item \label{alg:3a} Sample $\Fr{r} \subseteq E$ by picking each edge $e \in E$ independently and with probability 
			\begin{equation} \label{eq:p_e}
				\pr{r}_e :=  \min\set{1,\paren{\frac{10}{\eps} \cdot k \cdot n^{1+1/k} \log n} \cdot \frac{\qr{r}_e}{\Qr{r}}}.
			\end{equation}
			\item \label{alg:1} Let $\Hr{r} := (V, E_H \cup  F^{(i)})$; compute a shortest path tree $\Tr{r}$ in $\Hr{r}$ rooted at $s$.       
			\item\label{alg:3c} For every edge $e=(u, v) \in E$, if 
			\[
				\card{\dist_{\Tr{r}}(s, u) - \dist_{\Tr{r}}(s, v)} > w(u, v)
			\]
			let $\qr{r+1}_e := (1+n^{1/k}) \cdot \qr{r}_e$ and otherwise $\qr{r+1}_e = \qr{r}_e$. Let $\Qr{r+1} := \sum_{e \in E} \qr{r+1}_e$. 
		\end{enumerate}
		\item Return a shortest path tree rooted at $s$ in the graph $H^* := \Tr{1} \cup \dots \cup \Tr{R}$. 
	\end{enumerate}
\end{Algorithm}

We now analyze the algorithm. We first need the following ``sampling lemma'' that uses the randomness of $\Fr{r}$ to argue that
the tree $\Tr{r}$ computed by the algorithm in round $r \in [R]$ cannot lead to a ``large'' increase in the total importances. 

\begin{lemma}[Sampling Lemma] \label{lem:sampling-lemma}
	For every $r \in [R]$, with probability at least $1-n^{-\Omega(n)}$, 
	\[
	\Qr{r+1} \leq \paren{1+\frac{\eps}{10k}}  \cdot \Qr{r}.
	\]
\end{lemma}
\begin{proof}
	Fix any $r \in [R]$. 
	Let $\TT_1, \dots, \TT_t$ denote the set of all spanning trees of $G$; by Cayley's tree formula (\Cref{fact:treeform}), we have $t \le n^{n-2}$. 
	For a tree $\TT_j$ for $j \in [t]$, we define the set $B(\TT_j) \subseteq E$ as: 
	\[
		B(\TT_j) := \set{(u,v) \in E ~\text{s.t.}~ \card{\dist_{\TT_j}(s, u) - \dist_{\TT_j}(s, v)} > w(u, v)};
	\]
	namely, these are the ``bad'' edges of $G$ that will have their importances increased in~\Cref{alg:s_t}, \emph{if} the algorithm chooses $\TT_j$ 
	as the shortest path tree $\Tr{r}$ in round $r$. 
	
	We say a tree $\TT_j$ is \textbf{bad} iff 
	\begin{equation} \label{eq:B_j-importance}
		\sum_{e \in B(\TT_j)} \qr{r}_e > \frac{\eps}{10k \cdot n^{1/k} } \cdot \Qr{r},
	\end{equation}
	namely, the total importances of the edges $B(\TT_j)$ is ``high''. We similarly say that the computed tree $\Tr{r}$ is bad if $\Tr{r} = \TT_j$ for some bad tree $\TT_j$. 
	Notice that if the computed tree $\Tr{r}$ of round $r$ is \emph{not} bad, then we will have 
	\begin{align*}
		\Qr{r+1} &= \sum_{e \in E} \qr{r+1}_e = \sum_{e \notin B(\Tr{r})} \qr{r}_e + \sum_{e \in B(\Tr{r})} (1+n^{1/k}) \cdot \qr{r}_e \\
		&= \sum_{e \in E} \qr{r}_e + \sum_{e \in B(\Tr{r})} n^{1/k} \cdot \qr{r}_e \leq \Qr{r} + \frac{\eps}{10k} \cdot \Qr{r}, 
	\end{align*}
	by~\Cref{eq:B_j-importance} in the last step for the tree $\Tr{r}$ (which we assumed is not bad). Thus, to establish the claim, we only need to prove that with exponentially high probability, the tree $\Tr{r}$ 
	is not bad. To do so, we bound the probability that any bad tree $\TT_j$ can be a shortest path tree in the graph $\Hr{r}$, using the randomness of the sample $\Fr{r}$, 
	and union bound over all bad trees to conclude the proof. 
	
	Fix a bad tree $\TT_j$. For $\TT_j$ to be a possible shortest path tree in $\Hr{r}$, it should happen that none of the edges in $B(\TT_j)$ are sampled in $\Fr{r}$. This is because 
	for every edge $(u,v) \in \Fr{r} \subseteq \Hr{r}$, 
	\[
		\card{\dist_{\Hr{r}}(s, v) -  \dist_{\Hr{r}}(s, u) } \leq w(u, v),
	\]
	by the triangle inequality, and any shortest path tree of $\Hr{r}$ will also have to satisfy this triangle inequality (as $\{s\} \times V$  distances in the tree and $\Hr{r}$ are the same). 
	As such, 
	\begin{align*}
		\Pr\paren{\Tr{r} = \TT_j} &\leq \Pr\paren{\Fr{r} \cap B(\TT_j) = \emptyset} = \prod_{e \in B(\TT_j)} (1-\pr{r}_e) \tag{by the independence and sampling probabilities of edges} \\
		&\leq \exp\paren{-\sum_{e \in B(\TT_j)} {\frac{10}{\eps} \cdot k \cdot n^{1+1/k} \log n} \cdot \frac{\qr{r}_e}{\Qr{r}}} \tag{since $1+x \leq e^x$ and by~\Cref{eq:p_e} for values of $\pr{r}_e$} \\
		&\leq \exp\paren{-{\frac{10}{\eps} \cdot k \cdot n^{1+1/k} \log n} \cdot \frac{1}{\Qr{r}} \cdot \frac{\eps}{10 \cdot n^{1/k} \cdot k} \cdot \Qr{r}} \tag{since $\TT_j$ is bad and by~\Cref{eq:B_j-importance}} \\
		&= \exp\paren{-n\log{n}}. 
	\end{align*}
	A union bound over at most $t \leq 2^{n\log{n}}$ bad trees implies that 
	\[
		\Pr\paren{\text{$\Tr{r}$ is bad}} \leq \sum_{\text{bad trees $\TT_j$}} \Pr\paren{\Tr{r} = \TT_j} \leq \paren{\frac{2}{e}}^{n\log{n}} =n^{-\Omega(n)},  
	\]
	concluding the proof. 
\end{proof}

To continue, we need a quick definition. We say an edge $e=(u,v) \in E$ \textbf{fails} in round $r \in [R]$ of \Cref{alg:s_t} iff 
\[
\card{\dist_{\Tr{r}}(s, u) - \dist_{\Tr{r}}(s, v)} > w(u, v),
\]
namely, it violates the triangle inequality with respect to $\Tr{r}$ and  has its importance increased by a factor of $\approx n^{1/k}$. 
We use~\Cref{lem:sampling-lemma} to prove that no edge can fail in ``too many'' rounds. This will then be used to argue that the last step of the algorithm will return 
an approximate shortest path tree as desired.

\begin{claim} \label{clm:edge_fails}
	With probability $1-n^{-\Omega(n)}$, no edge $e \in E$ fails in more than ${\eps R}/({2k})$ rounds. 
\end{claim}
\begin{proof} 
	By a union bound over the rounds of the algorithm and~\Cref{lem:sampling-lemma}, we have that with probability $1-n^{-\Omega(n)}$, 
	\[
		\Qr{R+1} \leq \paren{1+\frac{\eps}{10k}}^{R} \cdot \Qr{1} \leq \exp\paren{\frac{\eps}{10k} \cdot \frac{10}{\eps} \cdot k^2} \cdot \card{E} \leq e^{k} \cdot n^2 \leq n^3,  
	\]
	using the choice of $R = (10/\eps) \cdot k^2$ and since $\card{E} \leq n^2$ and $k \leq \ln{n}$ (the second inequality is by $1+x \leq e^x$ and since $\Qr{1} = \card{E}$ as $\qr{r}_e = 1$ for every edge $e \in E$ by design).

	Suppose towards contradiction that there exists an edge $e \in E$ that fails more than ${\eps R}/({2k})$ rounds of \Cref{alg:s_t}. In each round $r$ that the edge $e$ fails, we increase its importance 
	by a factor larger than $n^{1/k}$. Thus, in that case, and by the positivity of importances, 
	\[
		\Qr{R+1} \geq \qr{R+1}_e \geq (n^{1/k})^{\eps R/(2k)} = n^{\eps \cdot R/(2k^2)} = n^5, 
	\]
	by the choice of $R = (10/\eps) \cdot k^2$. This contradicts the upper  bound on $\Qr{R+1}$, implying that $e$ could have not failed in more than $\eps R/(2k)$ rounds. 
\end{proof}

We are now ready to prove the correctness of \Cref{alg:s_t}. 

\begin{lemma}\label{lem:alg-correct} 
	\Cref{alg:s_t} returns a $(1+\eps)$-approximate single-source shortest tree rooted at $s$ with probability at least $1-n^{-\Omega(n)}$.  
\end{lemma}
\begin{proof}
	We condition on the high probability event of~\Cref{clm:edge_fails} and show that for every $t \in V$, 
	\[
		\dist_{H^*}(s,t) \leq (1+\eps) \cdot \dist_G(s,t);
	\]
	the rest follows immediately as the algorithm returns an exact shortest path tree of $H^*$ at the end.

	The proof is by the probabilistic method. Fix a vertex $t \in V$, and let $P = (s=v_0, \dots, v_{\ell+1}=t)$ be a shortest $(s, t)$-path in $G$. 
	Consider sampling an index $r \in [R]$ uniformly at random. We have 
	\begin{align}
		\dist_{H^*}(s,t) &\leq \Exp_r\bracket{\dist_{\Tr{r}}(s,t)} = \Exp_r\bracket{\sum_{i=0}^{\ell} \dist_{\Tr{r}}(s,v_{i+1}) - \dist_{\Tr{r}}(s,v_{i})} \tag{by a telescoping sum} \\
		&= \sum_{i=0}^{\ell} \Exp_r\bracket{\dist_{\Tr{r}}(s,v_{i+1}) - \dist_{\Tr{r}}(s,v_{i})} \tag{by linearity of expectation} \\
		&\leq  \sum_{i=0}^{\ell} \Exp_r\card{\dist_{\Tr{r}}(s,v_{i+1}) - \dist_{\Tr{r}}(s,v_{i})}, \label{eq:per-edge}
	\end{align}
	by taking the absolute values instead. We now bound each term in the RHS above. Consider an edge $e = (v_{i},v_{i+1})$ for some $i \in [\ell]$. 
	If $e$ does not fail in round $r$ of the algorithm, then we have 
	\[
		\card{\dist_{\Tr{r}}(s,v_{i+1}) - \dist_{\Tr{r}}(s,v_{i})} \leq w(v_i,v_{i+1}), 
	\]
	by definition. On the other hand, even if $e$ does fail in this round, since $\Tr{r}$ is a shortest path tree of $H \cup \Fr{r}$ which in particular includes a $2k$-spanner $H$, 
	we have 
	\[
		\card{\dist_{\Tr{r}}(s,v_{i+1}) - \dist_{\Tr{r}}(s,v_{i})} \leq 2k \cdot w(v_i,v_{i+1}), 
	\]
	i.e., it cannot ``fail too much''. Since each edge fails in at most $\eps R/2k$ of the rounds by~\Cref{clm:edge_fails}, over a random choice of $r \in [R]$, we have 
	\begin{align*}
		\Exp_r\card{\dist_{\Tr{r}}(s,v_{i+1}) - \dist_{\Tr{r}}(s,v_{i})} &\leq w(v_i,v_{i+1}) + \Pr_r\paren{\text{$e$ fails in round $r$}} \cdot 2k \cdot w(v_i,v_{i+1}) \\
		&\leq w(v_i,v_{i+1}) + \frac{\eps}{2k} \cdot 2k \cdot w(v_i,v_{i+1}) \\
		& \le (1+\eps) \cdot w(v_i, v_{i+1})
	\end{align*}
	for each $i \in [0, \ell]$. 
	Plugging in this bound in~\Cref{eq:per-edge}, we obtain that 
	\begin{align*}
		\dist_{H^*}(s,t) \leq \sum_{i=0}^{\ell} \Exp_r\card{\dist_{\Tr{r}}(s,v_{i+1}) - \dist_{\Tr{r}}(s,v_{i})} \leq \sum_{i=0}^{\ell} (1+\eps) \cdot w(v_i,v_{i+1})   = (1+\eps) \cdot \dist_G(s, t),  
	\end{align*}
	concluding the proof. 
\end{proof}

\subsection{A Streaming Implementation of \Cref{alg:s_t}}
\label{sec:streaming_implementation}

\Cref{lem:alg-correct} establishes the correctness of~\Cref{alg:s_t}. We now verify that this algorithm can be implemented in the streaming model as required by \Cref{thm:upper-main}.

\begin{lemma} \label{lem:upper-stream}
	\Cref{alg:s_t} can be implemented in insertion-only streams using 
	\[
		O\paren{\frac{k}{\eps} \cdot n^{1+1/k} \cdot \log{n}}~\text{space} \quad \text{and} \quad O\left(\frac{k^2}{\eps}\right) ~\text{passes}. 
	\]
	with probability at least $1-n^{-\Omega(n)}$. 
\end{lemma}
\begin{proof}
	We can compute and explicitly store the $2k$-spanner $H$ in one pass and $O(n^{1+1/k}\log{n})$ space using \Cref{prop:spanner} (by setting $\eps$ in the proposition to be $0.5$).
	
	We cannot explicitly store the importances $q_e^{(i)}$ for every edge $e \in E$ and round $r \in [R]$, since $|E|$ may be much larger than our desired space bound. 
	Instead, we store a small data structure that, when queried with $(e,r)$, returns $\qr{r}_e$. We implement this data structure as follows:
	\begin{itemize}
		\item After each round $r \in [R]$ of \Cref{alg:s_t}, we explicitly store the shortest path tree $\Tr{r}$. Each tree requires $O(n)$ space to store, 
		and given the total number of rounds is $R = O(k^2/\eps)$, we need $O(nk^2/\eps)$ space in total for this step, which is within our budget since $k = O(\log{n})$. 
		\item When queried with an edge $e \in E$ and an index $r \in [R]$ in round $r$ of \Cref{alg:s_t}, our data structure 
		uses the stored shortest path trees $T^{(1)}, \dots, T^{(r-1)}$ to compute edge importance $q_e^{(r)}$ by counting the number of trees in which $e$ violates the triangle inequality and applying the 
		formula in 	Step~\ref{alg:3c} of \Cref{alg:s_t}. (The data structure returns $\qr{1}_e = 1$ for all edges $e$.)
	\end{itemize}
	
	For every $r \in [R]$, we can use our implicit representation of importances $\qr{r}_e$ for $e \in E$ to implement the round $r$ of \Cref{alg:s_t} in two passes: 
	in the first pass, we only compute $\Qr{r}$ by summing up $\qr{r}_e$ for each arriving edge $e$ in the stream using our data structure above. 
	In the second pass, when the edge $e$ arrives in the stream, we compute $\pr{r}_e$ using our data structure and the value of $\Qr{r}$ and sample the edge into $\Fr{r}$ with probability $\pr{r}_e$ as required.

	We have shown how to efficiently compute edge probabilities $\pr{r}_e$ and sample $\Fr{r}$. Additionally, by the Chernoff bound,
	\[
	\card{\Fr{r}} \leq 2 \cdot \Exp\card{\Fr{r}} = 2 \cdot \sum_{e \in E} \pr{r}_e = O(\frac{k}{\eps} \cdot n^{1+1/k} \cdot \log{n}), 
	\]
	with probability $1-n^{-\Omega(n)}$. Then, we can store the edges in $\Fr{r}$ explicitly in each round $r$ given our space budget. 
	Since we have stored the spanner $H$ and sampled edges $\Fr{r}$ locally in round $r$, we can compute the tree $\Tr{r}$ locally as well, without any additional passes over the stream.  
	
	The total number of passes of the algorithm is also $O(R) = O(k^2/\eps)$, finalizing the proof. 
\end{proof}

\Cref{thm:upper-main} follows immediately from~\Cref{lem:alg-correct} and~\Cref{lem:upper-stream}.

	% !TEX root = main.tex
\section{The Lower Bound}\label{sec:lb}

In this section, we prove that approximating the length of the shortest path between any two vertices $s, t$ up to constant multiplicative factors requires $\Omega(\log n/\log \log n)$ passes in semi-streaming.

\begin{restatable}{theorem}{mainlb}
	\label{thm:main-lb}
	For any $\alpha \ge 1$ and any  integer $p$ sufficiently smaller than $\log(n) / \alpha$, any randomized streaming algorithm on unweighted graphs that  can output an $\alpha$-approximation to the distance between a fixed pair of vertices with probability at least $2/3$ requires
	$$
\Omega\left(p^{-7} \cdot  n^{1+\frac{1}{4\alpha \cdot (p+2)}}\right) \text{ space } \quad \text{or } \quad	p  \text{ passes.}
	$$
\end{restatable}

This directly proves \Cref{res:lower} because for any constant $\alpha \ge 1$, for any $p = (1/\alpha) \cdot o(\log n/\log \log n)$, the lower bound on space is $\omega(n \cdot \poly \log n)$, which is more than semi-streaming space. 

We begin by stating the different steps of our lower bound proof.  
Initially, we are given a streaming algorithm $\mathcal{A}$ that estimates the shortest path distance between input vertices $s, t$ in input graph $G$. We proceed as follows.
\begin{enumerate}[label=$(\roman*)$]
	\item\label{item:part1-or} 
	We translate algorithm $\mathcal{A}$ into a protocol $\pi_{\text{OR}}$ for the $\ORPPC$ (see \Cref{def:ORPPC}) problem, which corresponds to the 
	$\OR$ of many copies of the paired pointer chasing problem ($\PPC$, \Cref{def:paired-pc}). This is covered in \Cref{lem:psc-stream}.   
	\item\label{item:part2-dirsum} We use protocol $\pi_{\text{OR}}$ to get a protocol $\pi$ for solving the $\PPC$ problem with low communication using direct sum and message compression techniques, which is covered in \Cref{prop:orppc-ppc}.
	\item\label{item:part4-lb-ppc}  Lastly, we prove a lower bound on the communication required to solve $\PPC$ instances on our hard distribution, stated in \Cref{lem:PPC-lb}. This is the most novel part of our proof. 
\end{enumerate}

In this section, we first set up the graph constructions used in the lower bound. Then, we prove \Cref{thm:main-lb}, assuming the lower bound for $\PPC$. Lastly, we give our lower bound for $\PPC$. 

\subsection{Layered Graphs and Paired Pointer Chasing}\label{subsec:layer-ppc}

In this subsection, we develop our  definitions of the paired pointer chasing problem ($\PPC$) and the $\ORPPC$ problem.  
We will conclude the section by completing  \Cref{item:part1-or} of our lower bound framework, stated in \Cref{lem:psc-stream}. 

We begin by defining a graph family that we will use in our problem formulations.

%In this subsection, we define the graph primitives that will be used in our lower bound construction. We also define the main communication problem, paired pointer chasing. 

\begin{Definition}[Layered Graph]\label{def:layered-graph}
	For integers $w, d \geq 1$, a layered graph is any graph $G = (V, E)$ that satisfies the following properties:
	\begin{enumerate}[label=$(\roman*)$]
		\item The vertex set $V$ can be partitioned into $d$ equal size layers $V_1, V_2, \ldots, V_d$, each of size $w$. We let the set $V_i $ be identified by $\{\ver{i,1}, \ver{i,2}, \ldots, \ver{i,w}\}$. 
		\item The edge set $E$ can be partitioned into $d-1$ matchings $M_1, M_2, \ldots, M_{d-1}$ such that for $i \in [d-1]$, $M_i$ is a perfect matching between vertices in $V_i$ and vertices in $V_{i+1}$. 
	\end{enumerate}
	We use $\cL_{d, w}$ to denote the set of all layered graphs with $d$ layers and $w$ nodes in each layer.
	For any $j \in [w]$, we use $\point{G}{\ver{1,j}}$ to denote the unique vertex $\ver{d, \ell} \in V_d$ for some $\ell \in [w]$ that has a path to $\ver{1,j} \in V_1$ in graph $G$. 
\end{Definition}

We define a \emph{product} operation on  layered graphs that will be useful in our analysis in \Cref{sec:lb-ppc}. Roughly speaking, 
the product operation creates a new graph out of two layered graphs by taking the Cartesian product of their vertices as the new vertex set and then defining edges of each product vertex based on the
edges of each of its coordinates in the first two graphs.

\begin{Definition}[Product of Layered Graphs]\label{def:prod-2pc}
	For any integers $w, d \geq 2$, given any two graphs $G^1 = (V^1, E^1), G^2 = (V^2, E^2) \in \cL_{d, w}$, the product graph of $G^1, G^2$, denoted by $\product{G^1}{G^2} = (\Vprod, \Eprod)$ is defined as follows:
	\begin{enumerate}[label=$(\roman*)$]
%%		\item Let the partition of $V^1 $ be $V^1_1 \sqcup V^1_2 \sqcup \ldots \sqcup V^1_d$, and the partition of $V^2 $ be $V^2_1 \sqcup V^2_2 \sqcup \ldots \sqcup V^2_d$.
		\item The vertex set $\Vprod$ is $\Vprod_1 \sqcup \Vprod_2 \sqcup \ldots \sqcup \Vprod_d$ with $\Vprod_i = (V^1_i \times V^2_i) = \{(x,y)  \mid x \in V_i^1, y \in V_i^2\}$, 
		where $V^1_1 \sqcup V^1_2 \sqcup \ldots \sqcup V^1_d$ is the (layered) partition of $V^1$ and $V^2_1 \sqcup V^2_2 \sqcup \ldots \sqcup V^2_d$ is for $V^2$. 
		\item The edge set $\Eprod$ is comprised of $d-1$ matchings $\Mprod_1, \Mprod_2, \ldots, \Mprod_{d-1}$ with \[
		\Mprod_i = \{((x_{i}, y_{i}), (x'_{i+1}, y'_{i+1})) \mid (x_i, x'_{i+1}) \in M^1_i, (y_i, y'_{i+1}) \in M^2_i\},
		\]
		where $M^1_i$ and $M^2_i$ are the $i^{\textnormal{th}}$ matchings in $G^1$ and $G^2$ respectively, for $i \in [d-1]$. 
	\end{enumerate}
	See \Cref{fig:prod-layer} for an illustration.
\end{Definition}

\begin{figure}[h!]
	\centering
	\begin{subfigure}[b]{0.45\textwidth}
		\centering
		\includegraphics[scale=0.35]{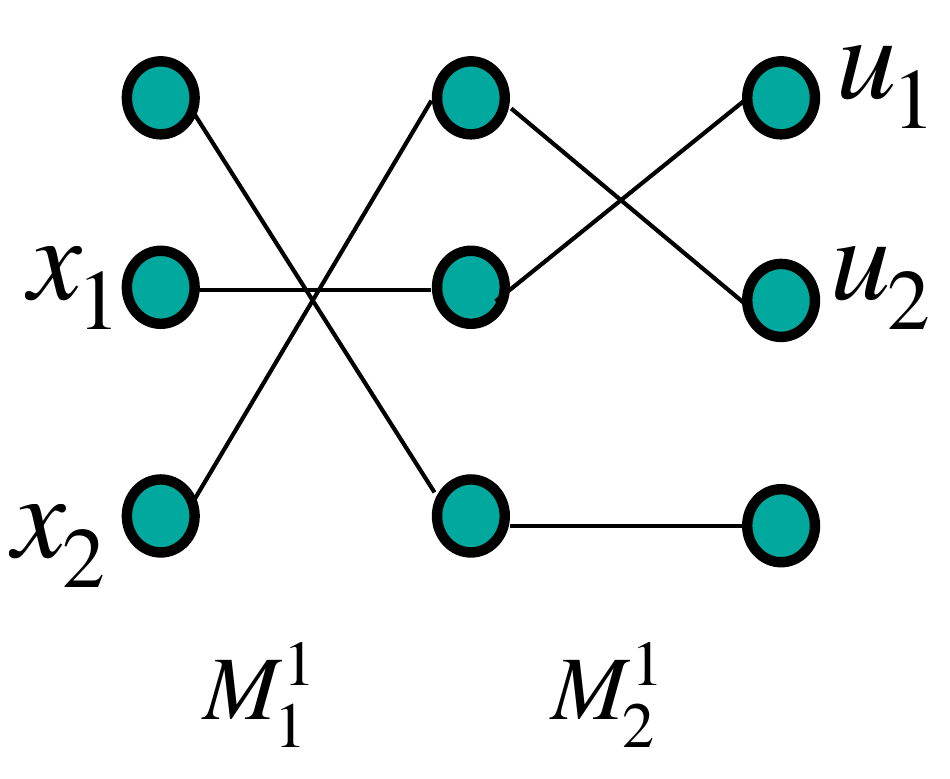}
		\caption{Layered graph $G^1$.}
	\end{subfigure}
	\hfill
	\begin{subfigure}[b]{0.45\textwidth}
		\centering
		\includegraphics[scale=0.35]{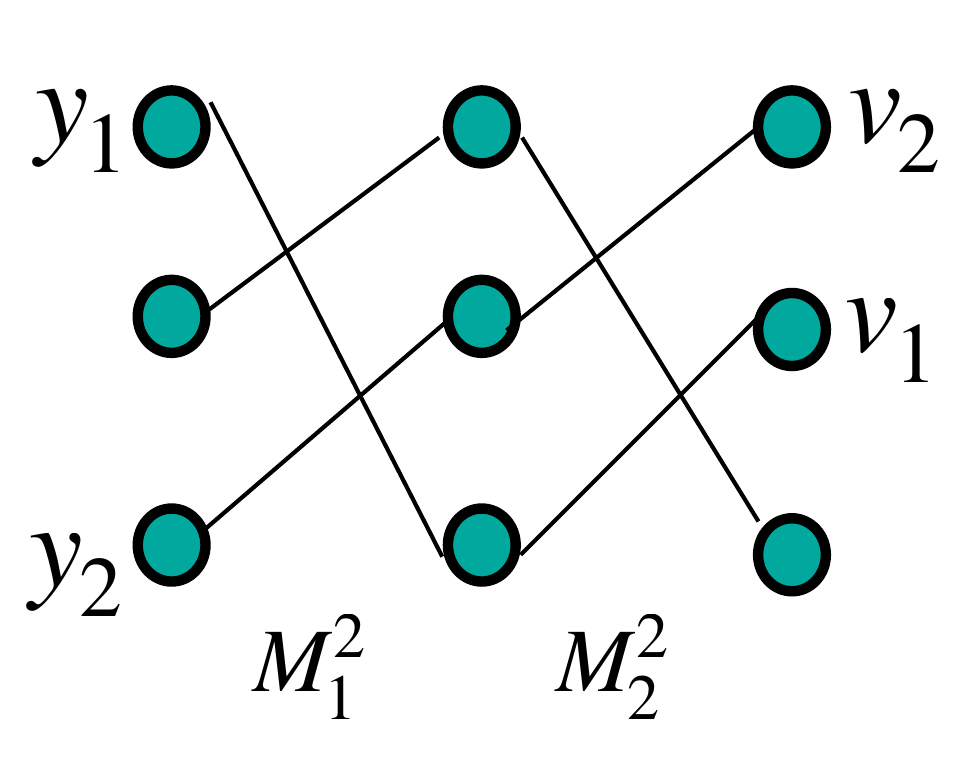}
		\caption{Layered graph $G^2$.}
	\end{subfigure}
	
	\vspace{0.5cm}
	
	\begin{subfigure}[b]{0.9\textwidth}
		\centering
		\includegraphics[scale=0.35]{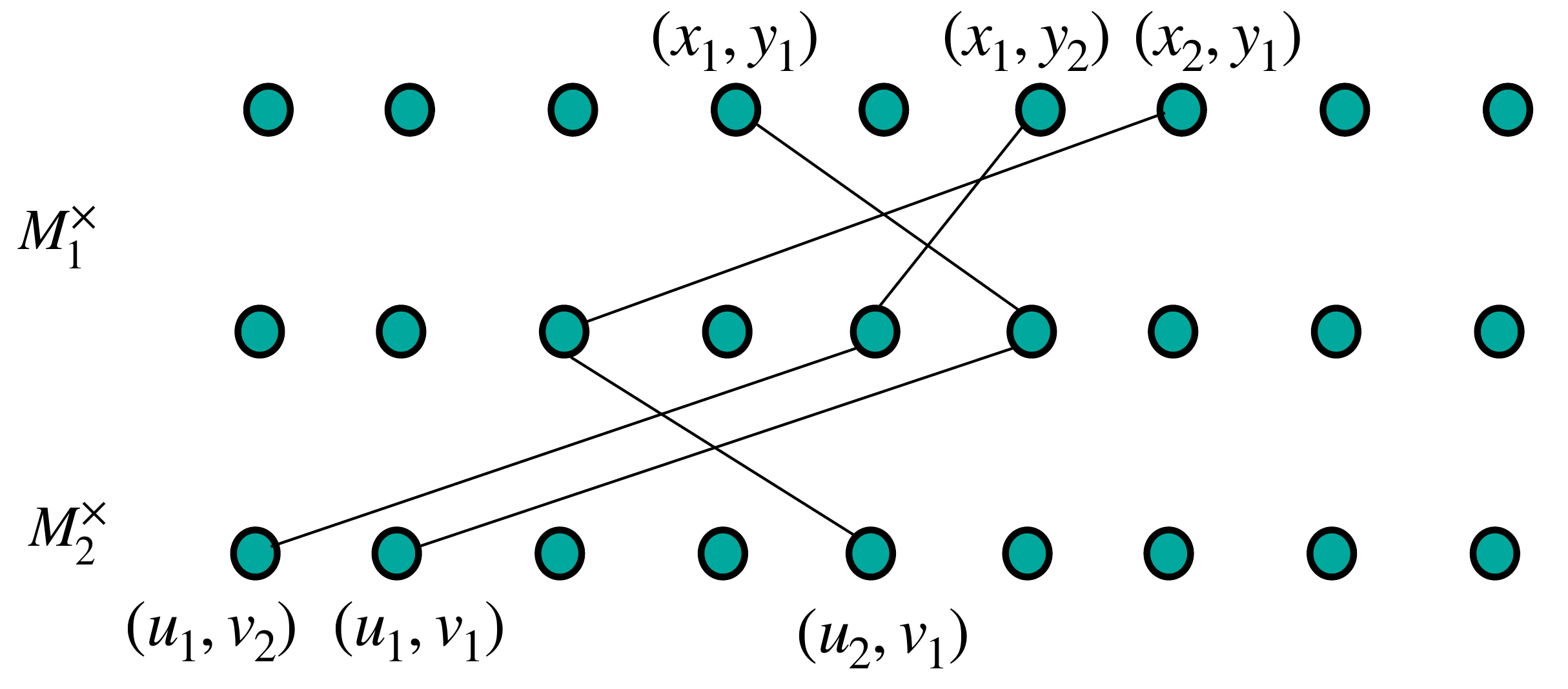}
		\caption{Graph $\Gprod =\product{G^1}{G^2}$ from set $\cL_{3,9}$. Also illustrates that $\point{\Gprod}{(x_1, y_1)} = (\point{G^1}{x_1}, \point{G^2}{y_1}) = (u_1, v_1)$. }
	\end{subfigure}
	\caption{ 
		Product of two layered graphs from $\cL_{3,3}$ lying in $\cL_{3, 9}$. Some edges are omitted.}
	\label{fig:prod-layer}
\end{figure}

We will  make use of the following simple observations about products of layered graphs.

\begin{observation}\label{obs:product-basic}
	For any $G^1, G^2 \in \cL_{d, w}$, and graph $G = \product{G^1}{G^2}$, we have, 
	\begin{enumerate}[label=$(\roman*)$]
		\item The graph $G$ belongs to $\cL_{d, w^2}$. 
		\item For any $(x,y) \in \Vprod_1$, $\point{G}{(x,y)} = (\point{G^1}{x}, \point{G^2}{y})$. 
	\end{enumerate}
\end{observation}
\begin{proof}
	For part $(i)$, we know that the number of vertices in each layer of $\Vprod$ is $w^2$. For $i \in [d-1]$, each matching $\Mprod_i$ in $\Eprod$ takes any $(x, y) \in \Vprod_i$ to a unique vertex in $\Vprod_{i+1}$, and is a perfect matching. 
	Part $(ii)$ is also evident from the definition of the product operation. See \Cref{fig:prod-layer} for an illustration of these properties.
\end{proof}

We also define a \textit{join} operation on layered graphs, which will be useful for establishing  \Cref{item:part1-or} of our lower bound framework in \Cref{lem:psc-stream}. 
The join operation takes two layered graphs, ``reverses'' the order of the second one, and connects the two graphs using a fixed perfect matching between their last layers. This reversal of the order of the second layered graph is necessary, as this operation allows us to check for equality between $\point{G^1}{v_{1,j}}$ and  $\point{G^2}{v_{1,j}}$ in two different layered graphs $G^1$ and $G^2$ for any $j \in [w]$. See \Cref{obs:join-basic} for details.

%This operation also serves as a way to visualize of $\ORPPC$ (to be defined later in \Cref{def:ORPPC}), as it 

\begin{Definition}[Join of Layered Graphs]\label{def:join-2pc}
		For integers $w, d \geq 2$, given two graphs $G^1 = (V^1, E^1), G^2 = (V^2, E^2) \in \cL_{d, w}$, the join graph of $G^1, G^2$, denoted by $\join{G^1}{G^2} = (\Vjoin, \Ejoin)$ is defined as follows:
	\begin{enumerate}[label=$(\roman*)$]
%%		\item Let the partition of $V^1 $ be $V^1_1 \sqcup V^1_2 \sqcup \ldots \sqcup V^1_d$, and the partition of $V^2 $ be $V^2_1 \sqcup V^2_2 \sqcup \ldots \sqcup V^2_d$.
		\item The vertex set $\Vjoin$ is a union of $2d$ many layers of size $w$ each where, 
		\[
			\Vjoin_i = \begin{cases}
				&V^1_i \textnormal{ if $i \in [d]$}, \\
				&V^2_{2d+1-i} \textnormal{ if $d < i \leq 2d$},
			\end{cases}
		\]
		where $V^1_1 \sqcup V^1_2 \sqcup \ldots \sqcup V^1_d$ is the (layered) partition of $V^1$ and $V^2_1 \sqcup V^2_2 \sqcup \ldots \sqcup V^2_d$ is for $V^2$. 
		\item The edge set $\Ejoin$ is,  
		 \[
		\Ejoin = E^1 \cup E^2 \cup \{(v^1_{d, j}, v^2_{d, j}) \mid j \in [w]\},
		\]
		where $v^1_{d, j}$ is the $j^{\textnormal{th}}$ vertex in $V^1_d$ and $v^2_{d, j}$ is the $j^{\textnormal{th}}$ vertex in $V^2_d$.
	\end{enumerate}
	See \Cref{fig:join-layer} for an illustration.
\end{Definition}

\begin{figure}[h!]
	\centering
	\begin{subfigure}[b]{0.45\textwidth}
		\centering
		\includegraphics[scale=0.35]{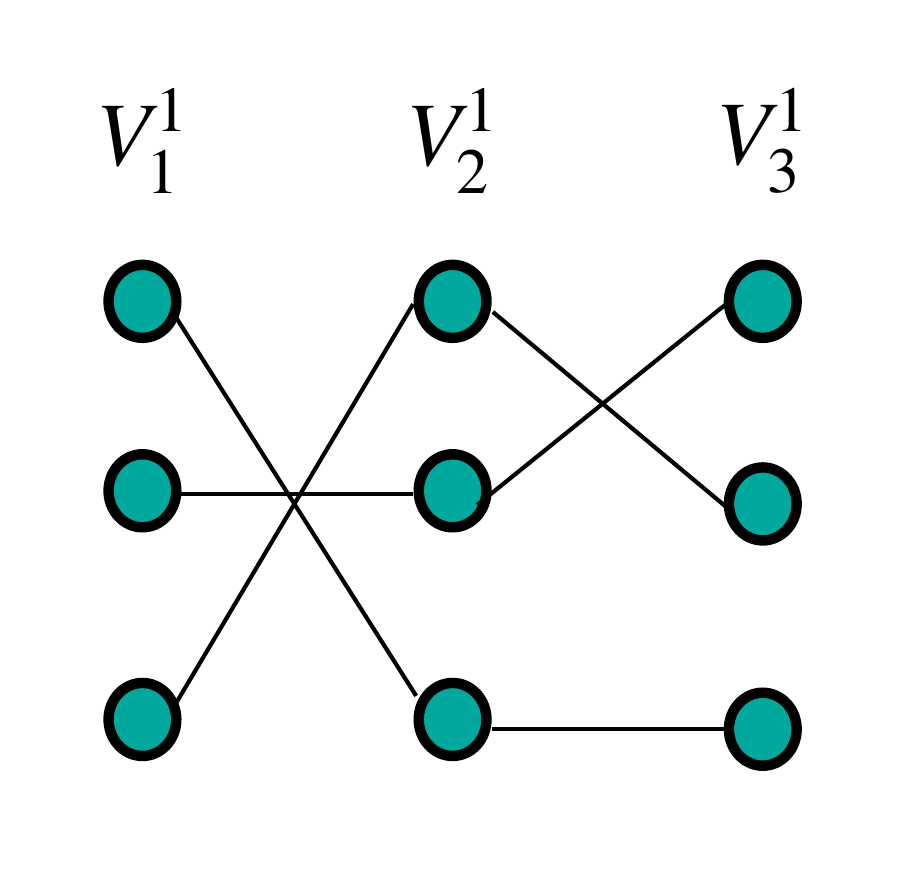}
		\caption{Layered graph $G^1$.}
	\end{subfigure}
	\hfill
	\begin{subfigure}[b]{0.45\textwidth}
		\centering
		\includegraphics[scale=0.35]{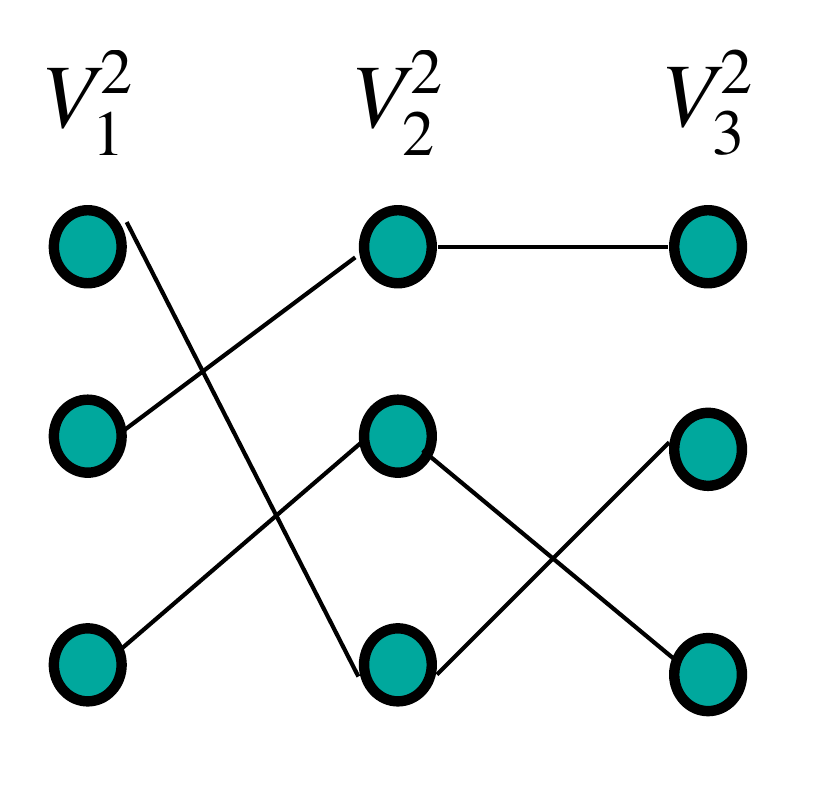}
		\caption{Layered graph $G^2$.}
	\end{subfigure}
	
	\begin{subfigure}[b]{0.9\textwidth}
		\centering
		\includegraphics[scale=0.40]{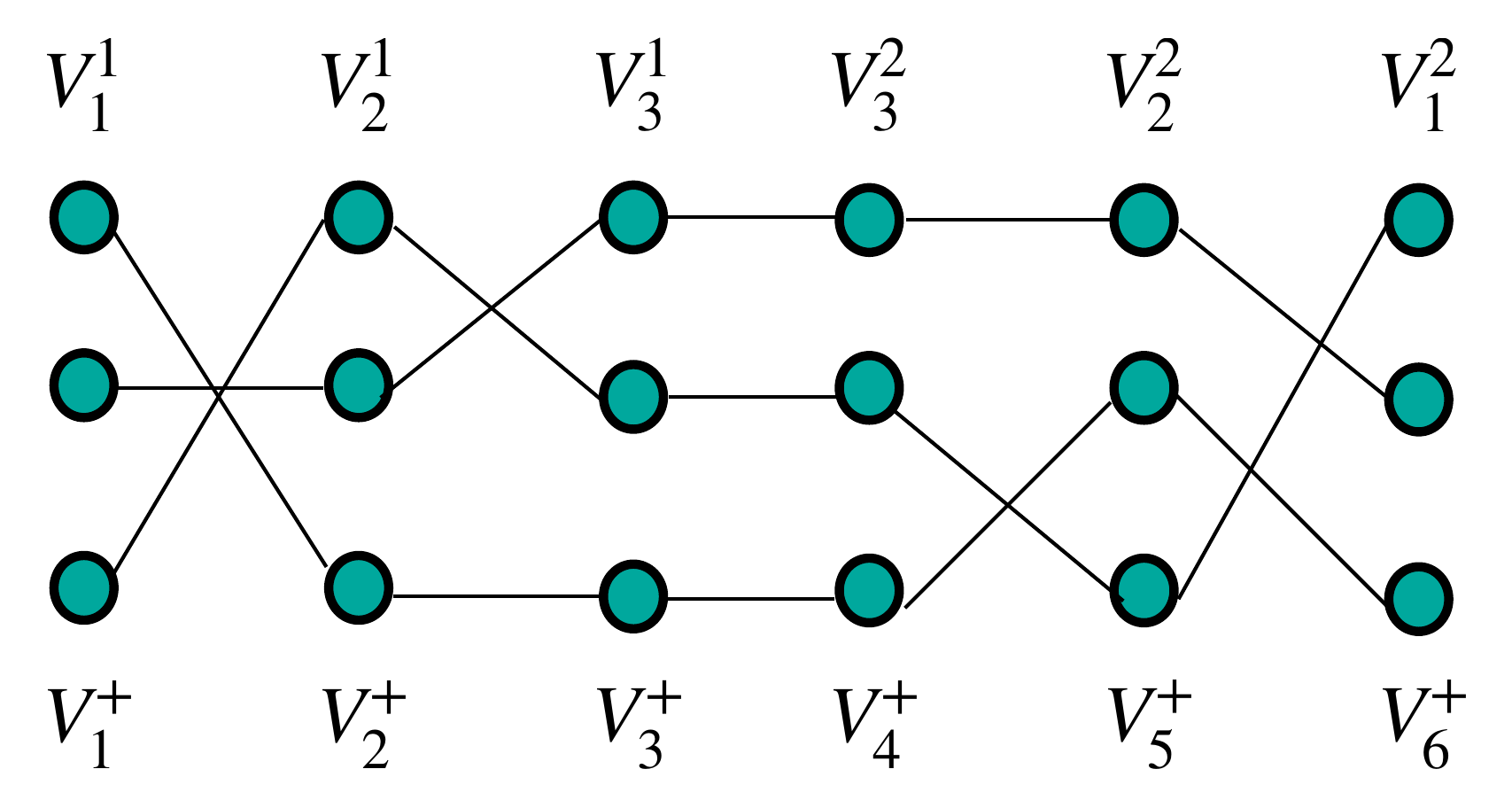}
		\caption{Graph $\Gprod =\join{G^1}{G^2}$ from set $\cL_{6,3}$. The name of the layers in $G^1$ and $G^2$ is specified above.}
	\end{subfigure}
	\caption{ 
		Join of two layered graphs from $\cL_{3,3}$ lying in $\cL_{6, 3}$. }
	\label{fig:join-layer}
\end{figure}
We will need the following observations about the join operation. 

\begin{observation}\label{obs:join-basic}
	For any $G^1, G^2 \in \cL_{d, w}$, and graph $G = \join{G^1}{G^2}$, we have, 
	\begin{enumerate}[label=$(\roman*)$]
		\item The graph $G$ belongs to $\cL_{2d, w}$. 
		\item For any $i \in [w]$, and vertex $v_{1, i} \in \Vjoin_1$, $\point{G}{v_{1,i}} = v_{2d, i}$ if and only if
		$\point{G^1}{v_{1, i}}$ and $\point{G^2}{v_{1, i}}$ are equal. 
	\end{enumerate}
\end{observation}
\begin{proof}
	Part $(i)$ follows directly from the join operation on two graphs $G^1, G^2 \in \mathcal{L}_{d, w}$. %Specifically, $G$ has $2d$ layers $V_1^+, \dots, V_{2d}^+$, and the set of edges in $G$ between adjacent layers $V_i^+$ and $V_{i+1}^+$  is a perfect matching for each $i \in [1, 2d-1]$. 
	
	We now verify part $(ii)$. We observe that $\point{G^1}{v_{1, i}}$ is the unique vertex in layer $V^1_d = V^+_d$ connected  to $v_{1, i}$ in $G$. This vertex is connected to $v_{1, i} \in V^2_1$ in $G$ if and only if $\point{G^2}{v_{1, i}}$ is the same as $\point{G^1}{v_{1,i}}$.
\end{proof}

We are ready to define the paired pointer chasing problem now. In the original pointer chasing problem, we have a layered graph and a fixed starting vertex, 
and the goal is to determine the unique vertex in the last layer reachable from the starting point. In this new problem, we have two instances of the pointer chasing problem 
that are correlated with each other so that with probability half they both reach the same vertex in the last layer, and  otherwise the two instances are independent. 

%Our new problem is 

\begin{Definition}[Paired Pointer Chasing]\label{def:paired-pc}
	For integers $w, d \geq 2$, the paired pointer chasing problem, denoted by $\PPC_{d, w}$ is defined as follows. 
	\begin{enumerate}[label=$(\roman*)$]
		\item There is a bit $b \in \{0,1\}$ chosen uniformly at random. 
		\item The input is two graphs $G^1 = (V^1, E^1)$ and $G^2 = (V^2, E^2)$ from $\cL_{d, w}$.  We use $M^j_i$ to denote the $i^{\textnormal{th}}$ matching in graph $G^j$ for $i \in [d-1]$ and $j \in \{1,2\}$. 
		\item Alice receives $M_i^j$ for all odd $i \in [d-1]$ and Bob receives $M_i^j$ for all even $i \in [d-1]$ for both $j \in \{1,2\}$.
		\item When $b = 0$, all the matchings in graphs $G^1$ and $G^2$ are chosen uniformly at random. 
		\item When $ b = 1$, all the matchings in graphs $G^1$ and $G^2$ are chosen uniformly at random conditioned on $\point{G^1}{\ver{1,1}} = \point{G^2}{\ver{1,1}}$. 
		\item The players have $d-2$ rounds to communicate. At the end of all $d-2$ rounds, the player who receives the last message has to output the value of bit $b$. 
	\end{enumerate}
	We use $\muPPC$ to denote the distribution of the input in $\PPC_{d, w}$ problem. We use $\muPPC^0, \muPPC^1$ to denote the distribution of the input conditioned on $b= 0$ and $b = 1$ respectively.
\end{Definition}

We can use a streaming algorithm that estimates shortest path lengths to check whether the $\OR$ of many instances of $\PPC$ evaluate to $1$. This is a promise problem with $t$-many instances of $\PPC$, distributed so that exactly one of these instances is drawn from $\muPPC^1$, or all of them are drawn from $\muPPC^0$. Notice that if all $t$-instances are sampled $\muPPC^0$, all matchings are chosen independently and uniformly, but this is not true even if one one $\PPC$ instance is sampled from $\muPPC^1$. The correlation between different matchings in $\PPC$ problem persists in $\ORPPC$ also.  
Formally, we define the $\ORPPC$ problem distributionally as follows. 

\begin{Definition}[$\OR$ of many $\PPC$ instances]\label{def:ORPPC}
	For integers $t, w, d \geq 2$, the $t$-$\OR$ of $\PPC$ problem, denoted by $\ORPPC_{t, d, w}$ is defined as follows. 
		\begin{enumerate}[label=$(\roman*)$]
		\item There is a bit $b \in \{0,1\}$ chosen uniformly at random.
		\item There are $t$ instances of $\PPC$, denoted by $(G^1_i, G^2_i)$ for $i \in [t]$. The input matchings of these instances are given to Alice and Bob as in \Cref{def:paired-pc}.
		\item  When $b = 0$, $(G^1_i, G^2_i)$ is sampled from $\muPPC^0$ for all $i \in [t]$. 
		\item When $b = 1$, $\istar \in [t]$ is chosen uniformly at random. $(G^1_{\istar}, G^2_{\istar})$ is sampled from $\muPPC^1$, and $(G^1_i, G^2_i)$ is sampled from $\muPPC^0$ for all  $i \in [t]$ with $i \neq \istar$. 
		\item The players have $d-2$ rounds to communicate. At the end of all $d-2$ rounds, the player who receives the last message has to output the value of bit $b$. 
	\end{enumerate}
\end{Definition}

Recall that \Cref{item:part1-or} of our lower bound framework is to reduce solving the $\ORPPC$ problem to computing a constant-factor approximation of the $s$-$t$ shortest path distance in the streaming model. With that goal in mind, for each $\ORPPC$ instance we define an associated graph  whose  shortest path distances can aid in finding the value of bit $b$. 

More formally, consider  an instance of $\ORPPC$ corresponding to $t$ instances of $\PPC$, denoted by  $(G^1_i, G^2_i)$ for $i \in [t]$. Let $G_i = G_i^1 + G_i^2$ for each $i \in [t]$. By \Cref{obs:join-basic}, graph $G_i$ belongs to $\mathcal{L}_{2d, w}$ for each $i \in [t]$. Then graphs $G_i$, $i \in [t]$, can be interpreted as sharing the same vertex set specified in \Cref{def:layered-graph}, so $V(G_i) = V(G_j)$ for all $i, j \in [t]$.

We define the \textbf{collection graph} $G = (V, E)$ of our $\ORPPC$ instance as the edge union of graphs $G_i$, $i \in [t]$, so that
$$
	G = G_1 \cup \dots \cup G_t.
$$
We use $V = V_1 \sqcup V_2 \sqcup \ldots \sqcup V_{2d}$ with $V_j = \{v_{j,\ell} \mid \ell \in [w]\}$ to denote the  $2d$ layers of the collection graph. The collection graph contains all the matchings in each of the graphs $G_i$, $i \in [t]$. 

We can prove the following properties about the collection graph. 

\begin{observation}\label{obs:b-1-collection}
For any instance of  $\ORPPC_{t, d, w}$, if $b=1$,
then the distance between   vertices $\ver{1,1}$ and $\ver{2d, 1}$ in  collection graph $G$ is at most $\dist_G(\ver{1, 1}, \ver{2d, 1}) \le 2d-1$.
\end{observation}
\begin{proof}
When $b=1$, there is an index $i^* \in [t]$ such that $(G_{i^*}^1, G_{i^*}^2)$ is sampled from $\mu_{\PPC}^1$. By the definition of $\mu_{\PPC}^1$, we have that  $\point{G_{i^*}^1}{v_{1, 1}} = \point{G_{i^*}^2}{v_{1, 1}}$. Additionally, by \Cref{obs:join-basic} it follows that $\point{G_{\istar}}{v_{1, 1}} = v_{2d, 1}$, so graph  $G_{i^*}$  contains a path of length $2d-1$ between $v_{1, 1}$ and $v_{2d, 1}$. Since $G_{i^*}$ is a subgraph of the collection graph $G$, we conclude that there is a path of length $2d-1$ between $v_{1, 1}$ and $v_{2d, 1}$ in $G$ as well.
\end{proof}

\begin{claim} \label{clm:b-0-collection-long}
	For any instance of $\ORPPC_{t, d, w}$, if $b=0$, then with probability at least $1-\frac{(2t)^{2\alpha d}}{w}$, the distance between vertices $v_{1, 1}$ and $v_{2d, 1}$ in collection graph $G$ is at least $$\dist_G(v_{1, 1}, v_{2d, 1}) > \alpha \cdot 2d.$$
\end{claim}
\begin{proof}
 When $b=0$, graph pairs $(G_i^1, G_i^2)$ are sampled from $\mu_{\PPC}^0$ for each $i \in [t]$. Then the matchings in graph $G_i^1$ and in graph $G_i^2$ are chosen uniformly at random for each $i \in [t]$.

The collection graph $G = (V, E)$ is composed of $2d$ layers. The edge set $E$ of collection graph  $G$ can be characterized as follows:
\begin{itemize}
	\item Every edge $e \in E$ is between adjacent layers $V_i$ and $V_{i+1}$ for some $i \in [2d-1]$. 
	\item For each $i \in [1, d-1] \cup [d+1, 2d-1]$, the edge set  $E(V_i, V_{i+1})$ is composed of $t$ perfect matchings between $V_i$ and $V_{i+1}$, each chosen uniformly at random.
	\item The edge set between layers $V_d$ and $V_{d+1}$ is the perfect matching $ \{(v_{d, j}, v_{d+1, j}) \mid j \in [w]\}$. 	
\end{itemize}
Let $B \subseteq V$ denote the set of vertices in $G$ at a distance at most $\alpha \cdot 2d$ from $v_{1, 1}$. Since every vertex in $G$ has degree at most $2t$, we have that 
$
|B|  \le (2t)^{\alpha \cdot 2d}.
$  Additionally, for each $u, v \in V_{2d},$
$$
\Pr[u \in B] = \Pr[v \in B].
$$

This equality follows from a simple bijection between the  probability events in the distribution of the collection graph $G$. Let the set of graphs $\mathcal{Z} = \{H_1, \dots, H_z\}$ denote the support of random graph $G$. For each graph $H \in \mathcal{Z}$,  let $\phi(H)$ denote the graph $H$ with the labels for vertices $u$ and $v$ swapped. Then $\phi:\mathcal{Z} \mapsto \mathcal{Z}$ is a bijection. Moreover, 
$
\Pr[G = H] = \Pr[G = \phi(H)]
$
for each $H \in \mathcal{Z}$,  
by our earlier  characterization of collection graph $G$. Then 
$
\Pr[u \in B ] = \Pr[v \in B ],
$
as claimed. 

Since each vertex in $V_{2d}$ is contained in $B$ with equal probability, we know that
$$
|V_{2d}| \cdot \Pr[v_{2d, 1} \in B] = \sum_{v \in V_{2d}}\Pr[v \in B] \le \mathbb{E}[|B|] \le (2t)^{2\alpha d},
$$
where the final inequality follows from our earlier observation that $|B| \le (2t)^{2\alpha d}$. 
Rearranging gives  $ \Pr[v_{2d, 1} \in B] \le (2t)^{2\alpha d}/w$, completing the proof.
\end{proof}

%Now, we move to defining the communication problems involved. 

Using \Cref{obs:b-1-collection} and \Cref{clm:b-0-collection-long}, we can now complete \Cref{item:part1-or} of our reduction with the following lemma.

\begin{lemma}\label{lem:psc-stream}
	Suppose there exists a  streaming algorithm $\alg$ that $\alpha$-approximates the distance between $u, v$ in an $n$-vertex graph with success  probability at least $2/3$ in $p$ passes and $s$ bits of space, for some $\alpha \ge 1$ and $p \le \frac{\log n}{4\alpha}-2$. 	Then there is a protocol $\prot$ for the $\ORPPC_{t, d, w}$
	problem with parameters \[
	t = n^{\frac{1}{4\alpha (p+2)}},	 \quad d = p+2,
	  \quad \text{and} \quad w = \frac{ n}{2(p+2)},
	\]
	such that $\pi$ has success probability at least $2/3 - \log(n) \cdot n^{-1/4}$  and  communication cost $s$.
\end{lemma}
\begin{proof}
	Let $t = n^{\frac{1}{4\alpha (p+2)}}, d = p+2, \text{ and }  w = \frac{ n}{2(p+2)}$.
	Let $G = (V, E)$ be the collection graph corresponding to an instance of $\ORPPC_{t, d, w}$. Notice that $G$ has $|V| = 2dw = n$ vertices.

	\paragraph{Protocol for $\ORPPC_{t, d, w}$.}
	Alice and Bob will simulate streaming algorithm $\mathcal{A}$ to compute an $\alpha$-approximation of the shortest path distance between $v_{1, 1}$ and $v_{2d, 1}$ in $G$. 
	In rounds where Bob sends the message, Bob first runs algorithm $\mathcal{A}$ locally, using his edges in $G$ as the input stream to $\mathcal{A}$. 
	Then Bob sends the final work state $\mathcal{W}$ of $\mathcal{A}$ to Alice using $s$ bits of communication. Alice then
	continues running algorithm $\mathcal{A}$ starting from work state $\mathcal{W}$, while streaming her edges in $G$ to the algorithm.  
	In this manner, Alice and Bob simulate one pass of $\mathcal{A}$ using one round and $s$  communication. We can use a similar protocol  in rounds where Alice sends the message.\footnote{Notice that the order of edges in the stream are not the same in different passes in our simulation. If it is required by the algorithm that the order of edges is the same across all the passes, this can be done in $2p$ rounds. Asymptotically all the results remain the same.}
	
	Then in $p \le d-2$ rounds and $s$ communication cost, Alice and Bob can completely simulate $\mathcal{A}$ and compute an $\alpha$-approximation $\sigma \in \mathbb{R}_{\ge 0}$ of $\dist_G(v_{1, 1}, v_{2d, 1})$ with success probability at least $2/3$. 
	Recall that if $\sigma$ is an $\alpha$-approximation of $\dist_G(v_{1, 1}, v_{2d, 1})$, then
	$$
	\dist_G(v_{1, 1}, v_{2d, 1}) \le \sigma \le \alpha \cdot \dist_G(v_{1, 1}, v_{2d, 1}).
	$$
	The player who receives the last message does the following:
	\begin{itemize}
		\item outputs $b=1$ if $\sigma \le \alpha \cdot( 2d-1)$, and
		\item outputs $b=0$ if $\sigma > \alpha \cdot (2d-1)$.
	\end{itemize}

	 \paragraph{Correctness.} If $b=1$, then with probability at least $2/3$,  algorithm $\mathcal{A}$ succeeds, 
	 so $\sigma \le \alpha \cdot (2d-1)$ by \Cref{obs:b-1-collection}, and our protocol outputs $b=1$. Then our protocol succeeds with probability at least $2/3$ when $b=1$.  
	 
	 If $b=0$, then by \Cref{clm:b-0-collection-long}, $\dist_G(v_{1, 1}, v_{2d, 1}) \le \alpha \cdot 2d$ with probability at most
	 \begin{align*}
	 	\Pr[\dist_G(v_{1, 1}, v_{2d, 1}) \le \alpha \cdot 2d] & \le \frac{(2t)^{2\alpha d}}{w} \\
	 	& = \frac{2^{2\alpha d \cdot (1 + \log t)}}{w}  \\
	 	& \le   \frac{2^{2\alpha d \cdot \left(\frac{3}{2} \cdot \log t\right)}}{w} \tag{since $\log t \ge 2$} \\
	 	& = \frac{2^{2\alpha d \cdot \left(\frac{3}{2} \cdot \frac{1}{4\alpha d(p+2)} \cdot  \log n\right)}}{w} \tag{since $t = n^{\frac{1}{4\alpha (p+2)}}$}  \\
	 	& \le \frac{2^{3/4\log n}}{w} \tag{since $p \ge 1$} \\
	 	& = 2dn^{-1/4} \tag{since $w = n/(2d)$} \\
	 	& \le 2 \log(n) \cdot n^{-1/4} \tag{since $d \le \log n$}
	 \end{align*}
	 Then by the union bound, with probability at least $2/3-2\log(n) \cdot n^{-1/4}$ algorithm $\alg$ succeeds and $\sigma \ge \dist_G(v_{1, 1}, v_{2d, 1}) > \alpha \cdot 2d$. Thus when $b = 0$, our protocol succeeds with probability at least $2/3 - 2\log(n) \cdot n^{-1/4}$. 
	 
	 Since $\ORPPC_{t, d, w}$ is distributionally defined so that $\Pr[b=1]=\Pr[b=0]= \frac{1}{2}$, our protocol  succeeds with probability at least
	 $$
	 \frac{1}{2} \cdot 2/3 + \frac{1}{2} \cdot (2/3 - 2\log(n) \cdot  n^{-1/4}) \ge  2/3 - \log(n) \cdot n^{-1/4},
	 $$   
	 which completes the proof.
\end{proof}

We have completed \Cref{item:part1-or} in our lower bound framework, by translating  semi-streaming algorithms for approximate $s$-$t$ shortest path into  protocols for $\ORPPC$. We  finish the proof of \Cref{thm:main-lb} in the next subsection.

\subsection{Finishing the Proof of \Cref{thm:main-lb}}

In the next step of our lower bound proof, we translate  protocols for $\ORPPC$ into protocols for $\PPC$. This reduction, summarized in \Cref{prop:orppc-ppc}, uses a  direct sum argument combined with message compression techniques. We defer the proof of \Cref{prop:orppc-ppc} to \Cref{sec:dirsum}.

\begin{proposition}\label{prop:orppc-ppc}
For any $t, d, w, s \geq 1$, for any $0 < \delta < 1$, given any protocol $\protOR$ for $\ORPPC_{t, d, w}$ that uses communication at most $s$ bits, and has probability of error at most $\delta$, 
there is a protocol $\protPPC$ for $\PPC_{d, w}$ that has communication at most $\cOR \cdot (d-2) \cdot (s/t + 1)$ bits for some absolute constant $\cOR \geq 1$, and has probability of error at most $\delta  + 1/200$. 
\end{proposition}

The final step in our proof of \Cref{thm:main-lb} is a lower bound on the communication cost of the paired pointer chasing problem. This lower bound is stated in \Cref{lem:PPC-lb} and proved in \Cref{sec:lb-ppc}.

\begin{restatable}{lemma}{ppclemma}
	\label{lem:PPC-lb}
	For any integers $d, w \geq 2$, any deterministic protocol $\prot$ for $\PPC_{d, w}$ that achieves a success probability of at least $0.95$ over the distribution $\mu_{\PPC}$, must have cost $\Omega(w/d^5- \log w)$.
\end{restatable}

We can now put everything together  and  prove \Cref{thm:main-lb}, which is restated below.

\mainlb*
\begin{proof}
	Let $\alpha \ge 1$, and let $p$ be an integer such that $p \le \log(n) /(4\alpha) - 2$. Let $\mathcal{A}$ be a streaming algorithm (deterministic or randomized) on unweighted graphs that, with success probability at least $0.99$, can output an $\alpha$-approximation to the distance between a fixed pair of vertices using at most $s$ space and at most $p$ passes. 
	
	Then by \Cref{lem:psc-stream}, there is a protocol $\pi$ for the $\ORPPC_{t, d, w}$ problem with parameters
	$$t=n^{\frac{1}{4\alpha (p+2)}},  \quad d = p+2, \quad  \text{ and } \quad w = n/(2(p+2)),
	$$
	such that $\pi$ has success probability at least $0.99-\log(n) \cdot n^{-1/4} \ge 0.98$ and communication $s$. 
	
	By applying  \Cref{prop:orppc-ppc}, this implies there is a protocol $\pi_{\PPC}$ for $\PPC_{d, w}$ with success probability at least $0.95$ and communication cost at most $$c_{\text{OR}} \cdot (d-2) \cdot (s/t + 1) \le c_{\text{OR}} \cdot d \cdot (s/t+1)$$
	bits, for some absolute constant $c_{\text{OR}} \ge 1 $. 
	
	Finally, \Cref{lem:PPC-lb} implies that the communication cost of $\pi_{\PPC}$ is at least
	$
	\Omega(w/d^5 - \log w). 
	$
	Combining our upper and lower bounds on the cost of $\pi_{\PPC}$ gives
	$$
	 c_{\text{OR}} \cdot d \cdot (s/t+1) = 	\Omega(w/d^5 - \log w).
	$$
	Solving for $s$ and plugging in the values of $t, d,$ and $w$, we obtain
	$$
	s = \Omega\left( \frac{tw}{d^6} - \frac{t}{d} \log(w) - t \right) = \Omega\left( \frac{n^{1+1/(4\alpha (p+2))}}{2(p+2)^7} - n^{1/(4\alpha (p+2))} \log(n)  \right) = \Omega\left( \frac{n^{1+1/(4\alpha (p+2))}}{p^7} \right).
	$$
	
	We have shown that any streaming algorithm that, with success probability at least $0.99$, can output an $\alpha$-approximation to the distance between a fixed pair of vertices using at most $p$ passes requires at least $\Omega\left( \frac{n^{1+1/(4\alpha (p+2))}}{p^7} \right)$ space. 
	We can extend this lower bound to hold for streaming algorithms that succeed with probability at least $2/3$, by using a standard probability of success amplification argument (running the algorithm in parallel $O(1)$ times and return the median answer).
\end{proof}

	%\clearpage

\subsection{Lower Bound for Paired Pointer Chasing}\label{sec:lb-ppc}

We prove \Cref{lem:PPC-lb}, which is the lower bound for $\PPC$ in this subsection.
We show that we can reduce $\PPC$ to a standard pointer chasing problem through a series of probabilistic arguments. 

We start by defining the standard pointer chasing problem. This communication problem is not defined with an inherent distribution, as we will analyze it on various distributions of the input while proving \Cref{lem:PPC-lb}.

\begin{Definition}[Pointer Chasing]\label{def:standard-pc}
	For integers $w, d \geq 2$, the pointer chasing problem, denoted by $\PC_{d, w}$ is defined as follows. 
	\begin{enumerate}[label=$(\roman*)$]
	\item The input is a graph $G \in \cL_{d, w}$ split between the players.  Alice receives $M_i$ for all odd $i \in [d-1]$ and Bob receives $M_i$ for all even $i \in [d-1]$.
	\item There is an arbitrary partition of vertex set $V_{d}  \subset V$ into sets $W, \overline{W}$ known to both players.
%	\item All the matchings $M_i$ for $i \in [d-1]$ are chosen uniformly at random and independently of each other. 
	\item The players speak for a total of $d-2$ rounds. At the end of all the rounds, the player who receives the last message has to output whether $\point{G}{\ver{1,1}} \in V_d$ belongs to set $W$ or $\overline{W}$. 
\end{enumerate}
%	We use $\cDPPC$ to denote the distribution of the input in $\PPC_{d, w}$ problem. 
\end{Definition}

$\PPC_{d, w}$ problem can be viewed as an instance of $\PC_{d, w^2}$ naturally. We use $\Gprod = (\Vprod, \Eprod)$ to denote the graph $\product{G^1}{G^2}$ where $G^1$ and $G^2$ are graphs in an instance of $\PPC$.  For $i \in [d]$, we use $\Vprod_i$ to denote the partitions of $\Vprod$ into $d$ layers. For $i \in [d-1]$, we use $\Mprod_i$ to denote the perfect matching between layers $\Vprod_i$ and $\Vprod_{i+1}$ of $\Vprod$ in graph $\Gprod$ (see \Cref{obs:product-basic} for properties of the product graph $\Gprod$). 

\begin{observation}\label{obs:PPC-to-PC-split}
	In any instance of $\PPC_{d, w}$, Alice knows all $\Mprod_i$ for odd $i \in [d-1]$ and no other edges. Similarly, Bob knows $\Mprod_i$ for all even $i \in [d-1]$ and no other edges.
\end{observation}

\begin{proof}
	In \Cref{def:prod-2pc}, for any $i \in [d-1]$, we know that $\Mprod_i$ depends only on $M^1_i$ and $M^2_i$, which are the $i^{\textnormal{th}}$ matchings in $G^1$ and $G^2$, respectively. For even $i$, both of these matchings are known to Bob, and neither of them are known to Alice. Likewise, for odd $i$, both of the matchings are known to Alice, whereas Bob knows neither of them.
\end{proof}

\begin{observation}\label{obs:PPC-to-PC-dist}
	In any instance of $\PPC_{d, w}$ with graphs $G^1$ and $G^2$, the distribution is chosen so that, 
	\begin{itemize}
		\item When $\rb = 0$, all the matchings in instances $G^1$ and $G^2$ are chosen uniformly at random and independently of each other. 
		\item When $\rb = 1$, all the matchings in instances $G^1$ and $G^2$ are chosen uniformly at random,  conditioned on $\point{\Gprod}{(\ver{1,1}, \ver{1,1})} = (\ver{d,i},\ver{d,i})$ for some $i\in[w]$. Moreover, the value of $\point{\Gprod}{(\ver{1,1}, \ver{1,1})}$ is uniform over the set $\set{(\ver{d,i}, \ver{d,i}) \mid i \in [w]} \subset \Vprod_d$. 
	\end{itemize}
\end{observation}
\begin{proof}
	The distribution of the matchings in $G^1$ and $G^2$ when $\rb=0$ are exactly as described in \Cref{def:paired-pc}. When $\rb = 1$, we know that in $\PPC_{d, w}$, the two graphs $G^1$ and $G^2$ are chosen so that $\point{G^1}{\ver{1,1}} = \point{G^2}{\ver{1,1}} = \ver{d,j}$ for some $j \in w$. By part $(ii)$ of \Cref{obs:product-basic}, we know that $\point{\Gprod}{(1, 1)} = (\point{G^1}{1}, \point{G^2}{1}) = (\ver{d,j}, \ver{d,j})$ also. For the ``Moreover" part, the distribution of the matchings is symmetric across all elements of the set $\set{(\ver{d,i}, \ver{d,i}) \mid i \in [w]}$, and thus $\point{\Gprod}{(\ver{1,1}, \ver{1,1})}$ is uniform over this set.
\end{proof}

Now, we know that the value of $\rb$ induces a distribution over $\cL_{d, w^2}$, where the two graphs $G^1$ and $G^2$ are sampled so that $\Gprod$ obeys the properties in the statement of \Cref{obs:PPC-to-PC-dist}. By \Cref{obs:PPC-to-PC-split}, we know that the graph $\Gprod$ is split between Alice and Bob exactly the same as if it was given as an instance of $\PC_{d, w^2}$.
However, observe that as of now, there is \textbf{no partition} of $\Vprod_d$ in $\Gprod$ into two sets $W, \overline{W}$ known to both Alice and Bob, as in \Cref{def:standard-pc}. They have to find the \textbf{value of $\bm{b}$}, and not where the final pointer $\point{\Gprod}{(1,1)}$ is. We will choose sets $W, \overline{W}$, so that finding the value of $\rb$ is (almost) equivalent to finding where $\point{\Gprod}{(1,1)}$ lies.

Let $\muPC$ denote the distribution of the random variables in input of $\PC_{d, w^2}$ induced by distribution $\mu_{\PPC}$, along with that of bit $\rb$. It is sufficient to show that outputting the value of $\rb$ is hard in $\PC_{d, w^2}$ under distribution $\mu_{\PC}$.
We work with another distribution $\mutilPC$ (again over graphs from $\PC_{d, w^2}$ and bit $\rb$). Distribution $\mutilPC$ is ``close'' to $\mu_{\PC}$, which will be useful in our analysis.

\begin{Distribution}\label{dist:muPC}
	\textbf{Distribution $\mutilPC$}:
	\begin{enumerate}[label=$(\roman*)$]
		\item Let set $S = \{(\ver{d,i}, \ver{d,i}) \mid i \in [w]\}$. Choose set $T$ of $w$ elements uniformly at random from the set  $\Vprod_d \setminus S$. 
		\item Sample graphs $G^1$ and $G^2$ uniformly at random conditioned on $\point{\Gprod}{(\ver{1,1}, \ver{1,1})} \in S \cup T$. 
		\item When $\point{\Gprod}{(\ver{1,1}, \ver{1,1})} \in S$, fix $\rb = 1$, and set $\rb = 0$ otherwise. 
	\end{enumerate}
\end{Distribution}

We will show that the distribution $\muPC$ and $\mutilPC$ are close to each other in total variation distance. %We need some terminology before we proceed. 

%For set $S \cup T$, we use $\leftset$ to denote the set $\{x \in V^1_d \mid \exists y \in V^2_d, (x,y) \in S \cup T\}$. That is, this is the subset of elements $x \in V^1_d$ such that there is some element in $S \cup T$ with $x$ as one of its co-ordinates. Similarly, $\rightset $ is the set $\{y \in V^2_d \mid \exists x \in V^1_d, (x,y) \in S \cup T\}$. 

\begin{claim}\label{clm:dist-uniform-ST}
In distribution $\mutilPC$, for any choice of set $T$, the value of $\point{\Gprod}{(\ver{1,1}, \ver{1,1})}$ is uniform over the set $S \cup T$.
\end{claim}
\begin{proof}
	
Firstly, we know that when the matchings in $G^1$ and $G^2$ are chosen uniformly at random, the values of $\point{G^1}{\ver{1,1}}$ and $\point{G^2}{\ver{1,1}}$ are uniform over the set $V^1_d$ and $V^2_d$ respectively by symmetry. Moreover, when the matchings are chosen independently of each other, the value of $\point{\Gprod}{(\ver{1,1}, \ver{1,1})}$ is uniform over the set $\Vprod_d= V^1_d \times V^2_d$, by part $(ii)$ of \Cref{obs:product-basic}. 
Therefore, 
\begin{equation}\label{eq:ScupT}
\Pr[\point{\Gprod}{(\ver{1,1}, \ver{1,1})} \in S \cup T] = \frac{\card{S \cup T}}{w^2}.
\end{equation}
	
Let $u = (x,y) \in S\cup T$ be any arbitrary vertex in $\Vprod_d$ with $x \in V_d^1$ and $y \in V_d^2$. We show that,
\begin{equation}
	\Pr[\point{\Gprod}{(\ver{1,1}, \ver{1,1})} = u = (x,y) \mid \point{\Gprod}{(\ver{1,1}, \ver{1,1})} \in S \cup T] = \frac1{\card{S \cup T}}, 
\end{equation}
which is sufficient to prove the claim.

We have, 
\begin{align*}
	&\Pr[\point{\Gprod}{(\ver{1,1}, \ver{1,1})} = u \mid \point{\Gprod}{(\ver{1,1}, \ver{1,1})} \in S \cup T]  \\
	&= \frac1{\Pr[\point{\Gprod}{(\ver{1,1}, \ver{1,1})} \in S \cup T]}\cdot \Pr[\point{\Gprod}{(\ver{1,1}, \ver{1,1})} = u\cap \point{\Gprod}{(\ver{1,1}, \ver{1,1})} \in S \cup T]  \tag{by law of total probability}\\
	&= \frac{w^2}{\card{S \cup T}}\cdot \Pr[\point{\Gprod}{(\ver{1,1}, \ver{1,1})} = u = (x,y)] \tag{by \Cref{eq:ScupT}}\\
	&= \frac{w^2}{\card{S \cup T}}\cdot \frac1{w^2} \tag{as $\point{\Gprod}{(\ver{1,1}, \ver{1,1})}$ is uniform over $\Vprod_d$} \\
	&= \frac1{\card{S \cup T}},
\end{align*}
which completes the proof. 
\end{proof}

%Let $\cE_1$ denote the event that $\rb = 0$ and $\point{\Gprod}{(\ver{1,1}, \ver{1,1})} = (x,x)$ for some $x \in \Vprod_d$. This is a bad event that we do not want to happen, so we bound the probability that it does. 
%
%\begin{observation}\label{clm:E1-bad-event}
%	In distribution $\muPC$, probability that event $\cE_1$ happens is $1/2w$. 
%\end{observation} 
%\begin{proof}
%	We have, 
%	\begin{align*}
%		\Pr[\cE_1] &= \Pr[b = 0 \cap \point{\Gprod}{(\ver{1,1}, \ver{1,1})} = (x,x) \in S ] \\
%		&= \Pr[b = 0] \cdot \Pr[\point{\Gprod}{(\ver{1,1}, \ver{1,1})} = (x,x) \in S \mid b = 0 ] \\
%		&= \frac12 \cdot  \Pr[\point{\Gprod}{(\ver{1,1}, \ver{1,1})} = (x,x) \in S \mid b = 0 ] \tag{as $\rb$ is uniformly random over $\{0,1\}$ in $\muPC$} \\
%		&= \frac12 \cdot \card{S} \cdot \frac1{w^2} \tag{as matchings are uniformly random and independent in $G^1, G^2$ when $\rb = 0$} \\
%		&= 1/2w, \tag{as $\card{S} = w$}
%	\end{align*} 
%	finishing the proof.
%\end{proof}

\begin{claim}\label{clm:dist-mu-mutil}
	Distributions $\muPC$ and $\mutilPC$ are close to each other in total variation distance. 
	\[
		\tvd{\muPC}{\mutilPC} \leq 1/2w.
	\]
\end{claim}

\begin{proof}
The random variables involved in distributions $\muPC$ and $\mutilPC$ are the bit $\rb$, and all the matchings in $G^1, G^2$ which lead to $\Gprod \in \cL_{d,w^2}$. First, we argue that the distribution of $\rb$ is uniform in $\muPC$ and $\mutilPC$, and then we show that for each value of $\rb$, the distribution of the matchings are close to each other. Finally, we use the weak chain rule of total variation distance in \Cref{fact:tvd-chain-rule} to finish the proof. 
	
	\noindent
\textbf{Distribution of $\rb$:}
In distribution $\muPC$, the value of $\rb$ is uniform over $\{0,1\}$ by construction. In $\mutilPC$, we know that the value of $\point{\Gprod}{(\ver{1,1}, \ver{1,1})}$ is uniform over $S \cup T$ by \Cref{clm:dist-uniform-ST}. As $\card{S}= \card{T} = w$ by construction, the value of bit $\rb$ is chosen uniformly at random from $\{0,1\}$ in $\mutilPC$. 

\noindent
\textbf{When $\rb = 1$:} In distribution $\muPC$, when $\rb = 1$, we know that all matchings are chosen uniformly at random conditioned on $\point{\Gprod}{(\ver{1,1}, \ver{1,1})} \in S$. We will argue that this is true in $\mutilPC$ also. 

In $\mutilPC$, we know that matchings are chosen uniformly random conditioned on $\point{\Gprod}{(\ver{1,1}, \ver{1,1})}$ lying in $S \cup T$. However, when $\rb = 1$, we know that $\point{\Gprod}{(\ver{1,1}, \ver{1,1})}$ lies only in $S$ by definition of $\mutilPC$. Hence, the matchings are chosen uniformly at random conditioned on $\point{\Gprod}{(\ver{1,1}, \ver{1,1})} \in S$.	
	
\noindent
\textbf{When $\rb = 0$:}
In distribution $\muPC$, when $\rb = 0$, all the matchings are chosen uniformly at random and independently of each other. 
Here, we have, 
\begin{equation}\label{eq:S-inter-1}
	\Pr[\point{\Gprod}{(\ver{1,1}, \ver{1,1})} \in S \mid b = 0] = \frac1{w^2} \cdot \card{S} = \frac1{w}, 
\end{equation}
by part $(ii)$ of \Cref{obs:product-basic}, and by the matchings being uniformly random and independent in the graphs $G^1$ and $G^2$. 

In distribution $\mutilPC$, conditioned on $\rb = 0$, we know that the matchings are chosen uniformly at random and independent of each other conditioned on $\point{\Gprod}{(\ver{1,1}, \ver{1,1})} \in T$ for some uniformly random set $T$ of $w$ elements from $(V^1_d \times V^2_d) \setminus S$.  By the uniformity of $T$ over set $\Vprod_d \setminus S$, this is the same as sampling all matchings uniformly at random conditioned on $\point{\Gprod}{(\ver{1,1}, \ver{1,1})} \notin S$.

We have argued that the distribution $\mutilPC \mid b = 0$, is the same as distribution $\muPC \mid b  = 0$ conditioned on $\point{\Gprod}{(\ver{1,1}, \ver{1,1})} \notin S$. 

Therefore, by \Cref{fact:tvd-small-event}, we have,
\begin{align}\label{eq:ST-inter-1}
	\tvd{(\muPC \mid b = 0)}{(\mutilPC \mid b = 0)} \leq \Pr[\point{\Gprod}{(\ver{1,1}, \ver{1,1})} \in S \mid b = 0] = \frac1w,
\end{align}
where the last equality follows by \Cref{eq:S-inter-1}.

Now, we can complete the proof by \Cref{fact:tvd-chain-rule}.
\begin{align*}
	\tvd{\muPC}{\mutilPC} &\leq \tvd{\muPC(\rb)}{\mutilPC(\rb)} + \Exp_{\rb = b} \tvd{(\muPC \mid \rb= b)}{(\mutilPC \mid \rb = b)} \\
	&= 0 + \Exp_{\rb = b} \tvd{(\muPC \mid \rb= b)}{(\mutilPC \mid \rb = b)}  \tag{as $\rb$ is uniform over $\{0,1\}$ in both $\muPC$ and $\mutilPC$} \\
	& = \frac12 \cdot (\tvd{(\muPC \mid b = 1)}{(\mutilPC \mid b = 1)} + \tvd{(\muPC \mid b = 0)}{(\mutilPC \mid b = 0)}) \tag{again, as $\rb$ is uniform} \\
	&= \frac12 \cdot (0 +\tvd{(\muPC \mid b = 0)}{(\mutilPC \mid b = 0)} ) \tag{as we have argued, distributions are same conditioned on $\rb = 1$} \\
	&\leq \frac1{2w},
\end{align*}
where the last inequality follows by \Cref{eq:ST-inter-1}.
\end{proof}

By \Cref{clm:dist-mu-mutil}, as the distributions are close to each other, it is sufficient to prove that finding the value of $\rb$ is hard under distribution $\mutilPC$. 
We have made some progress from $\muPC$ in getting closer to standard pointer chasing: now, we have two sets $S$ and $T$ such that when $\rb = 0$, the final pointer is in $S$ and when $\rb = 1$, the pointer is in $T$. When the set $\Vprod_d$ is restricted to $S \cup T$, these sets $S$ and $T$ take the role of $W$ and $\overline{W}$ in \Cref{def:standard-pc} of standard pointer chasing. We continue to work with distribution $\mutilPC$.

The final pointer we are interested in is distributed uniformly at random by \Cref{clm:dist-uniform-ST}, however
the distribution of the specific matchings in $\mutilPC$ are far from uniform. They are from the product of the matchings in $G^1$ and $G^2$. In the next step of this subsection, we show that we can pick subsets of the vertex set $\Vprod$ in graph $\Gprod$ sampled from $\mutilPC$ so that these matchings become uniformly random and independent of each other.

\begin{definition}\label{def:rook-set}
	A subset $P \subseteq \Vprod_d$ is said to be a \textbf{rook set} if for any $(x, y) \in P$, $(x, y') \notin P$ for any $y \neq y'$ and $(x', y) \notin P$ for any $x' \notin P$.\footnote{\label{foot:rook-set}The name rook set has been chosen because when visualized as a grid of $w \times w$ squares, at most one element can be chosen from each row or column, similar to a non-attacking rook placement on a chess board.}
\end{definition}

For any subset $P \subset \Vprod_d$, for $i \in [d-1]$, we use $\Vprod_i(P)$ to denote the set of $\card{P}$ elements in $\Vprod_i$ that have a path to some vertex in $P$ in graph $\Gprod$. For the last layer, $\Vprod_d(P) $ will be equal to $P$.

\begin{claim}\label{clm:other-rook-sets}
	For any rook set $P$, for any $i \in [d]$, $\Vprod_i(P)$ is also a rook set.
\end{claim}
\begin{proof}
	The proof is by induction on $i$, ranging from $d$ and decreasing to $1$. The base case is when $i = d$, and $\Vprod_d(P) = P$. The base case holds trivially as $P$ is a rook set. 
	
	Let us assume towards a contradiction that $\Vprod_{i-1}(P)$ is not a rook set while $\Vprod_i(P)$ is a rook set. There exist two elements $(x, y)$ and $(x', y)$ in $\Vprod_{i-1}(P)$. We consider the two elements in $\Vprod_i(P)$ that $(x,y)$ and $(x',y)$ are mapped to by the matching $\Mprod_{i-1}$. Let these two elements be $(u_1, v_1)$ and $(u_2, v_2)$. But, we know that $v_1 = v_2$, as these are the vertices from layer $V^2_i$ of $G^2$ that was used to construct $\Gprod$. This contradicts our assumption that $\Vprod_i(P)$ is a rook set. The case when some $(x,y)$ and $(x,y')$ belongs to $\Vprod_{i-1}(P)$ can be analyzed similarly.
\end{proof}

\begin{claim}\label{clm:unif-matching}
	In distribution $\mutilPC$, for any rook set $P$, 
conditioned on $\point{\Gprod}{(\ver{1,1}, \ver{1,1})} \in P$, and all choices of $\Vprod_i(P)$ for $1 \leq i \leq d$, for all $1 \leq j \leq d-1$, the subset of $\Mprod_j$ between sets $\Vprod_j(P)$ and $\Vprod_{j+1}(P)$ is a uniformly random permutation independent of the choices of all other matchings, and all edges of $\Mprod_j$ not incident to $\Vprod_j(P)$ and $\Vprod_{j+1}(P)$. 
\end{claim}
\begin{proof}
First, we observe that $\point{\Gprod}{(\ver{1,1}, \ver{1,1})}$ is uniform over the set $P$. This argument is the same as that of \Cref{clm:dist-uniform-ST}.

We will show that the matching edges between $\Vprod_{i-1}(P)$ and $\Vprod_i(P)$ are uniformly random conditioned on any choice of all other matchings, and all edges of $\Mprod_{i-1}$ not incident on $\Vprod_{i-1}(P), \Vprod_i(P)$ for any $i \in [d-1]$. Henceforth in this proof, we assume that all these other random variables are fixed to be something arbitrary.

Next, we observe that the condition of $\point{\Gprod}{(\ver{1,1}, \ver{1,1})} \in P$ is satisfied by the conditioning on these random variables already, as the vertex $(\ver{1,1}, \ver{1,1}) \in \Vprod_1$ belongs to the set $\Vprod_1(P)$, and hence has a path to some vertex in $P$. Therefore, now, in distribution $\mutilPC$, all the matchings $M^1_j$ and $M^2_j$ for $j \in[d-1]$ are chosen uniformly at random and independently of each other, conditioned on the choices of $\Vprod_j(P)$ for all $j \in [d]$ by definition. This implies that the matchings $M^1_{i-1}$ and $M^2_{i-1}$ are independent of the choices of all the other matchings, and the edges not incident to $\Vprod_{i-1}(P), \Vprod_i(P)$. 

It remains to argue that when the edges of the two matchings $M^1_{i-1} , M^2_{i-1}$ incident on $\Vprod_{i-1}(P)$, $\Vprod_i(P)$ are uniformly random conditioned on the choices of the two sets $\Vprod_{i-1}(P)$, $\Vprod_i(P)$, the respective edges in $\Mprod_{i-1}$ are uniformly random as well. 

Let us take any two arbitrary vertices $(x,y) \in \Vprod_{i-1}(P)$ and $(u,v) \in \Vprod_i(P)$. We will show that the probability of $(x,y)$ being mapped to $(u,v)$ is exactly $1/\card{\Vprod_i(P)}$, which is sufficient to prove the claim. 

Let $k = \card{\Vprod_i(P)} = \card{\Vprod_{i-1}(P)}$. 
First, we analyze what the probability is that the entirety of set $\Vprod_{i-1}(P)$ is mapped to $\Vprod_i(P)$. 
\begin{align*}
	&\Pr[\Mprod_{i-1} \textnormal{ maps $\Vprod_{i-1}(P)$ to $\Vprod_i(P)$}] \\
	& = \textnormal{Number of ways to arrange $\Vprod_{i}(P)$}\\
	&\hspace{15mm} \cdot \textnormal{Probability of getting a specific mapping from $\Vprod_{i-1}(P)$ to $\Vprod_{i}(P)$} \\
	&= k! \cdot \paren{\frac{(n-k)!}{n!}}^2, 
\end{align*}
as there are $k!$ ways of arranging set $\Vprod_i(P)$, and for any specific ordering, $k$ of the edges in both matchings $M^1_{i-1}$ and $M^2_{i-1}$ are fixed. The rest of the edges in both matchings remain fully random.

We also have, 
\begin{align*}
	&\Pr[(x,y) \textnormal{ is mapped to }(u,v) \textnormal{ and }\Mprod_{i-1}(P) \textnormal{ maps $\Vprod_{i-1}(P)$ to $\Vprod_i(P)$}] \\
	& =\Pr[(x,y) \textnormal{ is mapped to }(u,v)] \\
	& \hspace{15mm} \cdot \Pr[\Mprod_{i-1}(P) \textnormal{ maps $\Vprod_{i-1}(P)$ to $\Vprod_i(P)$} \mid (x,y) \textnormal{ is mapped to }(u,v)] \\
	&= \frac1{n} \cdot \frac1{n} \cdot \Pr[\Mprod_{i-1}(P) \textnormal{ maps $\Vprod_{i-1}(P)$ to $\Vprod_i(P)$} \mid (x,y) \textnormal{ is mapped to }(u,v)] \tag{as $M^1_{i-1}$ maps $x$ to $u$ and $M^2_{i-1}$ maps $y$ to $v$}\\
	& = \frac1{n^2} \cdot \Pr[\Mprod_{i-1}(P) \textnormal{ maps $\Vprod_{i-1}(P)$ to $\Vprod_i(P)$} \mid (x,y) \textnormal{ is mapped to }(u,v)] \\
	& = \frac1{n^2} \cdot (k-1)! \cdot \paren{\frac{(n-k)!}{(n-1)!}}^2, 
\end{align*}
where for the last equality, we know that conditioned on $(x,y)$ being mapped to $(u,v)$, the rest of the $n-1$ edges of both $M^1_{i-1}$ and $M^2_{i-1}$ are fully random. The probability that the vertices of $\Vprod_{i-1}(P) \setminus \{(x,y)\}$ are mapped to $\Vprod_i(P) \setminus \{(u,v)\}$ will be the number of ways of arranging the vertices of $\Vprod_i(P) \setminus \{(u,v)\}$ times the probability of a specific mapping from $\Vprod_{i-1}(P) \setminus \{(x,y)\}$ to $\Vprod_i(P) \setminus \{(u,v)\}$ being realized. 

We can continue the proof as follows. 
\begin{align*}
	&\Pr[(x,y) \textnormal{ is mapped to }(u,v) \mid \Mprod_{i-1}(P) \textnormal{ maps $\Vprod_{i-1}(P)$ to $\Vprod_i(P)$}] \\
	&= \frac{\Pr[(x,y) \textnormal{ is mapped to }(u,v) \textnormal{ and }\Mprod_{i-1}(P) \textnormal{ maps $\Vprod_{i-1}(P)$ to $\Vprod_i(P)$}]}{\Pr[\Mprod_{i-1} \textnormal{ maps $\Vprod_{i-1}(P)$ to $\Vprod_i(P)$}]} \\
	&= \frac1{n^2} \cdot (k-1)! \cdot \paren{\frac{(n-k)!}{(n-1)!}}^2 \cdot \frac1{k!} \cdot \paren{\frac{n!}{(n-k)!}}^2 \tag{from our equations above} \\
	&= \frac1{k} = \frac1{\card{\Vprod_i(P)}},
\end{align*}
which completes the proof. 
\end{proof}

We are very close to the standard pointer chasing problem over the uniform distribution with independent matchings. 
We will choose a rook set $P$ which is a subset of $S \cup T$ of large size, so that conditioned on the choices of $\Vprod_i(P)$ for all $i \in [d-1]$, the distribution of the matchings in $\PC_{d, w^2}$ is uniformly random. 

We need the following standard results about size of the largest matching in random bipartite graphs. The proof of this result is given in \Cref{app:random-graph-matching} for completeness. 

\begin{proposition}\label{prop:random-graph-matching}
For any large enough integer $k \geq 0$, in a random bipartite graph $G = (L \sqcup R, E)$ where $L= R= \{1,2, \ldots, k\}$, and $E$ is $k$ edges chosen uniformly at random without repetition from the set \[
E \subset L \times R \setminus \{(i,i) \mid i \in [k]\},
\]
with probability at least $1-1/k^2$, there exists a matching in $G$ of size at least $0.1k$. 
\end{proposition}

The next claim follows almost directly from \Cref{prop:random-graph-matching}.

\begin{claim}\label{clm:choose-rook-set}
	In any set $T$ of $w$ elements chosen unifomly at random from $[w] \times [w] \setminus S$ as in \Cref{dist:muPC}, there is a rook set $\Ttil \subset T$ with $\card{\Ttil} = 0.1w$ with probability at least $1-1/w^2$ over the randomness of sampling set $T$. 
\end{claim}
\begin{proof}
	We will look at a bipartite graph with vertex set $L = R= [w]$. 
	The randomness in sampling set $T$ from $[w] \times [w] \setminus S$ is the same as sampling a random set of $w$ edges from the set \[
		L \times R \setminus \{(i,i) \mid i \in [w]\}.
	\]
	We apply \Cref{prop:random-graph-matching} with $k = w$. 
	It is sufficient to show that matchings in our graph correspond to rook sets. Let $M \subset E$ be a matching in the graph $G$. 
	Assume towards a contradiction that $M \subset L \times R = [w] \times [w]$ is not a rook set. Then there exists some $(x,y), (x,y') \in M$ with $x,y,y' \in [w]$ and $y \neq y'$. However, the vertex $x \in L$ now has two edges incident on it, namely $(x,y)$ and $(x,y')$, contradicting the fact that $M$ is a matching. 
	
	The case when some $(x,y), (x',y) \in M$ with $x, x',y \in [w]$ and $x\neq x'$ is handled similarly, where vertex $y \in R$ will have more than one edge incident on it in $M$.
\end{proof}

%Let $\cE_2$ be the event that set $\Ttil$ does not exist with the properties from \Cref{clm:choose-rook-set}. We know that $\cE_2$ happens with probability at most $1/w^2$, and hence, we condition on $\neg \cE_2$ henceforth. 

In the next claim, we show that we can add to $\Ttil$ from set $S$ and create a larger rook set. 
\begin{claim}\label{clm:choose-Stil}
Let $\Ttil \subset T$ be the set from \Cref{clm:choose-rook-set}. It is possible to choose $\Stil \subset S$ with $\card{\Stil} = 0.1w$ such that $\Stil \cup \Ttil$ is a rook set of size $0.2w$.  
\end{claim}
\begin{proof}
We know that $S$ is a rook set, as $S   = \set{(i,i) \mid i \in [w]}$. Let \[
S' = \{(i,i) \mid (i,x) \in \Ttil \textnormal{ or } (y,i) \in \Ttil, \textnormal{ for } x,y \in [w]\}. 
\]
We observe that $\Ttil \cup (S \setminus S')$ is a rook set as all the elements $(x,x), (y,y)$ which interfere with some $(x,y) \in \Ttil$ have been removed. Moreover, the size of $S'$ is at most $0.2w$, as for each element in $\Ttil$, at most two elements are added to $S'$. Hence we can choose $\Stil$ as any arbitrary subset of $S \setminus S'$ of size $0.1w$. 
\end{proof}

\begin{lemma}\label{lem:final-reduction-unif}
For large integers $d, w \geq 100$, there exists a choice of set $T$ such that the following is true: 
	\begin{enumerate}[label=$(\roman*)$]
		\item $T$ has a rook set $\Ttil$ of size $0.1w$. 
	\item A set $\Stil \subset S$ can be chosen such that $\card{\Stil} = 0.1w$, and $\Stil \cup \Ttil$ is a rook set of size $0.2w$. 
		\item Any deterministic protocol $\prot$ which outputs the value of $\rb$ under distribution $\muPC$ with probability at least $0.95$ must output whether $\point{\Gprod}{(1,1)} \in \Stil$ or $\point{\Gprod}{(1,1)} \in \Ttil$, in distribution $\mutilPC $ conditioned on choice of $T$, $\Stil, \Ttil$ and $\Vprod_i(\Stil \cup \Ttil)$ for all $i \in [d-1]$ and event $\point{\Gprod}{(1,1)} \in \Stil \cup \Ttil$ with probability at least $0.7$.
		\item For $i,j \in [d-1]$, the subset of $\Mprod_i$ between vertices of $\Vprod_i(\Stil \cup \Ttil)$ and $\Vprod_{i+1}(\Stil \cup \Ttil)$ is uniformly random and independent of other $\Mprod_j$ for $j \neq i$. 
	\end{enumerate}
\end{lemma}

\begin{proof}
Let $\prot$ be a deterministic protocol that finds the value of $\rb$ with probability at least $2/3$ when the input distribution is $\muPC$.
First, we have, 
\begin{align*}
	&\Pr[\prot \textnormal{ succeeds on }\mutilPC] \\
	&\geq \Pr[\prot \textnormal{ succeeds on }\muPC] - \tvd{\muPC}{\mutilPC } \tag{by \Cref{fact:tvd-small}}\\
	&\geq  0.95  - 1/2w \tag{by \Cref{clm:dist-mu-mutil}}.
\end{align*}

We have argued that protocol $\prot$ outputs the value of $\rb$ with probability at least $0.95-1/2w$ when the input distribution is $\mutilPC$.
We can also charge the probability that the random set $T$ has a large rook set or not to the probability of success. 
\begin{align}
	&\Pr[\prot \textnormal{ succeeds on } \mutilPC \mid \rT \textnormal{ has a large rook set }] \\
	&\geq \Pr[\prot \textnormal{ succeeds on } \mutilPC]- \Pr[\rT \textnormal{ does not have a large rook set }] \tag{by law of total probability}\\
	&\geq 0.95-1/2w-1/w^2 \tag{by \Cref{clm:choose-rook-set}} \\
	&\geq 0.95 - 1/w,\label{eq:interim-12}
\end{align} 
where the last inequality holds for large $w$.
We know that the success of protocol $\prot$ under distribution $\mutilPC $ conditioned on $\rT$ having a large rook set  is on average over the randomness of variable $\rT$. Therefore, there must exist some specific choice of set $T$ such that conditioned on $\rT = T$, the protocol $\prot$ succeeds with probability at least $0.95-1/w$.  We fix $\rT$ to be this set for the statement of the lemma. 

We know that we can pick $\Ttil \subset T$ such that the statement of \Cref{clm:choose-rook-set} holds for set $\Ttil$, as we have conditioned on $T$ having a large rook set. This satisfies part $(i)$ of the lemma statement. 

We choose set $\Stil$ from \Cref{clm:choose-Stil}, based on our choice of $\Ttil$. Part $(ii)$ of the lemma statement is satisfied by \Cref{clm:choose-Stil}. For the rest of the proof, in distribution $\mutilPC $, we assume that the choice of $T, \Ttil$ and $\Stil$ are fixed as we have described.

For part $(iii)$, we begin by analyzing the probability that $\point{\Gprod}{(1,1)} $ belongs to $ \Stil \cup \Ttil$ in distribution $\mutilPC \mid T, \Stil, \Ttil$. By \Cref{clm:dist-uniform-ST}, we have, 
\begin{align}\label{eq:scupt-prob}
	\Pr[\point{\Gprod}{(1,1)} \in \Stil \cup \Ttil] = \frac{\card{\Stil \cup \Ttil}}{\card{S \cup T}}= 0.2.
\end{align}

We can write,
\begin{align*}
	 &\Pr[\prot \textnormal{ succeeds on }\mutilPC \mid T, \Stil, \Ttil] \\
	 &\leq  \Pr[\point{\Gprod}{(1,1)} \in \Stil \cup \Ttil \mid T, \Stil, \Ttil] \cdot \Pr[\prot \textnormal{ succeeds on }\mutilPC \mid \Stil, \Ttil, T, \point{\Gprod}{(1,1)} \in \Stil \cup \Ttil ] \\
	 &\hspace{2cm}  + \Pr[\point{\Gprod}{(1,1)} \notin \Stil \cup \Ttil \mid T, \Ttil, \Stil] \tag{by law of total probability} \\
	 &= 0.2 \cdot  \Pr[\point{\Gprod}{(1,1)} \in \Stil \cup \Ttil \mid S, \Ttil, \Stil] + 0. 8. \tag{by \Cref{eq:scupt-prob}}
\end{align*}
Hence, 
\begin{align*}
	\Pr[\prot \textnormal{ succeeds on }\mutilPC \mid \Stil, \Ttil, T, \point{\Gprod}{(1,1)} \in \Stil \cup \Ttil ] \\
	\geq (1/0.2) \cdot (0.95-1/w-0.8) \geq 0.7,
\end{align*}
by our lower bound on the probability in \Cref{eq:interim-12} and for $w \geq 100$. This probability is taken on average over choices of $\Vprod_i(\Stil \cup \Ttil)$ for $i \in [d-1]$, so there must exist a choice of these sets $\Vprod_i(\Stil \cup \Ttil)$ for all $i \in [d-1]$, so that conditioned on these choices, protocol $\prot$ still succeeds with probability $0.7$. These choices are fixed for part $(iii)$ of the lemma statement. 

We know from definition of distribution $\mutilPC$ that $\rb=0$ if and only if $\point{\Gprod}{(1,1)} \in T$ and $\rb = 1$ if and only if $\point{\Gprod}{(1,1)} \in S$. 
This proves part $(iii)$ of the lemma, as success of protocol $\prot$, when it finds the value of $\rb$ correctly, also identifies whether $\point{\Gprod}{(1,1)}$ belongs to $\Stil$ or $\Ttil$. 

Part $(iv)$ of the lemma follows directly from $\Stil \cup \Ttil$ being a rook set and \Cref{clm:unif-matching}.
\end{proof}

We have thus obtained a protocol that finds whether the final pointer $\point{\Gprod}{(1,1)}$ belongs to $\Stil$ or $\Ttil$ when the matchings are chosen uniformly at random and independently. We are done with the proof of \Cref{lem:PPC-lb} if we can prove a lower bound for $\PC_{d,w}$ under the uniform distribution. 

Let $\nuPC$ be the input distribution to $\PC$ where all the matchings of given $G$ are chosen uniformly at random and independently of each other.
When the input distribution is $\nuPC$, we have the following.

\begin{lemma}\label{lem:unif-pc}
	For any integers $k, \ell \geq 2$, 
	any deterministic protocol which outputs the answer for $\PC$ under uniform distribution $\nuPC$ with probability of success at least $0.7$ must have communication at least $\Omega(\ell/k^5 -\log \ell)$.
\end{lemma}

The difference between standard pointer chasing and our version is that the functions are permutations, and therefore the independence between the values taken by each element of the domain is lost. This version of pointer chasing has appeared in prior works before, \cite{AssadiKSY20,AssadiN21}, although not explicitly in the format we want, so we prove it in \Cref{app:unif-pc} for completeness. The techniques are largely borrowed from \cite{AssadiN21} in how they handle permutations.

We can easily prove \Cref{lem:PPC-lb} using \Cref{lem:final-reduction-unif} and \Cref{lem:unif-pc}.
\begin{proof}[Proof of \Cref{lem:PPC-lb}]
We know from part $(iii)$ that protocol $\prot$ succeeds on the following instance of $\PC_{d, 0.2w}$:
\begin{itemize}
	\item The vertex layers are $\Vprod_i(\Stil \cup \Ttil)$ for $i \in [d]$. 
	\item The final pointer belongs to set $\Vprod_d(\Stil \cup \Ttil) = \Stil \cup \Ttil$.
	\item The vertex partition in the final layer is $\Stil$ and $\Ttil$. 
	\item In each layer, the edges are perfect matchings $\Mprod_i$ from $\Vprod_i$ to $\Vprod_{i+1}$ for $i \in [d-1]$.
\end{itemize}
The matchings are chosen uniformly at random and independently from part $(iv)$ of \Cref{lem:final-reduction-unif}. This is the same as distribution $\nu$ for $\PC_{k,\ell}$ with $k = d$ and $\ell = 0.2 w$. By \Cref{lem:unif-pc}, we get directly that protocol $\prot$ must use communication at least $\Omega(\ell/k^5 - \log \ell) = \Omega(w/d^5 - \log w)$. 
\end{proof}

%\clearpage

%	\bigskip

\section*{Acknowledgments} 
We would like to thank Adrian Vladu for helpful discussions on the multiplicative weight update method. We are also thankful to Greg Bodwin for his support throughout this project.
%	\bigskip
	
	\bibliographystyle{alpha}
	\bibliography{general}
	
%	\bigskip
	
	%\clearpage
	\appendix
	
	\section{Further Extensions of~\Cref{thm:upper-main}}\label{app:extension}

In this appendix, we present  extensions of \Cref{thm:upper-main} to deterministic algorithms and dynamic streams. 
We assume prior knowledge of \Cref{alg:s_t} and the proofs in  \Cref{sec:upper_rand}.

\subsection{Derandomizing~\Cref{thm:upper-main}} 
\label{app:deterministic}

We now present our derandomized version of \Cref{thm:upper-main}.

\begin{theorem} \label{thm:alg-det}
	Let $G = (V, E, w)$ be an $n$-vertex graph  with non-negative edge weights presented in an insertion-only stream. Given a source vertex $s \in V$ and   parameters $k \in [\ln n]$ and $ \eps \in (0, 1)$, there is a 
	deterministic streaming algorithm that 
	 outputs a $(1+\eps)$-approximate  shortest path tree rooted at source $s$ using
	$$
	O\left( \frac{k^2}{\eps^2} \cdot  n^{1+1/k} \log^4 n \right) \text{ space} \quad \text{and} \quad O\left( \frac{k^2}{\eps}  \right) \text{ passes}.
	$$
\end{theorem}

We begin with a high-level overview of our derandomization strategy. The only  randomness in \Cref{alg:s_t} is from randomly sampling edge set $F^{(r)}$ in each round $r \in [R]$. The random edges in $F^{(r)}$ are needed to guarantee  \Cref{lem:sampling-lemma}, which roughly states that the potential function $Q^{(r)}$ does not increase too quickly, with high probability.

Our strategy is to replace the random edges $F^{(r)}$ with a certain graph sparsifier that implies \Cref{lem:sampling-lemma} 
 and can be computed deterministically. In the proof of \Cref{lem:sampling-lemma},  for any subgraph $\mathcal{T} \subseteq G$ of $G$, we defined the set of ``bad'' edges $B(\mathcal{T}) \subseteq E$ of $\mathcal{T}$ as:
\begin{equation} \label{eq:bad}
 B(\mathcal{T}) := \{(u, v) \in E \text{ s.t. } |\dist_{\mathcal{T}}(s, u) - \dist_{\mathcal{T}}(s, v)| > w(u, v) \},
\end{equation}
 where $s \in V$ is the source vertex for our shortest path problem instance. 
 The role of random edges $F^{(r)}$ in the proof of \Cref{lem:sampling-lemma} is to guarantee that the shortest path tree $\mathcal{T}^{(r)}$ in round $r \in [R]$ does not have many bad edges in  $B(\mathcal{T})$ of high importance. With the goal of replacing the randomly sampled edges in \Cref{alg:s_t}, we now define and study  
 \textit{smoothness  sparsifiers}.\footnote{The name `smoothness sparsifier' is inspired by the notion of $\alpha$-smooth distance estimates in \cite{rozhovn2023parallel}.} 

\subsection*{Smoothness Sparsifiers} 
%In this section, we define smoothness sparsifiers and present a streaming algorithm that computes a smoothness sparsifier in one pass and $\widetilde{O}_{\eps}(n)$ space. 
Throughout this section, we will use the concept of the set of bad edges with respect to a subgraph $\mathcal{T}$, as defined in \Cref{eq:bad}. We define a smoothness sparsifier as follows.

\begin{Definition}[Smoothness Sparsifier] \label{def:smooth} Let $G = (V, E, w)$ be an undirected weighted graph with an associated set of edge importances $\mathcal{Q} = \{q_e \in \mathbb{R}^+ \mid e \in E \}$. Let $s \in V$ be a source vertex, and let $\eps \in (0, 1)$. 
Let $H \subseteq G$ be a subgraph of $G$ with an associated set of edge importances $\widetilde{\mathcal{Q}} =   \{\widetilde{q}_e \in \mathbb{R}^+ \mid e \in E(H) \}$. Then  $(H, \widetilde{\mathcal{Q}})$ is an $\eps$-smoothness sparsifier of graph $G$ and source $s$ if the following two conditions hold:
\begin{enumerate}
	\item $\sum_{e \in E(H)} \widetilde{q}_e \le (1+\eps) \cdot \sum_{e \in E(G)} q_e$, and
	\item for any  acyclic graph $\mathcal{T} \subseteq G$, 
	$$
	\sum_{e \in B(\mathcal{T}) \cap E(H)} \widetilde{q}_e  \ge \sum_{e \in B(\mathcal{T})} q_e - \eps \cdot \sum_{e \in E(G)} q_e.
	$$
\end{enumerate}
\end{Definition}

We begin our study of smoothness sparsifiers by proving that smoothness sparsifiers of near-linear size exist.

\begin{lemma}[Sparsifier Existence]
	\label{lem:sparsifier}
	Let $G$ be an $n$-vertex undirected graph, let $s \in V$, and let $\eps \in (0, 1)$. Then for any choice of edge importances $\mathcal{Q}$, 	
	there exists a $\eps$-smoothness sparsifier $(H, \widetilde{\mathcal{Q}})$ of graph $G$ and source $s$ with $|E(H)| = O(\eps^{-2} n \log n)$ edges.  
\end{lemma}
\begin{proof}
	We will prove this existence result using the probabilistic method.
	Let $\mathcal{Q} = \{q_e \in \mathbb{R}^+ \mid e \in E(G)\}$ be the edge importances of $G$, and let $Q = \sum_{e \in E(G)} q_e$. 
	 We will sample the edge set $E(H)$ by picking each edge $e \in E(G)$ independently and with probability
	\begin{equation}
		p_e := 10\eps^{-2} \cdot  n \log n  \cdot \frac{q_e}{Q}.
	\end{equation} 
	This sampling procedure gives us our subgraph $H \subseteq G$. We also need to define our edge importances $\widetilde{\mathcal{Q}} = \{\widetilde{q}_e \in \mathbb{R}^+ \mid e \in E(H)\}$. For each $e \in E(H)$ we have sampled into $H$, we let $\widetilde{q}_e = \frac{q_e}{p_e}$. For our probability analysis later, we let $\widetilde{q}_e := 0$  for edges $e \in E(G) - E(H)$ that are not sampled into $H$. 
	
	We now verify that $(H, \widetilde{\mathcal{Q}})$ satisfies the smoothness sparsifier properties of \Cref{def:smooth}. We will need the Chernoff bound for independent bounded random variables, which states that given independent random variables $X_1, \dots, X_m$ taking values in the interval $[0, b]$, 
	$$
	\Pr[|X - \mu| \ge \delta\mu] \le 2e^{-\delta^2\mu / (3b)}, 
	$$
	where $X = \sum_{i=1}^mX_i$, $\mu = \mathbb{E}[X]$, and $\delta \in (0, 1)$. 
	
	Let $\widetilde{Q} := \sum_{e \in E(G)} \widetilde{q}_e$.
	Note that $\widetilde{Q} = \sum_{e \in E(H)} \widetilde{q}_e$, since $\widetilde{q}_e = 0$ for $e \not \in E(H)$. 
	 We observe $$\mathbb{E}\left[\widetilde{Q}\right] = \sum_{e \in E(G)} \widetilde{q}_e =  \sum_{e \in E(G)} p_e \cdot \frac{q_e}{p_e} = Q.$$
	Additionally, we observe that each independent random variable $\widetilde{q}_e$ takes value either $0$ or 
	\begin{equation}
		\label{eq:q_e_weight}
		\widetilde{q}_e := \frac{q_e}{p_e} = \frac{\eps^2 \cdot Q}{10  n \log n },
	\end{equation}
	 where $Q = \sum_{e \in E(G)}q_e$. 	Then applying Chernoff with $\delta = \eps$, $\mu = Q$, and $b = q_e/p_e$, we compute
	\begin{align*}
		\Pr[\widetilde{Q} \ge (1+\eps)Q] \le 2e^{-\eps^2 \cdot Q / (3q_e/p_e) } = 2e^{ -\eps^2 \cdot Q \cdot \frac{ 10 n \log n }{3 \cdot \eps^2 \cdot Q}} \le 2e^{-3 n \log n }.
	\end{align*}
We conclude that $$\sum_{e \in E(H)}\widetilde{q}_e = \widetilde{Q} \le (1+\eps) Q = (1+\eps) \cdot \sum_{e \in E(G)} q_e$$ with exponentially  high probability, so our $\eps$-smoothness sparsifier satisfies the first property of \Cref{def:smooth}.

We now prove that our sparsifier satisfies the second property of \Cref{def:smooth}, which will also follow from the Chernoff bound. Let $\mathcal{T} \subseteq G$ be an acyclic graph. We define the sum $$\mu := \sum_{e \in B(\mathcal{T})}q_e,$$ and we define the random variable $$X := \sum_{e \in E(\mathcal{T})}\widetilde{q}_e$$ 
as the sum over all random variables $\widetilde{q}_e$, where $e \in E(\mathcal{T})$. We observe that
$$
\mathbb{E}[X] = \sum_{e \in E(\mathcal{T})} \widetilde{q}_e =  \sum_{e \in E(\mathcal{T})} p_e \cdot \frac{q_e}{p_e} = \mu.
$$
Additionally, we observe that $\mu \le Q$, since $\mu = \sum_{e \in B( \mathcal{T})} q_e \le \sum_{e \in E(G)} q_e = Q$. 
Now applying Chernoff with $\delta := \eps \cdot Q / \mu$ and $b = q_e/p_e$, we compute
\begin{align*}
\Pr[X \le \mu - \eps \cdot Q] & = \Pr\left[X \le \left(1-\eps \cdot \frac{Q}{\mu}\right) \cdot \mu \right] \\
& \le 2e^{-\eps^2 \cdot \frac{Q^2}{\mu^2} \cdot  \mu \cdot \frac{p_e}{3q_e}} \\
& \le 2e^{-\eps^2 \cdot Q \cdot \frac{p_e}{3q_e}} \tag{since $\mu \le Q$} \\
& = 2e^{ -\eps^2 \cdot Q \cdot \frac{ 10 n \log n }{3 \cdot \eps^2 \cdot Q}} \\
& \le 2e^{-3 n \log n }.
\end{align*}
We conclude that $$ \sum_{e \in B(\mathcal{T} ) \cap E(H )} \widetilde{q}_e =   X \ge \mu - \eps \cdot Q = \sum_{e \in B(\mathcal{T})} q_e - \eps \cdot \sum_{e \in E(G)} q_e$$ with   probability at least $2e^{-3n \log n}$. By \Cref{fact:treeform}, there are at most $n^{n-1} \cdot 2^n \le 2^{n + n \log n}$ acyclic subgraphs $\mathcal{T}$ of $G$. Then by a union bound argument, for every acyclic subgraph $\mathcal{T} \subseteq G$, 
	$$
\sum_{e \in B(\mathcal{T}) \cap E(H)} \widetilde{q}_e  \ge \sum_{e \in B(\mathcal{T})} q_e - \eps \cdot \sum_{e \in E(G)} q_e,
$$
with exponentially high probability.

Finally, by one more application of the Chernoff bound, 
$$
|E(H)| \le 2 \mathbb{E}\left[ |E(H)| \right] = 2\sum_{e \in E(G)} p_e = 20 \eps^{-2} \cdot n \log n  \cdot \sum_{e \in E(G)} \frac{q_e}{Q} = 20 \eps^{-2} \cdot n \log n
$$
with exponentially high probability. 

We have shown that $(H, \widetilde{\mathcal{Q}})$ satisfies the smoothness sparsifier properties of \Cref{def:smooth} with exponentially high  probability, confirming the existence of sparse $\eps$-smoothness sparsifiers. 
\end{proof}

The construction in \Cref{lem:sparsifier} is probabilistic in nature. In order to make our construction deterministic, we will use the standard ``merge and reduce'' framework in streaming algorithms (see  \cite{McGregor14} for a similar application of this framework). In order to apply this framework, we need to prove that our smoothness sparsifiers are well-behaved under graph unions and compositions. 

\begin{claim}[Mergeable]
	\label{clm:union}
	Suppose that graph $G_1 = (V, E_1, w_1)$ with edge importances $\mathcal{Q}_1$ has an $\eps$-smoothness sparsifier $(H_1, \widetilde{\mathcal{Q}}_1)$. Likewise, suppose that graph $G_2 = (V, E_2, w_2)$ with edge importances $\mathcal{Q}_2$ has an $\eps$-smoothness sparsifier $(H_2, \widetilde{\mathcal{Q}}_2)$. If $E_1 \cap E_2 = \emptyset$, then
	$(H_1 \cup H_2, \widetilde{\mathcal{Q}}_1 \cup \widetilde{\mathcal{Q}}_2)$ is an $\eps$-smoothness sparsifier of graph $G_1 \cup G_2$ with edge importances $\mathcal{Q}_1 \cup \mathcal{Q}_2$. 
\end{claim}
\begin{proof}
	$H_1 \cup H_2$ is a subgraph of $G_1 \cup G_2$, and every edge $e \in E(H_1 \cup H_2)$ has an importance attribute $\widetilde{q}_e$ in $\widetilde{\mathcal{Q}}_1 \cup \widetilde{\mathcal{Q}}_2$. Likewise, every edge $e \in E(G_1 \cup G_2)$ has an importance attribute $q_e$ in $\mathcal{Q}_1 \cup \mathcal{Q}_2$. 
	We can verify the first property of smoothness sparsifiers as follows: 
	\begin{align*}
		\sum_{e \in E(H_1 \cup H_2)} \widetilde{q}_e & = \sum_{e \in E(H_1)} \widetilde{q}_e + \sum_{e \in E(H_2)} \widetilde{q}_e \\
		& \le (1+\eps) \cdot \sum_{e \in E(G_1)} q_e + (1+\eps) \cdot \sum_{e \in E(G_2)} q_e \tag{ by \Cref{def:smooth}} \\
		& = (1+\eps) \cdot \sum_{e \in E(G_1 \cup G_2)} q_e.
 	\end{align*}
 	Likewise, we can verify the second property of smoothness sparsifiers as follows. Let $\mathcal{T} \subseteq G_1 \cup G_2$ be an acyclic subgraph of  $G_1 \cup G_2$.
 	\begin{align*}
 		\sum_{e \in B(\mathcal{T}) \cap E(H_1 \cup H_2)} \widetilde{q}_e  & =   \sum_{e \in B(\mathcal{T}) \cap E(H_1 )} \widetilde{q}_e + \sum_{e \in B(\mathcal{T}) \cap E( H_2)} \widetilde{q}_e \\
 		& \ge \left( \sum_{e \in B(\mathcal{T}) \cap E(G_1) }q_e - \eps \cdot \sum_{e \in E(G_1)} q_e \right) + \left( \sum_{e \in B(\mathcal{T}) \cap E(G_2) }q_e - \eps \cdot \sum_{e \in E(G_2)} q_e \right) \tag{by applying \Cref{def:smooth} on the acyclic subgraphs $\mathcal{T} \cap G_1 \subseteq G_1$ and $\mathcal{T} \cap G_2 \subseteq G_2$  } \\
 		& = \sum_{e \in B(\mathcal{T})} q_e - \eps \cdot \sum_{e \in E(G_1 \cup G_2)} q_e. 
 	\end{align*}
 %	\vspace{-12pt}
\end{proof}

\begin{claim}[Composable]
	\label{clm:composition} 
		Suppose that graph $G_1 = (V, E_1, w_1)$ with edge importances $\mathcal{Q}_1$ has an $\eps_1$-smoothness sparsifier $(G_2, \mathcal{Q}_2)$.  Further, suppose that graph $G_2$ with edge importances  $\mathcal{Q}_2$  has an $\eps_2$-smoothness sparsifier $(G_3, \mathcal{Q}_3)$. Then $(G_3, \mathcal{Q}_3)$ is a $(\eps_1 + \eps_2 + \eps_1 \cdot \eps_2)$-smoothness sparsifier of  graph $G_1 = (V, E_1, w_1)$ with edge importances $\mathcal{Q}_1$.
\end{claim}
\begin{proof}
	For $i \in \{1, 2, 3\}$, let graph $G_i$ have edge importances $\mathcal{Q}_i = \{q_e^i \in \mathbb{R}^+ \mid e \in E(G_i)\}$. We can verify that  $(G_3, \mathcal{Q}_3)$ satisfies the first condition of \Cref{def:smooth} as follows:
	\begin{align*}
	\sum_{e \in E(G_3)} q_e^3 \le (1+\eps_2) \cdot \sum_{e \in E(G_2)} q_e^2 \le (1+\eps_1)(1+\eps_2) \sum_{e \in E(G_1)} q_e^1 = (1+\eps_1 + \eps_2 + \eps_1 \eps_2) \cdot \sum_{e \in E(G_1)} q_e^1.
	\end{align*}
We can verify the second condition as follows. Let $\mathcal{T} \subseteq G_1$ be an acyclic graph. 

\begin{align*}
	\sum_{e \in B(\mathcal{T}) \cap E(G_3)} q_e^3 & \ge \sum_{e \in B(\mathcal{T}) \cap E(G_2)}q_e^2 - \eps_2 \cdot \sum_{e \in E(G_2)} q_e^2 \tag{since $(G_3, \mathcal{Q}_3)$ is an $\eps_2$-smoothness sparsifier of $(G_2, \mathcal{Q}_2)$}  \\
	& \ge \left(     \sum_{e \in B(\mathcal{T})} q_e^1 - \eps_1 \cdot \sum_{e \in E(G_1)} q_e^1    \right) - \eps_2 \cdot \sum_{e \in E(G_2)} q_e^2 \tag{since $(G_2, \mathcal{Q}_2)$ is an $\eps_1$-smoothness sparsifier of $(G_1, \mathcal{Q}_1)$} \\
	& \ge \left(     \sum_{e \in B(\mathcal{T})} q_e^1 - \eps_1 \cdot \sum_{e \in E(G_1)} q_e^1    \right) - \eps_2 \cdot (1+\eps_1) \cdot  \sum_{e \in E(G_1)} q_e^1 \tag{since $(G_2, \mathcal{Q}_2)$ is an $\eps_1$-smoothness sparsifier of $(G_1, \mathcal{Q}_1)$}
	\\
	& \ge \sum_{e \in B(\mathcal{T})} q_e^1 - (\eps_1 + \eps_2 + \eps_1 \cdot \eps_2) \cdot \sum_{e \in E(G_1)} q_e^1.  
\end{align*}
\end{proof}

In addition to proofs that smoothness sparsifiers are mergeable and composable, we will need an additional technical claim that essentially says that we can merge and compose smoothness sparsifiers deterministically.

\begin{claim} \label{clm:deterministic-clm}
	Suppose  we are given as input a  graph $G_1$ with edge importances $\mathcal{Q}_1$, and graph $G_2$ with edge importances $\mathcal{Q}_2$, and $E(G_1) \cap E(G_2) = \emptyset$. 
 Consider the graph $G_3 := G_1 \cup G_2$ with edge importances $\mathcal{Q}_3 := \mathcal{Q}_1 \cup \mathcal{Q}_2$.  Then we can  \textit{deterministically}  compute and explicitly store  $\eps$-smoothness sparsifier  $(G_4, \mathcal{Q}_4)$
	of graph $G_3$ and edge importances $\mathcal{Q}_3$	of size at most $|E(G_4)| = O(\eps^{-2}n \log n )$, using at most $O(|E(G_3)| + \eps^{-2}n \log n)$ words of space.
\end{claim}
\begin{proof}
We can compute and store graph $G_3$ and importances $\mathcal{Q}_3$ deterministically in $O(|E(G_3)|)$ space. To deterministically compute and store an $\eps$-smoothness sparsifier of $G_3, \mathcal{Q}_3$, we will do the following:
\begin{itemize}
	\item Enumerate every subgraph $H$ of $G_3$ one-by-one. We can enumerate the  subgraphs $H$ of $G_3$  in $O(|E(G_3)|)$ space by maintaining a bit string $b$ of length $|E(G_3)|$ such that the $i$th edge of $G_3$ is in $H$ if and only if $b[i]=1$. 
	\item Given a subgraph $H$ of $G_3$ in our enumeration, if $|E(H)| = O(\eps^{-2} n \log n )$, then we assign every edge $e \in E(H)$ an importance value  $$q_e^H :=  \frac{\eps^{2} }{10 n \log n } \cdot \sum_{q \in \mathcal{Q}_3} q,$$
	and we define $\mathcal{Q}_H := \{q_e^H \in \mathbb{R}^+ \mid e \in E(H)\}$. 
	\item Check if $(H, \mathcal{Q}_H)$ is an $\eps$-smoothness sparsifier of $G_3$ and $\mathcal{Q}_3$  as defined in \Cref{def:smooth}. This can be done in $O(\eps^{-2}n  \log n)$ space by enumerating every acyclic subgraph $\mathcal{T} \subseteq H$ and checking the inequalities in \Cref{def:smooth}.
	If $(H, \mathcal{Q}_H)$ is an $\eps$-smoothness sparsifier of $G_3$ and $\mathcal{Q}_3$, then we output $(H, \mathcal{Q}_H)$ and terminate the algorithm. Otherwise, we continue enumerating subgraphs $H$ of $G_3$. 
	
	\item We claim that at least one of the subgraphs $H$ with edge importances $\mathcal{Q}_H$ that we enumerate will be a valid $\eps$-smoothness sparsifier of graph $G_3$ and edge importances $\mathcal{Q}_3$. This follows from our proof of the existence of $\eps$-smoothness sparsifiers in \Cref{lem:sparsifier}. Specifically, in \Cref{lem:sparsifier} we show that every $n$-vertex graph $G$ with edge importances $\mathcal{Q}$ admits an $\eps$-smoothness sparsifier $(H, \mathcal{Q}_H)$ such that 
	\begin{itemize}
	\item $|E(H)| = O(\eps^{-2} n \log n )$, and
	\item every edge $e \in E(H)$ has an edge importance value $q_e^H \in \mathcal{Q}_H$ defined to be  exactly $$q_e^H = \frac{\eps^{2} }{10  n \log n } \cdot \sum_{q \in Q} q$$
	(see \Cref{eq:q_e_weight} in the proof of \Cref{lem:sparsifier} for reference). 
	\end{itemize}
	Then this algorithm will eventually output the  $\eps$-smoothness sparsifier  $(H, \mathcal{Q}_H)$  of graph $G_3$ and edge importances  $\mathcal{Q}_3$ that we proved existed in \Cref{lem:sparsifier}. 
\end{itemize}
\end{proof}

We can now give a deterministic streaming algorithm to compute our sparsifiers, using the merge and reduce technique. 

\begin{lemma}[Deterministic Construction] \label{lem:det_sparsifier}
	Let $G = (V, E, w)$ be  an $n$-vertex undirected graph with edge importances $\mathcal{Q}$. Suppose that the edges of $G$, along with their associated edge weights and edge importances, are arriving one-by-one in a stream. Then there exists a   deterministic streaming algorithm to compute a $\eps$-smoothness sparsifier $(H, \widetilde{\mathcal{Q}})$ of $G$ in one pass and $O(\eps^{-2} n \log^4 n)$ space. 
\end{lemma}
\begin{proof}
 By \Cref{lem:sparsifier}, every $n$-vertex weighted graph and corresponding set of edge importances admits an $\eps$-smoothness sparsifier of size $\phi(\eps) = O(\eps^{-2}n \log n)$ edges. Additionally, by \Cref{clm:union} and \Cref{clm:composition}, $\eps$-smoothness  sparsifiers can be merged with no loss, and can be composed with only a constant factor loss in $\eps$.
 
 The merge and reduce technique partitions the stream hierarchically. We will compute an $\eps$-smoothness sparsifier on each of the partitions and then we repeatedly merge and compose them, to obtain our final sparsifier. 

 First, partition the input stream into $\lambda = m / \phi(\eps)$ segments of length $\phi(\eps)$ each. For ease of presentation, we assume that $\lambda$ is a power of two. Let $E_{i, 0}$ denote the $i$th segment of the stream, for $i \in [\lambda]$. We define a graph $G_{i, 0} := (V, E_{i, 0})$ for each $i \in [\lambda]$.

 For $j \in [1, \log \lambda]$ and $i \in [1, \lambda / 2^j]$, we recursively define the graph
 	$$
 	G_{i, j} := G_{2i-1, j-1} \cup G_{2i, j-1}.
 	$$
 For each graph $G_{i, j}$,  with  $j \in [1, \log \lambda]$ and $i \in [1, \lambda / 2^j]$, we recursively define an $\eps$-smoothness sparsifier $(H_{i, j}, \mathcal{Q}_{i, j})$, where $H_{i, j} \subseteq G_{i, j}$ and $\mathcal{Q}_{i, j}$ is the associated collection of edge importances. 
 	\begin{itemize}
 		\item For each $i \in [\lambda]$, we let $H_{i, 0} := G_{i, 0}$, and we let $\mathcal{Q}_{i, 0}$ be the edge importances in $\mathcal{Q}$ restricted to edge set $E_{i, 0}$.
 		\item 	For $j \in [1, \log \lambda]$ and $i \in [1, \lambda / 2^j]$, we let $(H_{i, j}, \mathcal{Q}_{i, j})$ be an $\eps$-smoothness sparsifier of the  graph $H_{2i-1, j-1} \cup H_{2i, j-1}$ with associated edge importances $\mathcal{Q}_{ 2i-1, j-1 } \cup \mathcal{Q}_{2i, j-1}$. We maintain the invariant that  subgraph $H_{i, j}$ has $|E(H_{i, j})| \le  \phi(\eps)$ edges. 
 		\item Suppose that in our streaming algorithm, we have graphs  $H_{2i-1, j-1}$ and $H_{2i, j-1}$ with associated edge importances $\mathcal{Q}_{ 2i-1, j-1 }$ and $\mathcal{Q}_{2i, j-1}$ stored locally and explicitly. Then using the algorithm in \Cref{clm:deterministic-clm}, we can locally and  \textit{deterministically}  compute and explicitly store an  $\eps$-smoothness sparsifier  $(H_{i, j}, \mathcal{Q}_{i, j})$ with at most $|E(H_{i, j})| \le \phi(\eps)$ edges, while using at most $O(  \phi(\eps)  )$ words of space. 
 	\end{itemize}

  By our mergeable and composable properties in \Cref{clm:union} and \Cref{clm:composition}, 
 	 every graph $H_{i, j}$, where $j \in [\log \lambda]$ and $i \in [\lambda/2^j]$, is a $\varphi(j)$-smoothness sparsifier of graph $G_{i, j}$, where $\varphi(j)$ is defined recursively as
 	 $$
 	 \varphi(0) :=0 \qquad \text{ and } \qquad \varphi(j) := 
 	  (1+\eps) \cdot \varphi(j-1) + \eps,
 	 $$
 	 since $H_{2i-1, j-1} \cup H_{2i, j-1}$ is a $\varphi(j-1)$-smoothness sparsifier of $G_{2i-1, j-1} \cup G_{2i, j-1}$ by \Cref{clm:union}, and 
 	 $H_{i, j}$ is defined as an $\eps$-smoothness sparsifier of  $H_{2i-1, j-1} \cup H_{2i, j-1}$. Straightforward calculations give us an upper bound of  $\varphi(j) \le  (1+\eps)^{j} - 1$.

In particular, the earlier discussion implies that graph $H_{1, \log \lambda}$ with edge importances $\mathcal{Q}_{1, \log \lambda}$ is a $\left( (1+\eps)^{\log \lambda} - 1 \right)$-smoothness sparsifier of graph $G_{1, \log \lambda}$ with edge importances $\mathcal{Q}_{1, \log \lambda}$. It is straightforward to verify that graph $G_{1, \log \lambda}$  is  our input graph $G$ and importances $\mathcal{Q}_{1, \log \lambda}$ are our input importances  $\mathcal{Q}$. If we reparameterize $\eps$ so that $\eps_{\texttt{new}} \gets \frac{\eps_{\texttt{}}}{10 \log \lambda}$, then $(H_{1, \log \lambda}, \mathcal{Q}_{1, \log \lambda})$ will be a $\varphi$-smoothness sparsifier of $G$ and $\mathcal{Q}$, where   
\begin{align*}
	\varphi & =  (1+\eps_{\texttt{new}})^{\log \lambda} - 1 \\
	 & = 	 \left(1+ \frac{\eps_{\texttt{}}}{10 \log \lambda}\right)^{\log \lambda} - 1 \\
& \le  e^{\eps_{\texttt{}}/10}   -1 \tag{since $1+x \le e^x$} \\
& \le   e^{\eps_{\texttt{}}/2 - \left( \eps_{\texttt{}}/2 \right)^2}  - 1  \tag{since $\eps_{\texttt{}} \in (0, 1)$} \\
& \le  (1+\eps_{\texttt{}}/2) -1 \tag{since $e^{y-y^2} \le 1+y$ when $y \in (0, 1)$} \\
& \le \eps_{\texttt{}}.
\end{align*}

We have shown that if we recursively construct $\eps_{\texttt{new}}$-smoothness sparsifiers $(H_{i, j}, \mathcal{Q}_{i, j})$, then 
$(H_{1, \log \lambda}, \mathcal{Q}_{1, \log \lambda})$ will be an $\eps$-smoothness sparsifier of $G$ and $\mathcal{Q}$.  Additionally, by \Cref{clm:deterministic-clm} we know that we can merge and compose two $\eps_{\texttt{new}}$-smoothness sparsifiers using $O(\phi(\eps_{\texttt{new}})) = O(\eps^{-2} n \log^3 n)$ space. Now notice that for each $j \in [\log \lambda]$, we only need to explicitly store two sparsifiers  at a time because as soon as we have constructed $(H_{i, j}, \mathcal{Q}_{i, j})$, we can forget $(H_{2i-1, j-1}, \mathcal{Q}_{2i-1, j-1})$ and $(H_{2i, j-1}, \mathcal{Q}_{2i, j-1})$. We conclude that we only need to store $O(\log n)$ $\eps_{\texttt{new}}$-smoothness sparsifiers simultaneously at any given time in our construction. 
Then we can compute our $\eps$-smoothness sparsifier of $G$ in one pass and most $O(\eps^{-2} n \log^4 n)$ space. 
\end{proof}

We have shown that smoothness sparsifiers can be deterministically computed efficiently within the space limits of semi-streaming algorithms.

\subsection*{Derandomizing \Cref{alg:s_t} using Smoothness Sparsifiers}

We can obtain our deterministic streaming algorithm by changing only the edge sampling step of  \Cref{alg:s_t}. 
This step will be replaced by a smoothness sparsifier with appropriate parameters, and we will argue that this does not affect the correctness of the algorithm.
Specifically, we replace Step~\ref{alg:3a} of \Cref{alg:s_t} with the following: 

\begin{Algorithm}[Displaying Step 3a only] \label{alg:det}
	\normalfont \hfill
	%\vspace{1mm}
	\begin{itemize}
		\item  Compute a $\left( \frac{\eps}{10k \cdot n^{1/k}} \right) $-smoothness sparsifier $(J, \mathcal{Q})$ of  graph $G$ and edge importances  $\{q_e^{(r)} \in \mathbb{R}^+ \mid e \in E\}$.
		\item  Let  $F^{(r)} = E(J)$.
	\end{itemize} 
\end{Algorithm}

We  let \Cref{alg:det} refer to  the algorithm obtained by taking \Cref{alg:s_t} and replacing its  Step~\ref{alg:3a} with the above procedure. We first prove that \Cref{alg:det} correctly outputs a $(1+\eps)$-approximate shortest path tree. Then we verify that it can be implemented deterministically in the streaming model as required by \Cref{thm:alg-det}. 
 
\paragraph{\Cref{alg:det} Correctness.} 

The correctness proof for this new algorithm will be largely the same as \Cref{alg:s_t}. We first prove that our new algorithm  satisfies the guarantee of  \Cref{lem:sampling-lemma} deterministically.

\begin{lemma}[cf. \Cref{lem:sampling-lemma}]
For every $r \in [R]$, 
$$
Q^{(r+1)} \le \left( 1 + \frac{\eps}{10k}\right) \cdot Q^{(r)}.
$$
\end{lemma}
\begin{proof}
	Note that for every edge $(u, v) \in E(J) = F^{(r)} \subseteq H^{(r)}$,
	$$
	|\dist_{H^{(r)}}(s, v) - \dist_{H^{(r)}}(s, u)| \le w(u, v),
	$$ 
	by the triangle inequality, and any shortest path tree of $H^{(r)}$ will also satisfy this triangle inequality (as $\{s\} \times V$ distances in the tree and $H^{(r)}$ are the same). In particular, this implies that
	\begin{equation*} \label{eq:no_bad_edges}
		B(T^{(r)}) \cap F^{(r)} = \emptyset.
	\end{equation*}
	However, since graph $ J = (V, F^{(r)})$ and edge importances $\mathcal{Q} = \{\widetilde{q}_e \in \mathbb{R}^+ \mid e \in F^{(r)} \}$ form a 
	$\left( \frac{\eps}{10k \cdot n^{1/k}} \right) $-smoothness sparsifier  of  graph $G$ and edge importances  $\{q_e^{(r)} \in \mathbb{R}^+ \mid e \in E\}$, we can apply the second property of smoothness sparsifiers in \Cref{def:smooth} to obtain
	\begin{equation*}
		0 = \sum_{e \in B(T^{(r)}) \cap F^{(r)}} \widetilde{q}_e \ge \sum_{e \in B(T^{(r)})} q_e^{(r)} - \left( \frac{\eps}{10k \cdot n^{1/k}} \right) \cdot \sum_{e \in E(G)} q_e^{(r)},
	\end{equation*}
	where the equality follows from the fact that $B(T^{(r)}) \cap F^{(r)} = \emptyset$. Rearranging, we find that
	\begin{equation} \label{eq:useful_sparsifier}
		\sum_{e \in B(T^{(r)})} q_e^{(r)} \le \left( \frac{\eps}{10k \cdot n^{1/k}} \right) \cdot \sum_{e \in E(G)} q_e^{(r)} =  \frac{\eps}{10k \cdot n^{1/k}} \cdot Q^{(r)}.
	\end{equation}
	Additionally, notice from the definition of $\Qr{r}$ and $\Qr{r+1}$ that
		\begin{align} \label{eq:sparsifier_app}
		\Qr{r+1} &= \sum_{e \in E} \qr{r+1}_e = \sum_{e \notin B(\Tr{r})} \qr{r}_e + \sum_{e \in B(\Tr{r})} (1+n^{1/k}) \cdot \qr{r}_e = \Qr{r} +  n^{1/k} \cdot  \sum_{e \in B(\Tr{r})}\qr{r}_e.
	\end{align}
	Combining inequality \ref{eq:useful_sparsifier} and equation \ref{eq:sparsifier_app}, we get
	$$
	Q^{(r+1)} = Q^{(r)} + n^{1/k} \cdot  \sum_{e \in B(\Tr{r})}\qr{r}_e \le Q^{(r)} + n^{1/k} \cdot \frac{\eps}{10k \cdot n^{1/k}} \cdot Q^{(r)} = \left( 1 + \frac{\eps}{10k}\right) \cdot Q^{(r)}.
	$$
\end{proof}

It is straightforward to verify that the  \Cref{clm:edge_fails} and \Cref{lem:alg-correct} hold deterministically for \Cref{alg:det}, using the same proofs as in \Cref{sec:upper_rand}. We conclude that \Cref{alg:det} correctly returns a $(1+\eps)$-approximate single-source shortest path tree rooted at $s$.

\paragraph{Streaming Implementation of \Cref{alg:det}.} Our streaming implementation of \Cref{alg:det}  exactly follows that of \Cref{alg:s_t} described in \Cref{sec:streaming_implementation}, with the only new detail being the implementation of the new Step 3a in \Cref{alg:det}. 

By \Cref{lem:det_sparsifier}, we can implement Step 3a of \Cref{alg:det} in the streaming model deterministically  using one pass and $O(\eps^{-2} \cdot k^2 n^{1+2/k}\log^4n  )$ space. If we implement the remainder of \Cref{alg:det} in the same way as \Cref{alg:s_t}, then we obtain a deterministic streaming algorithm for $(1+\eps)$-approximate single-source shortest paths using $O(\eps^{-2} \cdot k^2 \cdot n^{1+2/k} \log^{4}n )$ space and $O(k^2/\eps)$ passes. This completes the proof of \Cref{thm:alg-det}.

%This bound can be made to match the bound claimed in \Cref{thm:alg-det} by choosing the algorithm parameter $k$ appropriately. 

\subsection{Extending~\Cref{thm:upper-main} to Dynamic Streams}\label{app:dynamic}

We now present our extension of \Cref{thm:upper-main} to  dynamic graph streams. We start with a formal definition of this model. 

\begin{Definition}[Dynamic Graph Streams]
	
	\label{def:dynamic_stream}
	Given a set of $n$ vertices $V$, a dynamic graph stream $S = \langle s_1, \dots, s_t  \rangle$ is a sequence of edge updates $s_i \in {V \choose 2} \times \{-1, 1\}$ for $i \in [t]$. The dynamic graph stream defines a final graph $G = (V, E)$, where $(u, v) \in E$ if and only if 
	$$
	| \{i \in [t] \mid s_i = (u, v, 1) \} | - | \{i \in [t] \mid s_i = (u, v, -1) \} | = 1.
	$$
	If $(u, v) \not \in E$, then we require that 
		$$
	| \{i \in [t] \mid s_i = (u, v, 1) \} | - | \{i \in [t] \mid s_i = (u, v, -1) \} | = 0.
	$$
\end{Definition}

The dynamic graph stream model can be interpreted as a sequence of edge insertions and deletions that result in a final graph $G$. The goal is to design an algorithm that makes a few passes over the dynamic graph stream $S$ (in the same order) and outputs a solution to our problem on the final graph $G$. 

We prove the following upper bound in the dynamic graph stream model. 

\begin{theorem} \label{thm:dynamic_stream}
	Let $G = (V,E, w)$ be an $n$-vertex graph  with non-negative edge weights presented in a dynamic graph stream. Given a source vertex $s \in V$ and  parameters $k \in [\ln n]$ and $\eps \in (0, 1)$, there exists a randomized streaming algorithm that outputs a $(1+\eps)$-approximate shortest path tree rooted at source $s$ using
	$$
	O(\eps^{-2}\log^2(\eps^{-1}) \cdot kn^{1+1/k} \cdot \log^3 n ) \text{ space } \qquad \text{ and } \qquad O\left( \frac{k^2}{\eps} \right) \text{ passes,}
	$$
	with probability at least $1 - n^{-\Omega(n)}$. 
\end{theorem}

We prove \Cref{thm:dynamic_stream} by adapting \Cref{alg:s_t} to dynamic streams. The main obstacle to extending \Cref{alg:s_t} is that the edge sampling procedure in Step~\ref{alg:3a} of \Cref{alg:s_t} does not immediately extend to  dynamic streams. We resolve this issue by defining a new edge sampling procedure that is easier to implement in dynamic streams. 

Our new edge sampling procedure will make use of $\ell_1$-samplers, which we define next.

\begin{Definition}[$\ell_1$ Sampling] \label{def:sampler}
	An $(\eps, \delta)$ $\ell_1$-sampler for a vector $\vec{x} \ne 0$  in  $\mathbb{R}^n$ is an algorithm that either fails with probability at most $\delta$, or succeeds and outputs a random index $i \in [n]$ with probability in the range
	$$
	\left[ \frac{(1-\eps)|x_i|}{\sum_{i \in [n]} |x_i|}, \quad  \frac{(1+\eps)|x_i|}{\sum_{i \in [n]} |x_i|} \right].
	$$
\end{Definition}

We will use the efficient   $\ell_1$-samplers of \cite{jowhari2011}.

\begin{proposition}[Theorem 1 of \cite{jowhari2011}] \label{thm:samplers}
Given  a vector $\vec{x} \in \mathbb{R}^n$ that is presented in a dynamic stream of addition and subtraction updates to its coordinates, there exists an $(\eps, \delta)$  
	 $\ell_1$-sampler algorithm for $\vec{x}$ using one pass and  $O(\eps^{-1}\log(\eps^{-1}) \log^2(n) \log(\delta^{-1}))$ space.
\end{proposition}

Using $\ell_1$-samplers, we can obtain our algorithm for dynamic streams by changing only the edge sampling step of \Cref{alg:s_t}. Specifically, we replace Step~\ref{alg:3a} of \Cref{alg:s_t} with the following:
\begin{Algorithm}[Displaying Step 3a only] \label{alg:dyn}
	\normalfont \hfill
	\begin{itemize}
		\item Let $\vec{x} \in \mathbb{R}^{\binom{n}{2}}$ be a vector with $\binom{n}{2}$ entries, corresponding to each pair $u, v$ of distince vertices.  For each pair $(u,v)= e \in E$, the coordinate corresponding to edge $e$ has value  $x_e = q_e^{(r)}$, and for $(u,v) \notin E$, $x_e = 0$. 
		\item  Let $\lambda = 20 \eps^{-1} \cdot kn^{1+1/k} \cdot  \log n$.  	For each index $i \in [1, \lambda]$:
		\begin{itemize}
			\item Use an $(\eps/10, \eps/10)$ $\ell_1$-sampler of vector $\vec{x} \in \mathbb{R}^{\binom{n}2}$ to compute a random edge $e_i \in E$. \footnote{Vertex pairs which are not in $E$ are not sampled by \Cref{def:sampler}, as $x_{u,v} = 0$ for $(u,v) \notin E$.} If the $\ell_1$-sampler fails, then let $e_i = \emptyset$. 
		\end{itemize}
		\item Let  $F^{(r)} = \{e_i \in E \mid i \in  [1, \lambda]\}$.
	\end{itemize}
\end{Algorithm}

We refer to \Cref{alg:dyn} as the algorithm obtained by taking \Cref{alg:s_t} and replacing its  Step~\ref{alg:3a} with the above procedure. We first prove that \Cref{alg:dyn} correctly outputs a $(1+\eps)$-approximate shortest path tree. Then we verify that it can be implemented  in the dynamic streaming model as required by \Cref{thm:dynamic_stream}. 

\paragraph{\Cref{alg:dyn} Correctness.} 
The correctness proof for \Cref{alg:dyn} will essentially be the same as that of \Cref{alg:s_t}. The only statement in the proof of \Cref{alg:s_t} that will require a new argument is \Cref{lem:sampling-lemma}. 

\begin{lemma}[cf. \Cref{lem:sampling-lemma}]
For every $r \in [R]$, with probability at least $1-n^{-\Omega(n)}$, 
$$
Q^{(r+1)} \le \left( 1 + \frac{\eps}{10k} \right) \cdot Q^{(r)}.
$$
\end{lemma}
\begin{proof}
	This proof will largely follow the proof of \Cref{lem:sampling-lemma} with some small changes to the probabilistic argument at the end. For completeness, we repeat some of the proof of \Cref{lem:sampling-lemma}. 
	
Fix any $r \in [R]$. 
		Let $\TT_1, \dots, \TT_t$ denote the set of all spanning trees of $G$; by Cayley's tree formula (\Cref{fact:treeform}), we have $t \le n^{n-2}$. 
		For a tree $\TT_j$ for $j \in [t]$, we define the set $B(\TT_j) \subseteq E$ as: 
		\[
		B(\TT_j) := \set{(u,v) \in E ~\text{s.t.}~ \card{\dist_{\TT_j}(s, u) - \dist_{\TT_j}(s, v)} > w(u, v)};
		\]
		namely, these are the ``bad'' edges of $G$ that will have their importances increased in~\Cref{alg:s_t}, \emph{if} the algorithm chooses $\TT_j$ 
		as the shortest path tree $\Tr{r}$ in round $r$. 
		
		We say a tree $\TT_j$ is \textbf{bad} iff 
		\begin{equation} \label{eq:B_j-importance-dyn}
			\sum_{e \in B(\TT_j)} \qr{r}_e > \frac{\eps}{10k \cdot n^{1/k} } \cdot \Qr{r},
		\end{equation}
		namely, the total importances of the edges $B(\TT_j)$ is ``high''. We similarly say that the computed tree $\Tr{r}$ is bad if $\Tr{r} = \TT_j$ for some bad tree $\TT_j$. 
		Notice that if the computed tree $\Tr{r}$ of round $r$ is \emph{not} bad, then we will have 
		\begin{align*}
			\Qr{r+1} &= \sum_{e \in E} \qr{r+1}_e = \sum_{e \notin B(\Tr{r})} \qr{r}_e + \sum_{e \in B(\Tr{r})} (1+n^{1/k}) \cdot \qr{r}_e \\
			&= \sum_{e \in E} \qr{r}_e + \sum_{e \in B(\Tr{r})} n^{1/k} \cdot \qr{r}_e \leq \Qr{r} + \frac{\eps}{10k} \cdot \Qr{r}, 
		\end{align*}
		by~\Cref{eq:B_j-importance-dyn} in the last step for the tree $\Tr{r}$ (which we assumed is not bad). Thus, to establish the claim, we only need to prove that with exponentially high probability, the tree $\Tr{r}$ 
		is not bad. To do so, we bound the probability that any bad tree $\TT_j$ can be a shortest path tree in the graph $\Hr{r}$, using the randomness of the sample $\Fr{r}$, 
		and union bound over all bad trees to conclude the proof. 
		
		Fix a bad tree $\TT_j$. For $\TT_j$ to be a possible shortest path tree in $\Hr{r}$, it should happen that none of the edges in $B(\TT_j)$ are sampled in $\Fr{r}$. This is because 
		for every edge $(u,v) \in \Fr{r} \subseteq \Hr{r}$, 
		\[
		\card{\dist_{\Hr{r}}(s, v) -  \dist_{\Hr{r}}(s, u) } \leq w(u, v),
		\]
		by triangle inequality, and any shortest path tree of $\Hr{r}$ will also have to satisfy this triangle inequality (as $\{s\} \times V$  distances in the tree and $\Hr{r}$ are the same). 
		As such, 
		\begin{align*}
			\Pr\paren{\Tr{r} = \TT_j} &\leq \Pr\paren{\Fr{r} \cap B(\TT_j) = \emptyset} = \prod_{i \in [1, \lambda]}\left( 1 - \Pr\paren{e_i \in B(\mathcal{T}_j)} \right)  \tag{by the independence of the $\ell_1$-samplers} \\
			& \le  \prod_{i \in [1, \lambda]}\left( 1 -  \left(1-\eps/10 \right) \cdot \frac{(1-\eps/10) \cdot \sum_{e \in B(\mathcal{T}_j)} q_e^{(r)} }{Q^{(r)}}   \right)
			\tag{by \Cref{def:sampler} and since edge $e_i$ is sampled using an $(\eps/10, \eps/10)$ $\ell_1$-sampler of $\vec{x}$}
			\\
			& \le  \prod_{i \in [1, \lambda]}\left( 1 -  (1-\eps/10)^2 \cdot \frac{\eps}{10k \cdot n^{1/k}}   \right) 
			\tag{by \Cref{eq:B_j-importance-dyn}} \\
			& \le \exp\paren{- \lambda \cdot \frac{\eps}{20k \cdot n^{1/k}}} 
			\tag{since $1+x \le e^x$ and $\eps \le 1$}
			\\
			&= \exp\paren{-n\log{n}}
			\tag{since $\lambda  = 20\eps^{-1}kn^{1+1/k} \log n$} 
		\end{align*}
		A union bound over at most $t \leq 2^{n\log{n}}$ bad trees implies that 
		\[
		\Pr\paren{\text{$\Tr{r}$ is bad}} \leq \sum_{\text{bad trees $\TT_j$}} \Pr\paren{\Tr{r} = \TT_j} \leq \paren{\frac{2}{e}}^{n\log{n}} =n^{-\Omega(n)},  
		\]
		concluding the proof. 
\end{proof}

It is straightforward to verify that the  \Cref{clm:edge_fails} and \Cref{lem:alg-correct} hold  for \Cref{alg:dyn}, using the same proofs as in \Cref{sec:upper_rand}. We conclude that \Cref{alg:dyn} correctly returns a $(1+\eps)$-approximate single-source shortest path tree rooted at $s$.

\subsection*{Streaming Implementation of \Cref{alg:dyn}}

Our implementation of \Cref{alg:dyn} in dynamic streams will largely follow the streaming implementation of \Cref{alg:s_t} in \Cref{sec:streaming_implementation} with only two small changes. 

First, \Cref{alg:dyn} requires computing and  storing a 
 $2k$-spanner $H$ (where $H$ is  defined explicitly in  \Cref{alg:s_t}). We can achieve this in $O(k)$ passes and $O(n^{1+1/k})$ space in the dynamic streaming model 
using the spanner construction algorithm of  \cite{AhnGM12b}. We summarize this algorithm in the following theorem. 
\begin{proposition}[Section 5 of \cite{AhnGM12b}]
	For any $k \ge 1$, there is an algorithm for computing a $(2k-1)$-spanner of size $O(n^{1+1/k})$ of a weighted graph presented in a dynamic  stream using 
	$$
	O(n^{1+1/k}) \text{ space} \qquad \text{ and } \qquad  O(k) \text{ passes.}
	$$
\end{proposition}

Second, \Cref{alg:dyn} requires implementing $\lambda$ different $(\eps/10, \eps/10)$  $\ell_1$-samplers of the vector $\vec{x} \in \mathbb{R}^E$ in each round $r \in [R]$. Since $\vec{x}$ encodes the edge importances $\{q_e^{(r)}\in \mathbb{R}^+ \mid e \in E \}$ in round $r$ of \Cref{alg:dyn}, we can implement each $\ell_1$-sampler in one pass and $O(\eps^{-1} \log^2(\eps^{-1}) \log^2(n))$ space using the $\ell_1$-samplers of \Cref{thm:samplers} and using the data structure that reports edge importances described in the proof of  \Cref{lem:upper-stream}. Then in each round $r \in [R]$,  our $\ell_1$-samplers use one pass and  at most $$O(\lambda \cdot \eps^{-1} \log^2(\eps^{-1}) \log^2(n)) = O(\eps^{-2} \log^2(\eps^{-1}) \cdot kn^{1/k} \cdot \log^3 n )$$
space total. We conclude that \Cref{alg:dyn} can be implemented in dynamic streams using 
$$ O(\eps^{-2} \log^2(\eps^{-1}) \cdot kn^{1+1/k} \cdot \log^3 n ) \text{ space } \qquad \text{ and } \qquad O(R + k) = O\left( \frac{k^2}{\eps} \right) \text{ passes}.$$
This completes the proof of \Cref{thm:dynamic_stream}.

	\section{Missing Proofs from the Lower Bound}

We give the proofs of the missing steps from \Cref{sec:lb} now.

\newcommand{\ic}{\ensuremath{\textnormal{IC}}}
\newcommand{\rR}{\ensuremath{\rv{R}}}

\newcommand{\bstar}{\ensuremath{b^{\star}}}
\newcommand{\bOR}{\ensuremath{b_{\textnormal{OR}}}}

\newcommand{\rG}{\ensuremath{\rv{G}}}
\newcommand{\rGBob}{\ensuremath{\rv{G}_{\textnormal{Bob}}}}
\newcommand{\rGAlice}{\ensuremath{\rv{G}_{\textnormal{Al}}}}

\newcommand{\ristar}{\ensuremath{\textnormal{\textsf{i}}^{\star}}}
\newcommand{\rRprotOR}{\rR_{\textnormal{OR}}}
\newcommand{\rRprotprime}{\rR_{\prot'}}

\newcommand{\raoconst}{\ensuremath{c_{\textnormal{comp}}}}

\newcommand{\cA}{\ensuremath{\mathcal{A}}}
\newcommand{\cB}{\ensuremath{\mathcal{B}}}

\subsection{Direct Sum for $\ORPPC$: Proof of \Cref{prop:orppc-ppc}}\label{sec:dirsum}

In this subsection we see how to give a protocol for $\PPC_{d, w}$ from \Cref{def:paired-pc} using a protocol for $\ORPPC_{t, d, w}$ from \Cref{def:ORPPC}. 

\begin{restate}[\Cref{prop:orppc-ppc}]
For any $t, d, w, s \geq 1$, for any $0 < \delta < 1$, given any protocol $\protOR$ for $\ORPPC_{t, d, w}$ that uses communication at most $s$ bits, and has probability of error at most $\delta$, 
there is a protocol $\protPPC$ for $\PPC_{d, w}$ that has communication at most $\cOR \cdot (d-2) \cdot (s/t + 1)$ bits for some absolute constant $\cOR \geq 1$, and has probability of error at most $\delta  + 1/200$. 
\end{restate}

\Cref{prop:orppc-ppc} is proved by using an information complexity direct sum, and then using standard message compression techniques to reduce the communication of the resulting protocol. 
We use $\rProt$ to denote the random variables corresponding to all the messages sent by protocol $\protOR$. We use $\rRprotOR$ to denote the public randomness used by $\rProt$. 
%\paragraph{Notation.}

%We use $\rProt_i$ to denote the message sent in round $i$ for $i \in [d-2]$. Note that when $i$ is odd, this is a message sent by Bob, and when $i $ is even, the message is sent by Alice. 

\begin{Definition}\label{def:ext-information}
	For any protocol $\prot$, with input distribution $\mu$ over $\cX \times \cY$, the external information complexity is defined as, 
	\[
		\ic_{\mu}(\prot) = \mi{\rX, \rY}{\rProt \mid \rR}.
	\]
\end{Definition}
In words, multi-party external information cost measures the information revealed by the entire protocol about the entire graph to an \emph{external} observer. 
The following result follows from~\cite{BarakBCR10}. 
\begin{proposition}[cf.~\cite{BarakBCR10}]\label{prop:ic-cc}
	For any multi-party protocol $\prot$ on any input distribution $\mu$, 
	\[
	\ic_{\mu}(\prot)\leq \cc{\prot},
	\]
	where $\cc{\prot}$ is the communication cost of $\prot$.
\end{proposition}
\begin{proof}
	We have, 
	\begin{align*}
		\ic_{\mu}(\prot) \Eq{(1)} \mi{\rG}{\rProt \mid \rR} \Eq{(2)} \en{\rProt \mid \rR} - \en{\rProt \mid \rG,\rR} \Leq{(3)} \en{\rProt \mid \rR} \Leq{(4)} \en{\rProt} \Leq{(5)} \log{\card{\supp{\rProt}}} \Eq{(6)} \cc{\prot}; 
	\end{align*}
	here, $(1)$ is by the definition information cost, $(2)$ is by the definition of mutual information, (3) is by the non-negativity of entropy (\itfacts{uniform}), (4) is because conditioning can only reduce the entropy (\itfacts{cond-reduce}), 
	(5) is because uniform distribution has the highest entropy (\itfacts{uniform}), and (6) is by the definition of worst-case communication cost. 
\end{proof}

We prove an intermediate lemma that reduces gives a protocol for $\PPC_{d, w}$ that has low information complexity from $\protOR$. 

\begin{Algorithm}
	\textbf{Protocol $\prot'$ for $\PPC_{d, w}$ with input $(G^1, G^2)$}:
	\begin{enumerate}[label=$(\roman*)$]
		\item Pick an index $\istar \in [t]$ uniformly at random using public randomness. 
		\item Sample $(G^1_i, G^2_i)$ for all $i < \istar$ from public randomness from distribution $\muPPC^0$. 
		\item Set $(G^1_{\istar}, G^2_{\istar}) = (G^1, G^2)$. 
		\item Alice and Bob sample $(G^1_{i}, G^2_{i})$ for all $\istar < i \leq t$ from $\muPPC^0$ using private randomness. 
		\item Run protocol $\protOR$ on the input of $\ORPPC_{t, d,w}$ and output the answer from this protocol. 
	\end{enumerate}
\end{Algorithm}

Let $\bstar$ denote the bit from $\{0,1\}$ such that the input $(G^1, G^2)$ is sampled from $\muPPC^{\bstar}$, and let $\bOR$ denote the bit from $\{0,1\}$ corresponding to the distribution of the input in $\ORPPC_{t, d, w}$.

\begin{claim}\label{clm:input-dirsum-correct}
	In protocol $\prot'$, the input given to $\protOR$ is sampled exactly the same as in $\ORPPC_{t, d, w}$, and $\prot'$ succeeds with probability at least $1-\delta$. 
\end{claim}
\begin{proof}
	Firstly, it is easy to see that $\muPPC^0$ is a product distribution over the matchings given to Alice and Bob, as all the matchings are sampled independently of each other and uniformly at random. Hence, both Alice and Bob can sample a graph $(G^1_i, G^2_i)$ for $i > \istar$ in step $(iv)$ from $\muPPC^0$. 
	
	When $\bstar = 1$, then the input given to $\ORPPC$ is sampled so that $(G^1_{\istar}, G^2_{\istar})$ is samped from $\muPPC^1$, and all the other inputs $(G^1_i, G^2_i)$ for $i \neq \istar$ are sampled from $\muPPC^0$, exactly as if $\bOR = 1$. This happens with probability $1/2$. 
	
	When $\bstar = 0$, all the graphs in $\ORPPC$ are sampled so that $(G^1_i, G^2_i)$ for $i \in [t]$ is from $\muPPC^0$. This is the same as when $\bOR = 0$, and happens also with probability $1/2$. Therefore, the overall distribution that $\protOR$ is run on protocol $\prot'$ is the same as the input distribution from \Cref{def:ORPPC}. 
	
	Moreover, when $\bstar = 1$, we saw that the input of $\ORPPC_{t, d, w}$ is sampled with $\bOR = 1$, and when $\bstar = 0$, the input is sampled with $\bOR = 0$. Therefore, when $\protOR$ outputs the value of $\bOR$ correctly, the output of $\prot'$ is also correct on instance $(G^1, G^2)$, and protocol $\prot'$ succeeds with probability at least $1-\delta$. 
\end{proof}

Next, we see that the information complexity of $\prot'$ is small when the input distribution is $\muPPC^0$. 
\begin{claim}[Information Complexity Direct Sum]\label{clm:dirsum-output}
\[
	\ic_{\muPPC^0}(\prot') \leq (d-2) \cdot s/t. 
\]
\end{claim}
\begin{proof}
We have,
	\begin{align*}
			\ic_{\muPPC^0}(\prot') &= 	\mi{\rG^1, \rG^2}{\rProt \mid \rR} \tag{by \Cref{def:ext-information}} \\
			&= \mi{\rG^1, \rG^2}{\rProt \mid \ristar, \rG^1_{< \ristar}, \rG^2_{< \ristar}, \rRprotOR} \\ 
			&= \mi{\rG^1_{\ristar}, \rG^2_{\ristar}}{\rProt \mid \ristar, \rG^1_{< \ristar}, \rG^2_{< \ristar}, \rRprotOR} \\
			&=\frac1t \cdot \sum_{i=1}^t \mi{\rG^1_i, \rG^2_i}{\rProt \mid  \rG^1_{< i}, \rG^2_{< i}, \rRprotOR, \ristar= i} \tag{as $\ristar$ is sampled uniformly at random from $[t]$} \\
			&= \frac1t \cdot \sum_{i=1}^t \mi{\rG^1_i, \rG^2_i}{\rProt \mid  \rG^1_{< i}, \rG^2_{< i}, \rRprotOR}, 
	\end{align*}
	where we have used that the event $\ristar = i$ is independent of the joint distribution of $(G^1_1, G^2_1$;$ G^1_2, G^2_2; $ $\ldots $; $G^1_t, G^2_t)$, $\rProt$ and $\rRprotOR$. This is because all the input graphs $(G^1_j, G^2_j)$ are sampled from $\muPPC^0$, the public randomness $\rRprotOR$ is independent of the input, and the messages $\rProt$ are a deterministic function of the input graphs, and $\rRprotOR$.  
	
	We continue as, 
	\begin{align*}
		\ic_{\muPPC^0}(\prot') &=  \frac1t \cdot \sum_{i=1}^t \mi{\rG^1_i, \rG^2_i}{\rProt \mid  \rG^1_{< i}, \rG^2_{< i}, \rRprotOR} \\
		&=\frac1t \cdot \mi{\rG^1_1, \rG^2_1, \ldots, \rG^1_t, \rG^2_t}{\rProt \mid \rRprotOR} \tag{by chain rule of mutual information \itfacts{chain-rule}} \\
		&\leq \frac1t \cdot \en{\rProt} \tag{by \itfacts{info-entropy}} \\
		&\leq \frac1t \cdot (d-2) \cdot s,
	\end{align*}
	which proves the claim. 
\end{proof}

Next, we want to use $\prot'$ to get a protocol $\protPPC$ which has low communication when the input distribution is $\muPPC$. To this end, we need the standard message compression tools which allow us to simulate a protocol with low information using low communication as well. 

The following result is due to~\cite{HarshaJMR07} which was further strengthened slightly in~\cite{BravermanG14}. We follow the textbook presentation in~\cite{RaoY20}. 

\begin{proposition}[c.f.{\cite[Theorem 7.6]{RaoY20}}]\label{prop:msg-compress}
	Suppose Alice knows two distributions $\cA, \cB$ over the same set ${U}$ and Bob only knows $\cB$. Then, there is a protocol for Alice and Bob to sample an element according to $\cA$ 
	by Alice sending a single message of size 
	\begin{align}
		\kl{\cA}{\cB} + 2 \log (1+\kl{\cA}{\cB}) + O(1)  \label{eq:kl-msg-compress}
	\end{align}
	%  \leq \raoconst \cdot \kl{\cA}{\cB}
	
	bits \underline{in expectation}. This protocol has no error. %, for some constant $\raoconst > 0$. 
	%	This is a one round protocol with no error where Alice sends a single message to Bob.
\end{proposition}
\noindent
We note that somewhat weaker bounds on the KL-Divergence in \Cref{eq:kl-msg-compress} already suffice for our purposes. Thus, to simplify the exposition, we use,
\begin{align}
	\kl{\cA}{\cB} + 2 \log (1+\kl{\cA}{\cB}) + O(1)  &\leq 	\kl{\cA}{\cB} + 10 	\kl{\cA}{\cB}  + O(1) \tag{as $\log(1+x) \leq 5 x$ for every $x \geq 0$ } \\
	&\leq \raoconst \cdot \paren{\kl{\cA}{\cB}  + 1},\label{eq:msg-compress-final}
\end{align}
for some \textbf{absolute constant} $\raoconst  \geq 1$ that we use from now on in our proofs.  

We use $\rGBob$ to denote the random variable corresponding to the part of the input graph given to Bob in $\PPC_{d, w}$. We use $\rGAlice$ to denote the part given to Alice. We use $\rRprotprime$ to denote the public randomness of protocol $\prot'$. We use $\rProt'$ to denote the random variable corresponding to all messages sent by $\prot'$, and $\rProt'_r$ for $r \in [d-2]$ to denote the message sent in round $r$. 

We are ready to give the final protocol $\protPPC$. 

\begin{Algorithm}
	\textbf{Protocol $\protPPC$ for $\PPC_{t, w}$ with input $(G^1, G^2)$:}
	\begin{enumerate}[label=$(\roman*)$]
		\item Alice and Bob sample messages and simulate $\prot'$ in all rounds from $1$ to $d-2$ in increasing order. 
		\item For odd $r \in [d-2]$, Alice knows the distribution $\cB$ of random variable $(\rProt'_r \mid \rProt'_{<r}, \rRprotprime)$ and Bob knows the distribution $\cA$ of $(\rProt'_r \mid \rProt'_{<r}, \rGBob, \rRprotprime)$. Bob sends a single message to Alice that allows both of them to sample from the distribution of $(\rProt'_r \mid \rProt'_{<r}, \rGBob, \rRprotprime)$. 
		\item Similarly for even $r \in [d-2]$, Alice sends a single message to Bob so that both of them can jointly sample from $(\rProt'_r \mid \rProt'_{<r}, \rGAlice, \rRprotprime)$. 
		\item If the total length of all messages from round 1 till round $r$ exceeds $100 \raoconst \cdot (d-2) \cdot  (s/t + 1)$ for some $r \in [d-2]$, then stop and output that the answer is $1$. 
		\item Otherwise, output the answer of $\prot'$. 
	\end{enumerate}
\end{Algorithm}

First, we see how we can relate the number of bits in $\protPPC$ to the external information of $\prot'$. 
In the next claim, $\mu$ can be $\muPPC$, or $\muPPC^0$ or $\muPPC^1$, which are the three distributions we work with. 
\begin{claim}\label{clm:comm-ext-info}
	For any input distribution $\mu$ of instances of $\PPC$, the total expected number of bits sent over all rounds in $\protPPC$ is at most $ \raoconst( \ic_{\mu}(\prot') + d-2)$. 
\end{claim}
\begin{proof}
	Using \Cref{prop:msg-compress}, in round $r$ for $r \in [d-2]$, when $r $ is odd, we can say that the total number of bits sent by Bob is, 
	\begin{align*}
		&\Exp_{\rRprotprime, \rProt'_{<r}} \raoconst (\kl{(\rProt'_r \mid \rGBob, \rRprotprime, \rProt'_{<r})}{(\rProt'_r \mid  \rRprotprime, \rProt'_{<r})} + 1) \\
		&=\raoconst( \mi{\rProt'_r}{\rGBob \mid \rRprotprime, \rProt'_{< r}} + 1)\tag{by \Cref{fact:kl-info}}\\
		& \leq \raoconst( \mi{\rProt'_r}{\rGBob, \rGAlice \mid \rRprotprime, \rProt'_{< r}} + 1),
	\end{align*}
	where the last inequality follows from \itfacts{data-processing}. We get a similar equation for even $r$ and when Alice sends a message to Bob. 
	
	Over all rounds $r \in [d-2]$, we have that the total expected number of bits sent is at most,
	\begin{align*}
	 &\sum_{r \in [d-2]}\raoconst(  \mi{\rProt'_r}{\rGBob, \rGAlice \mid \rRprotprime, \rProt'_{< r}} + 1) \\
	 &= \raoconst \cdot ( \mi{\rProt'}{\rGBob, \rGAlice \mid \rRprotprime} + d-2)\tag{by \itfacts{chain-rule}} \\
	 &= \raoconst \cdot (\ic_{\mu}(\prot') + d-2),
	\end{align*}
	where the last equality follows by \Cref{def:ext-information}.
\end{proof}

\begin{claim}\label{clm:final-protPPC}
	Protocol $\protPPC$ for $\PPC_{d, w}$ has communication at most $100\raoconst \cdot (d-2)  \cdot (s/t + 1)$ and probability of error at most $\delta + 1/200$.
\end{claim}

\begin{proof}
	The upper bound on the communication used by $\protPPC$ is evident as the protocol stops if the communciation exceeds our threshold at any round. We will analyse the probability of success of $\protPPC$. Again, we use $\bstar$ to denote the bit that corresponds to the answer that $\protPPC$ has to output. 
	
	Let $\delta_0$ denote the probability by which $\prot'$ outputs the wrong answer when $\bstar = 0$, and $\delta_1$ denote the probability of error conditioned on $\bstar = 1$. We know that the probability of error of $\prot'$ is $1/2 \cdot (\delta_0 + \delta _1) \leq \delta$. 
	
We know that the expected number of bits sent over all rounds in $\protPPC$ conditioned on $\bstar = 0$ is at most $\raoconst \cdot ( \ic_{\muPPC^0}(\prot') + (d-2))$ by \Cref{clm:comm-ext-info}, and we also know this value is at most $\raoconst \cdot (d-2) \cdot (s/t + 1)$, by \Cref{clm:dirsum-output}. Therefore, conditioned on $\bstar = 0$, the probability that the total number of bits sent over all rounds exceeds $100 \raoconst \cdot  (d-2)  \cdot (s/t   + 1)$ is at most $1/100$ by Markov's inequality. 
Conditioned on $\bstar = 0$, the probability of error of $\protPPC$ is at most $\delta_0 + 1/100$, as when the communication is within the threshold, the simulation of $\prot'$ in $\protPPC$ is without error. 

Now let us look at what happens when $\bstar = 1$. Here, we have no guarantee that the simulation of $\prot'$ does not exceed the threshold. However, whenever it exceeds the threshold, $\protPPC$ declares the answer as $1$, and there is no error. When the communication is within threshold, the simulation is without error, and the error is $\delta_1$. 

The total probability of error of $\protPPC$ is at most, \[
	\Pr[\bstar = 0] \cdot (\delta_0 + 1/100) + \Pr[\bstar = 1] \cdot \delta_0 \leq (1/2) \cdot (\delta_0 + \delta_1) + 1/200 = \delta + 1/200,
\]
which proves the claim. 
\end{proof}

Proof of \Cref{prop:orppc-ppc} follows directly from \Cref{clm:final-protPPC} and setting $\cOR = 100 \raoconst$.

\subsection{Matchings in Random Bipartite Graphs: Proof of \Cref{prop:random-graph-matching}}\label{app:random-graph-matching}

In this subsection, we prove \Cref{prop:random-graph-matching}. 

\begin{restate}[\Cref{prop:random-graph-matching}]
For any large enough integer $k \geq 0$, in a random bipartite graph $G = (L \sqcup R, E)$ where $L= R= \{1,2, \ldots, k\}$, and $E$ is $k$ edges chosen uniformly at random without repetition from the set \[
E \subset L \times R \setminus \{(i,i) \mid i \in [k]\},
\]
with probability at least $1-1/k^2$, there exists a matching in $G$ of size at least $0.1k$. 
\end{restate}

\begin{proof}
We analyse the probability that the maximum matching size is at most $k/10$. As we are working with bipartite graphs, by Konig's theorem, we know that the size of a vertex cover is also $k/10$ when the matching is of size $k/10$. For any specific set $A \subset L \cup R$, we can bound the probability that $A$ is a vertex cover of $G$. Then, we do a union bound for all possible choices of $A$. 

For any set $A \subseteq L \cup R$, with $\card{A} \leq ck$, for some constant $0 < c < 1$, we have,
\begin{align*}
&\Pr[A \textnormal{ is a vertex cover}] \\
&= \Pr[(L\setminus A) \times (R \setminus A) \cap E = \Phi] \tag{as there are no edges with both end points out of $A$} \\
&\leq \binom{k^2 - \card{L \setminus A} \cdot \card{R \setminus A}}{k} \cdot \frac1{\binom{k^2-k}{k}} \tag{no vertex pairs are chosen from $(L \setminus A) \times (R\setminus A)$} \\
&\leq \binom{k^2 (1-(1-c)^2)}{k} \cdot \frac1{\binom{k^2-k}{k}} \tag{as $\card{L \setminus A} \geq (1-c)k, \card{R \setminus A} \geq (1-c)k$} \\
&\leq \paren{e\cdot (2c-c^2) \cdot k}^k \cdot \paren{\frac{k}{k^2-k}}^k \tag{as $(p/q)^q \leq \binom{p}{q} \leq (ep/q)^q$ for any $p,q \geq 1$} \\
&\leq (1.1e \cdot (2c-c^2))^k. \tag{as $k^2/(k^2-k) \leq 1.1$ for $k > 11$}
\end{align*}

The total number of choices for set $A$ is, 
\[
\binom{2k}{1} + \binom{2k}{2} + \ldots + \binom{2k}{ck} \leq ck \cdot \binom{2k}{ck} \leq ck \cdot (2e/c)^{ck},
\]
again, as $\binom{2k}{ck} \leq (e\cdot 2k/ck)^{ck}$.

We continue with the union bound. 
\begin{align*}
&\Pr[\textnormal{Largest matching in $G$ has $\leq 0.1k$ edges}] \\
& = \Pr[\textnormal{There exists vertex cover of size $0.1k$}] \\
&\leq \sum_{A: \card{A} \leq 0.1k = ck} \Pr[A \textnormal{ is a vertex cover}] \\
&\leq ck \cdot ((2e/c)^c \cdot 1.1e \cdot (2c-c^2))^k \tag{by our bounds from earlier} \\
&\leq 0.1k \cdot (0.85)^k \tag{for choice of small constant $c = 0.1$} \\
&\leq 1/k^2. \tag{for large values of $k$}
\end{align*}
This completes the proof.
\end{proof}

\subsection{Pointer Chasing on Uniform Permutations : Proof of \Cref{lem:unif-pc}}\label{app:unif-pc}

In this subsection, we proven the lower bound for pointer chasing when the matchings are uniform and independent, stated in \Cref{lem:unif-pc}. 

\begin{restate}[\Cref{lem:unif-pc}]
	For any integers $k, \ell \geq 2$, 
	any deterministic protocol $\prot$ which outputs the answer for $\PC$ under uniform distribution $\nuPC$ with probability of success at least $0.7$ must have communciation at least $\Omega(\ell/k^5 - \log \ell)$.
\end{restate}

\paragraph{Notation.}
We use $\cUell$ to denote the uniform distribution over any set of $\ell$ elements for $\ell \geq 1$.

In any input instance $G = (V, E) \in \cL_{k, \ell}$, which is an instance of $\PC_{k, \ell}$, we use the following notation. 
We use $\rModd$ to denote all $\rM_i$ for odd $i \in [k-1]$ (this is held by Alice), and $\rMeven$ for all $\rM_i$ with even $i \in [k-1]$ (held by Bob).

For any round $i \in [k-2]$, we use $\Prot_i$ to denote the message sent by the player in round $i$. For odd $i$, this is a message sent by Bob, and for even $i$, this is a message sent by Alice. 
We use $\rProt_{\leq i}$ for $0 \leq i \leq k-2$ to denote all the messages sent till round $i$. When $i = 0$, this $\rProt_{\leq 0}$ is empty. We use $\rProtAl_{\leq i}$ and $\rProtBob_{\leq i}$ to denote the subset of messages from $\rProt_{\leq i}$ that are sent by Alice and Bob respectively.

For $i \in [k]$, we use $\pt{i}$ to denote the unique element from vertex set $V_i$ that has a path to $1_1 \in V_1$. For any $i \in [k-1]$, and matching $M_i$, we use $M_i(v)$ to denote the vertex in $V_{i+1}$ that vertex $v \in V_{i}$ is connected to in matching $M_i$. 

For any $0 \leq i \leq k-2$, we define a random variable $\rGam_i$ as follows. 
\[
\rGam_i = (\rProt_1, \rProt_{2}, \ldots, \rProt_{i}, \rpt{1}, \ldots, \rpt{i+1}).
\]
When $i = 0$, $\rGam_0 = (\rpt{1})$ only. This random variable $\rGam_i$ denotes the information gained by the players in $i$ rounds. 

We prove by induction the following lemma. 
\begin{lemma}\label{lem:unif-pc-induction}
	For any $0 \leq r \leq k-2$, for any protocol $\prot$ where the communication cost is $s < \paren{\frac1{10^5} \cdot \frac1{k^5} \cdot \ell - \log \ell}$, we have,
	\[
	\Exp_{\rGam_r = \Gam_r} \tvd{(\rpt{r+2} \mid \rGam_r =\Gam_r)}{\cUell} \leq  0.1 \cdot \frac{r}{k}.
	\]
\end{lemma}

It is easy to see how this proves \Cref{lem:unif-pc}.

\begin{proof}[Proof of \Cref{lem:unif-pc}]
	We only prove the case when $k $ is even. The proof is similar when $k$ is odd.
	
	For even $k$, Alice sends the last message to Bob, who has to output the answer. Bob does not have access to the last matching $\rM_{k-1}$, and only knows $\rMeven$. We assume Bob knows all the random variables in $\rGam_{k-2}$, as this can only increase the probability of success of the protocol. 
	
	Bob has to find whether a vertex from $V_k$ sampled from the distribution of $\rpt{k} \mid \rMeven, \rGam_{k-2}$ belongs to $W$ or $\overline{W}$. 
	Any value of $(\rMeven, \rGam_{k-2})$ fixes the answer that Bob outputs, as protocol $\prot$ is deterministic. For any value $(\Meven, \Gam_{k-2})$ of random variables $(\rMeven, \rGam_{k-2})$, we use $\outB(\Meven, \Gam_{k-2} ) \in \{ W, \overline{W}\}$ to denote the answer Bob outputs at this value. 
	
	We make the following simple observation:
	\begin{observation}\label{obs:inter-1}
	\[
		\rpt{k} \perp \rMeven \mid \rGam_{k-2}.
	\]
	\end{observation}
	\begin{proof}
		First, we know that $\rpt{k}\perp \rMeven \mid \rpt{1}, \rpt{2}, \ldots, \rpt{k-1}$, as $\rpt{k}$ is a function of $\rM_{k-1}$ when conditioned on $\rpt{k-1}$, and $\rM_{k-1}$ is not a part of $\rMeven$.
	The matchings are chosen independently of each other. 
	
	We also have that $\rpt{k} \perp \rMeven \mid \rGam_{k-2}$, as the independence continues to hold even conditioned on all the messages and pointers by the rectangle property of communication protocols. 
	\end{proof}
	
	The probability of success of the protocol is:
	\begin{align*}
		&\Exp_{\rGam_{k-2}, \rMeven} \Pr_{\rpt{k} \mid \rGam_{k-2}, \rMeven}[\rpt{k} \in \outB(\Gam_{k-2}, \Meven)] \\
	&\leq \Exp_{\rGam_{k-2}, \rMeven} \paren{\Pr_{\rpt{k}}[\rpt{k} \in \outB(\Gam_{k-2}, \Meven)] + \tvd{(\rpt{k} \mid \rGam_{k-2}, \rMeven)}{(\rpt{k})}} \tag{by \Cref{fact:tvd-small}} \\
	& = \frac12 +  \Exp_{\rGam_{k-2}, \rMeven} \tvd{(\rpt{k} \mid \rGam_{k-2}, \rMeven)}{(\rpt{k})} \tag{as $\rpt{k}$ is uniform, and $W, \overline{W}$ is an equipartition} \\
	& =  \frac12 +  \Exp_{\rGam_{k-2}} \tvd{(\rpt{k} \mid \rGam_{k-2})}{(\rpt{k})} \tag{by \Cref{obs:inter-1}} \\
	&\leq \frac12 + 0.1 \cdot (k-2)/k \tag{by \Cref{lem:unif-pc-induction}} \\
	&\leq 0.6,
	\end{align*}
which gives a contradiction to the probability of success being at least $0.7$. Therefore, the premise of \Cref{lem:unif-pc-induction} does not hold, and $ s > \paren{\frac1{10^5} \cdot \frac1{k^5} \cdot \ell - \log \ell} = \Omega(\ell/k^5-\log \ell)$. 
\end{proof}

We need an intermediate lemma to prove \Cref{lem:unif-pc-induction}.

\begin{lemma}\label{lem:high-ent-perm}
	For any vertex $v$ chosen uniformly at random from $V_{r+1}$, for any $0 \leq r \leq k-2$, in any protocol $\prot$ with cost $s < \paren{\frac1{10^5} \cdot \frac1{k^5} \cdot \ell - \log \ell}$, we have, 
	\[
	\Exp_{(\rGam_{r-1}, \rProt_r) = (\Gam_{r-1}, \Prot_r)} \Exp_{\randvert} \tvd{(\rM_{r+1}(\randvert) \mid \rGam_{r-1}, \rProt_r)}{\cUell} \leq 0.1 \cdot \frac{1}{k}.
	\]
\end{lemma}

\Cref{lem:high-ent-perm} can be viewed as an analong of Lemma 6.6 of \cite{AssadiN21} for two-party communication. This is the technical part which handles the difference of functions being permutations.

\begin{proof}[Proof of \Cref{lem:high-ent-perm}]
	First, by \Cref{fact:pinskers}, we have, 
	\begin{align}
		&\Exp_{\rGam_{r-1}, \rProt_r} \Exp_{\randvert} \tvd{(\rM_{r+1}(\randvert) \mid \rGam_{r-1}, \rProt_r)}{\cUell} \notag\\
		&\leq\Exp_{\rGam_{r-1}, \rProt_r} \Exp_{\randvert} \sqrt{\frac12 \cdot \kl{(\rM_{r+1}(\randvert) \mid \rGam_{r-1}, \rProt_r)}{\cUell}} \tag{by Pinsker's inequality \Cref{fact:pinskers}}\\
		&\leq \sqrt{  (1/2) \cdot \Exp_{\rGam_{r-1}, \rProt_r} \Exp_{\randvert} (\kl{(\rM_{r+1}(\randvert) \mid \rGam_{r-1}, \rProt_r)}{\cUell})} \tag{by Jensen's inequality and concavity of $\sqrt{\cdot}$}\\
		&=  \sqrt{  (1/2) \cdot \Exp_{\randvert}  \Exp_{\rGam_{r-1}, \rProt_r} (\kl{(\rM_{r+1}(\randvert) \mid \rGam_{r-1}, \rProt_r)}{\cUell})} \tag{by linearity of expectation} \\
		&=   \sqrt{  (1/2\ell) \cdot \sum_{v \in V_r} \Exp_{\rGam_{r-1}, \rProt_r} (\kl{(\rM_{r+1}(v) \mid \rGam_{r-1}, \rProt_r)}{\cUell})} \tag{as $\randvert$ is chosen uniformly at random} \\
		&= \sqrt{  (1/2\ell) \cdot \sum_{v \in V_r}  \mi{\rM_{r+1}(v)}{\rGam_{r-1},\rProt_r}} \label{eq:inter-1}.
	\end{align}
	where the last line follows from \Cref{fact:kl-info}, as $\rM_{r+1}(v)$ is uniform over $V_{r+1}$ with no conditioning. 
	We focus on bounding the sum of the mutual information terms using \Cref{fact:entropy-subset}.
	
	Let $\beta \leq \ell/2$ be an integer parameter which we will fix later. For any $S \subset V_r$, we use $\rM_{r+1}(S)$ to denote the set of random variables $\rM_{r+1}(v)$ for $v \in S$. Using the entropy subset inequality from \Cref{fact:entropy-subset}, we get, 
	\begin{align*}
		\frac1{\ell} \cdot \sum_{v \in V_r} \en{\rM_{r+1}(v)\mid \rGam_{r-1}, \rProt_r} &\geq \frac1{\binom{\ell}{\beta}} \sum_{S \subset V_r: \card{S} = \beta}  \frac1{\beta} \cdot \en{\rM_{r+1}(S) \mid \rGam_{r-1}, \rProt_{r}} \\
		&= \frac1{\binom{\ell}{\beta}} \sum_{S \subset V_r: \card{S} = \beta}  \frac1{\beta} \cdot (\en{\rM_{r+1}(S), \rGam_{r-1} \rProt_r } - \en{ \rGam_{r-1}, \rProt_{r}} )\tag{by chain rule of entropy \itfacts{en-chain-rule}}\\
		&\geq \frac1{\binom{\ell}{\beta}} \sum_{S \subset V_r: \card{S} = \beta}  \frac1{\beta} \cdot (\rM_{r+1}(S) - \en{ \rGam_{r-1}, \rProt_{r}} )\tag{by \itfacts{en-chain-rule} and by \itfacts{uniform}} \\
		&= \frac1{\beta} \cdot \paren{\log\paren{\frac{\ell!}{(\ell-\beta)!}} - \en{\rGam_{r-1}, \rProt_r}} \tag{as $\en{\rM_S}$ is a uniform partial permutation to $V_{r+1}$ for any $S$} \\
		&\geq \log(\ell - \beta) - \frac1{\beta} \cdot \en{\rGam_{r-1}, \rProt_r} \tag{as $\frac{(\ell)!}{(\ell-\beta)!} \geq (\ell-\beta)^{\beta}$} \\
		&\geq \log(\ell - \beta) - \frac1{\beta} \cdot r \cdot (s + \log \ell) \tag{as $\rGam_{r-1}, \rProt_r$ consists of totally $r$ messages and pointers, and by \itfacts{uniform}}
	\end{align*}
	Hence, for any integer $0 \leq \beta \leq \ell/2$, we have obtained, 
	\begin{equation}\label{eq:inter-2}
			\frac1{\ell} \cdot \sum_{v \in V_r} \en{\rM_{r+1}(v)\mid \rGam_{r-1}, \rProt_r} \geq \log(\ell - \beta) - \frac1{\beta} \cdot r \cdot (s + \log \ell)
	\end{equation}
	The mutual information term can be bounded easily now. 
	\begin{align*}
		&\frac1{\ell} \cdot \sum_{v \in V_r}  \mi{\rM_{r+1}(v)}{\rGam_{r-1},\rProt_r}  \\
		&=\frac1{\ell} \cdot \sum_{v \in V_r}  (\en{\rM_{r+1}(v)} - \en{\rM_{r+1}(v) \mid \rGam_{r-1},\rProt_r} )\tag{by definition of mutual information} \\
		&= \log \ell - 	\frac1{\ell} \cdot \sum_{v \in V_r} \en{\rM_{r+1}(v)\mid \rGam_{r-1}, \rProt_r}  \tag{as $\rM_{r+1}(v)$ is uniform over $V_{r+1}$} \\
		&\leq \log \ell - \log(\ell - \beta) + \frac1{\beta} \cdot r \cdot (s + \log \ell) \tag{by \Cref{eq:inter-2}} \\
		& = \log \paren{1 + \frac{\beta}{\ell - \beta}} + \frac1{\beta} \cdot r \cdot (s + \log \ell) \\
		&\leq \frac{\beta}{\ell - \beta} \cdot \log e + \frac1{\beta} \cdot r \cdot (s + \log \ell) \tag{as $1+x \leq e^x$} \\
		&\leq \frac1{\ell} \cdot 4 \beta + r\cdot (s+ \log \ell)\cdot \frac1{\beta}. \tag{for any choice of $\beta \leq \ell/2$}
	\end{align*}
We set the value of $\beta$ as $\sqrt{r \cdot (s+ \log \ell) \cdot \ell}$ which is less than $ \ell/2$ for our choice of $s$. The equation simplifies to: 
\begin{equation}\label{eq:inter-3}
	\frac1{\ell} \cdot \sum_{v \in V_r}  \mi{\rM_{r+1}(v)}{\rGam_{r-1},\rProt_r}  \leq 5 \cdot \sqrt{\frac{r \cdot (s +\log \ell)}{\ell}}.
\end{equation}
We plug \Cref{eq:inter-3} into \Cref{eq:inter-1} to get:
\begin{align*}
	\Exp_{\rGam_{r-1}, \rProt_r} \Exp_{\randvert} \tvd{(\rM_{r+1}(\randvert) \mid \rGam_{r-1}, \rProt_r)}{\cUell} \leq \sqrt{\frac12 \cdot 5 \cdot \sqrt{\frac{r \cdot (s +\log \ell)}{\ell}}} \leq 0.1 \cdot \frac1{k},
\end{align*}
for our choice of $s$. This proves the lemma. 
\end{proof}

We will prove \Cref{lem:unif-pc-induction} by induction. We need two more claims. 

	\begin{claim}\label{clm:inter-1}
		For any $r$ with $1 \leq r \leq k-2$, 
	for any value of $(\rGam_{r-1}, \rProt_r)$ and any vertex $v \in V_{r+1}$, we have, 
	\[
	\rM_{r+1}(v) \perp \rpt{r+1} = v \mid \rGam_{r-1}, \rProt_r.
	\]
\end{claim}
\begin{proof}
We prove the claim for even $r$. The case when $r$ is odd follows similarly.
	We can write, 
	\begin{align*}
		&\mi{\rM_{r+1}(v)}{\rpt{r+1} \mid \rGam_{r-1}, \rProt_r}\\
		&\leq \mi{\rM_{r+1}}{\rpt{r+1} \mid \rGam_{r-1}, \rProt_r} \tag{by data processing inequality \itfacts{data-processing}} \\
		&\leq  \mi{\rM_{r+1}}{\rM_r \mid \rGam_{r-1}, \rProt_r} \tag{again by \itfacts{data-processing}, as $\rpt{r+1} = \rM_r(\rpt{r})$, and $\rpt{r}$ is fixed by $\rGam_{r-1}$} \\
		&\leq \mi{\rM_{r+1}}{\rM_r \mid\rProtAl_{\leq r}, \rProtBob_{\leq r}, \rpt{1}, \rpt{2}, \ldots, \rpt{r}}. \tag{expanding $\rGam_{r-1}$} \\
		&\leq  \mi{\rM_{r+1}, \rpt{r}}{\rM_r \mid\rProtAl_{\leq r}, \rProtBob_{\leq r}, \rpt{1}, \rpt{2}, \ldots, \rpt{r-1}} \tag{by \itfacts{chain-rule} and non-negativity of mutual information} \\
		&\leq \mi{\rM_{r+1}, \rM_{r-1}}{\rM_r \mid \rProtAl_{\leq r}, \rProtBob_{\leq r}, \rpt{1}, \rpt{2}, \ldots, \rpt{r-1}} \tag{by \itfacts{data-processing}} \\
		&\leq \mi{\rM_{r+1}, \rM_{r-1}, \ldots, \rM_1}{\rM_r, \rM_{r-2}, \ldots, \rM_2 \mid   \rProtAl_{\leq r}, \rProtBob_{\leq r}} \tag{by repeating the earlier two steps for other pointers} \\
		&\leq  \mi{\rModd}{\rMeven \mid   \rProtAl_{\leq r}, \rProtBob_{\leq r}} \tag{by \itfacts{data-processing}} \\
		&\leq \mi{\rModd, \rProtAl_r}{\rMeven \mid \rProtAl_{\leq (r-2)}, \rProtBob_{\leq (r-1)}} \tag{by \itfacts{chain-rule} and non-negativity of mutual information} \\
		&= \mi{\rModd}{\rMeven \mid \rProtAl_{\leq (r-2)}, \rProtBob_{\leq (r-1)}} \tag{by \itfacts{data-processing} as $\rProtAl_r$ if fixed by $\rModd$ and $\rProtBob_{\leq r-1}$} \\
		&\leq \mi{\rModd}{\rMeven, \rProtBob_{ (r-1)} \mid \rProtAl_{\leq (r-2)}, \rProtBob_{\leq (r-3)}}  \tag{by \itfacts{chain-rule} and non-negativity of mutual information} \\
		&= \mi{\rModd}{\rMeven \mid \rProtAl_{\leq (r-2)}, \rProtBob_{\leq (r-3)}}  \tag{as $\rProtBob_{ (r-1)}$ if fixed by $\rMeven, \rProtAl_{\leq (r-2)}$, and by \itfacts{data-processing}} \\
		&\leq \mi{\rModd}{\rMeven} \tag{by repeating earlier two steps for all messages} \\
		&= 0. \tag{as all the matchings are independent of each other and \itfacts{info-zero}}
	\end{align*}
\itfacts{info-zero} now completes the claim.
\end{proof}

\begin{claim}\label{clm:inter-2}
	For any $1 \leq r \leq k-2$, we have, 
	\[
		\rpt{r+1} \perp \rProt_r \mid \rGam_{r-1}.
	\]
\end{claim}

\begin{proof}
	We only prove the statement when $r $ is even. The proof follows in the same way for odd $r$. 
	\begin{align*}
		&\mi{\rpt{r+1}}{\rProt_r \mid \rGam_{r-1}} \\
		&\leq \mi{\rpt{r+1}}{\rProt_r \mid \rProtAl_{\leq r-1}, \rProtBob_{\leq r-1}, \rpt{1}, \rpt{2}, \ldots, \rpt{r}} \tag{expanding $\rGam_{r-1}$} \\
		&\leq \mi{\rM_r}{\rProt_r \mid \rProtAl_{\leq r-1}, \rProtBob_{\leq r-1}, \rpt{1}, \rpt{2}, \ldots, \rpt{r}} \tag{by \itfacts{data-processing}, as $\rpt{r+1} = \rM_r(\rpt{r})$} \\
		&\leq  \mi{\rM_r}{\rProt_r, \rpt{r} \mid \rProtAl_{\leq r-1}, \rProtBob_{\leq r-1}, \rpt{1}, \rpt{2}, \ldots, \rpt{r-1}} \tag{by \itfacts{chain-rule}} \\
		&\leq   \mi{\rM_r}{\rProt_r, \rM_{r-1} \mid \rProtAl_{\leq r-1}, \rProtBob_{\leq r-1}, \rpt{1}, \rpt{2}, \ldots, \rpt{r-1}} \tag{by \itfacts{data-processing}} \\
		&\leq \mi{\rM_r, \rM_{r-2}, \ldots, \rM_2}{\rProt_r, \rM_{r-1}, \rM_{r-3}, \ldots, \rM_1 \mid \rProtAl_{\leq r-1}, \rProtBob_{\leq r-1}} \tag{repeating the earlier two steps for all other pointers} \\
		&\leq \mi{\rMeven}{\rProtAl_{ r}, \rModd \mid  \rProtAl_{\leq r-1}, \rProtBob_{\leq r-1}} \\
		&=  \mi{\rMeven}{ \rModd \mid  \rProtAl_{\leq r-1}, \rProtBob_{\leq r-1}} \tag{by \itfacts{data-processing} as $\rProtAl_r$ is fixed by $\rModd, \rProtBob_{\leq (r-1)}$} \\
		&\leq  \mi{\rMeven, \rProtBob_{(r-1)}}{ \rModd \mid  \rProtAl_{\leq r-1}, \rProtBob_{\leq r-3}} \tag{by \itfacts{chain-rule} and non-negativity of mutual information} \\
		&=   \mi{\rMeven}{ \rModd \mid  \rProtAl_{\leq r-1}, \rProtBob_{\leq r-3}}\tag{as $\rProtBob_{ (r-1)}$ is fixed by $\rMeven$ and $\rProtAl_{\leq (r-1)}$}\\
		&\leq \mi{\rMeven}{\rModd} \tag{by repeating earlier two steps for all messages} \\
		&= 0. \tag{as the matchings are independent of each other}
	\end{align*}
The proof can be finished using \itfacts{info-zero}.
\end{proof}

\begin{proof}[Proof of \Cref{lem:unif-pc-induction}]
	The base case is when $r = 0$. Then, $\rGam_0 = \rpt{1} $. Conditioned on this value, the next pointer $\rpt{2} $ remains uniform. Thus, both LHS and RHS are zero. 
	
	Let us assume that the following statement is true for the induction hypothesis:
	\begin{equation}\label{eq:inter-4}
		\Exp_{\rGam_{r-1}} \tvd{(\rpt{r+1} \mid \rGam_{r-1})}{\cUell} \leq 0.1 \cdot \frac{(r-1)}{k}.
	\end{equation}
	
	We want to prove the statement for $r$ with $1 \leq r \leq k-2$.
	\begin{align}
		&\Exp_{\rGam_{r}} \tvd{(\rpt{r+2} \mid \rGam_{r})}{\cUell} \notag\\
		&= \Exp_{\rGam_{r-1}, \rProt_r} \Exp_{\rpt{r+1} \mid \rGam_{r-1}, \rProt_r} \tvd{(\rpt{r+2} \mid \rGam_{r-1}, \rProt_r, \rpt{r+1})}{\cUell} \notag\\
		&= \Exp_{\rGam_{r-1}, \rProt_r} \Exp_{v \sim \rpt{r+1} \mid \rGam_{r-1}, \rProt_r} \tvd{(\rM_{r+1}(v) \mid \rGam_{r-1}, \rProt_r, \rpt{r+1} = v)}{\cUell} \notag\\
		&\leq  \Exp_{\rGam_{r-1}, \rProt_r} \paren{\Exp_{v \sim V_{r+1}} \tvd{(\rM_{r+1}(v) \mid \rGam_{r-1}, \rProt_r, \rpt{r+1} = v)}{\cUell} + \tvd{(\rpt{r+1} \mid \rGam_{r-1}, \rProt_r)}{\cUell}}, \label{eq:inter-5}
	\end{align}  
	where $v \sim V_{r+1}$  is drawn uniformly at random from $V_{r+1}$. The inequality in the last step follows from \Cref{fact:tvd-small}.
	We bound the two terms in \Cref{eq:inter-5} separately.
	
	For the first term: 
	\begin{align}
	& \Exp_{\rGam_{r-1}, \rProt_r}\Exp_{v \sim V_{r+1}} \tvd{(\rM_{r+1}(v) \mid \rGam_{r-1}, \rProt_r, \rpt{r+1} = v)}{\cUell} \notag \\
	&= \Exp_{\rGam_{r-1}, \rProt_r}\Exp_{v \sim V_{r+1}} \tvd{(\rM_{r+1}(v) \mid \rGam_{r-1}, \rProt_r)}{\cUell} \tag{by \Cref{clm:inter-1}} \\
	&\leq 0.1 \cdot (1/k). \tag{by \Cref{lem:high-ent-perm}}
	\end{align} 
	
	For the second term:
	\begin{align}
		&\Exp_{\rGam_{r-1}, \rProt_r} \tvd{(\rpt{r+1} \mid \rGam_{r-1}, \rProt_r)}{\cUell} \notag \\
		&= \Exp_{\rGam_{r-1}} \tvd{(\rpt{r+1} \mid \rGam_{r-1})}{\cUell} \tag{by \Cref{clm:inter-2}} \\
		& \leq 0.1 \cdot (r-1) \cdot (1/k). \tag{by induction hypothesis from \Cref{eq:inter-4}} 
	\end{align}
	Plugging the two upper bounds in \Cref{eq:inter-5} completes the proof.
\end{proof}

\section{Background on Information Theory}\label{app:info}

We now briefly introduce some definitions and facts from information theory that are used in our proofs. We refer the interested reader to the text by Cover and Thomas~\cite{CoverT06} for an excellent introduction to this field, 
and the proofs of the statements used in this Appendix. 

For a random variable $\rA$, we use $\supp{\rA}$ to denote the support of $\rA$ and $\distribution{\rA}$ to denote its distribution. 
When it is clear from the context, we may abuse the notation and use $\rA$ directly instead of $\distribution{\rA}$, for example, write 
$A \sim \rA$ to mean $A \sim \distribution{\rA}$, i.e., $A$ is sampled from the distribution of random variable $\rA$. 

\begin{itemize}[leftmargin=10pt]
	\item We denote the \emph{Shannon Entropy} of a random variable $\rA$ by
	$\en{\rA}$, which is defined as: 
	\begin{align}
		\en{\rA} := \sum_{A \in \supp{\rA}} \Pr\paren{\rA = A} \cdot \log{\paren{1/\Pr\paren{\rA = A}}} \label{eq:entropy}
	\end{align} 
	\noindent
	\item The \emph{conditional entropy} of $\rA$ conditioned on $\rB$ is denoted by $\en{\rA \mid \rB}$ and defined as:
	\begin{align}
		\en{\rA \mid \rB} := \Ex_{B \sim \rB} \bracket{\en{\rA \mid \rB = B}}, \label{eq:cond-entropy}
	\end{align}
	where 
	$\en{\rA \mid \rB = B}$ is defined in a standard way by using the distribution of $\rA$ conditioned on the event $\rB = B$ in \Cref{eq:entropy}.

	\item The \emph{mutual information} of two random variables $\rA$ and $\rB$ is denoted by
	$\mi{\rA}{\rB}$ and is defined:
	\begin{align}
		\mi{\rA}{\rB} := \en{A} - \en{A \mid  B} = \en{B} - \en{B \mid  A}. \label{eq:mi}
	\end{align}
	\noindent
	\item The \emph{conditional mutual information} $\mi{\rA}{\rB \mid \rC}$ is $\en{\rA \mid \rC} - \en{\rA \mid \rB,\rC}$ and hence by linearity of expectation:
	\begin{align}
		\mi{\rA}{\rB \mid \rC} = \Ex_{C \sim \rC} \bracket{\mi{\rA}{\rB \mid \rC = C}}. \label{eq:cond-mi}
	\end{align} 
\end{itemize}

\subsection{Useful Properties of Entropy and Mutual Information}\label{sec:prop-en-mi}

We shall use the following basic properties of entropy and mutual information throughout. 

\begin{fact}\label{fact:it-facts}
	Let $\rA$, $\rB$, $\rC$, and $\rD$ be four (possibly correlated) random variables.
	\begin{enumerate}
		\item \label{part:uniform} $0 \leq \en{\rA} \leq \log{\card{\supp{\rA}}}$. The right equality holds
		iff $\distribution{\rA}$ is uniform.
		\item \label{part:info-zero} $\mi{\rA}{\rB \mid \rC} \geq 0$. The equality holds iff $\rA$ and
		$\rB$ are \emph{independent} conditioned on $\rC$.
		\item \label{part:info-entropy} $\mi{\rA}{\rB \mid \rC} \leq \en{\rB}$ for any random variables $\rA, \rB, \rC$. 
		\item \label{part:cond-reduce} \emph{Conditioning on a random variable reduces entropy}:
		$\en{\rA \mid \rB,\rC} \leq \en{\rA \mid  \rB}$.  The equality holds iff $\rA \perp \rC \mid \rB$.
		%%    \item \label{part:sub-additivity} \emph{Subadditivity of entropy}: $\en{\rA,\rB \mid \rC}
		%%    \leq \en{\rA \mid C} + \en{\rB \mid  \rC}$.
		%%   \item \label{part:ent-chain-rule} \emph{Chain rule for entropy}: $\en{\rA,\rB \mid \rC} = \en{\rA \mid \rC} + \en{\rB \mid \rC,\rA}$.
		\item \label{part:chain-rule} \emph{Chain rule for mutual information}: $\mi{\rA,\rB}{\rC \mid \rD} = \mi{\rA}{\rC \mid \rD} + \mi{\rB}{\rC \mid  \rA,\rD}$.
		\item \label{part:en-chain-rule} \emph{Chain rule for entropy}: $\en{\rA, \rB} = \en{\rA} + \en{\rB \mid \rA}$.
		\item \label{part:data-processing} \emph{Data processing inequality}: for a function $f(\rA)$ of $\rA$, $\mi{f(\rA)}{\rB \mid \rC} \leq \mi{\rA}{\rB \mid \rC}$. 
	\end{enumerate}
\end{fact}

\begin{fact}[Entropy Subset Inequality \cite{Te78}]\label{fact:entropy-subset}
	For any set of $n$ random variables $\rX_1, \rX_2, \ldots, \rX_n$, and set $S \subset [n]$, we define,
	$\rX_S :=\set{\rX_i \mid i \in S}$. For every $k \in [n]$, we define:
	\[
		\ents{k}{\rX} := \frac1{\binom{n}{k}} \sum_{S \subset [n]: \card{S} = k} \frac1{k} \cdot \en{\rX_S}.
	\]
	Then, $\ents{1}{X}^{(1)} \geq \ents{2}{X} \geq \ldots \geq \ents{n}{X}$. This inequality holds for conditional entropy also. 
\end{fact}

%\noindent
%We also use the following two standard propositions, regarding the effect of conditioning on mutual information.

%\begin{proposition}\label{prop:info-increase}
%	For random variables $\rA, \rB, \rC, \rD$, if $\rA \perp \rD \mid \rC$, then, 
%	\[\mi{\rA}{\rB \mid \rC} \leq \mi{\rA}{\rB \mid  \rC,  \rD}.\]
%\end{proposition}
%\begin{proof}
%	Since $\rA$ and $\rD$ are independent conditioned on $\rC$, by
%	\itfacts{cond-reduce}, $\HH(\rA \mid  \rC) = \HH(\rA \mid \rC, \rD)$ and $\HH(\rA \mid  \rC, \rB) \ge \HH(\rA \mid  \rC, \rB, \rD)$.  We have,
%	\begin{align*}
%		\mi{\rA}{\rB \mid  \rC} &= \HH(\rA \mid \rC) - \HH(\rA \mid \rC, \rB) = \HH(\rA \mid  \rC, \rD) - \HH(\rA \mid \rC, \rB) \\
%		&\leq \HH(\rA \mid \rC, \rD) - \HH(\rA \mid \rC, \rB, \rD) = \mi{\rA}{\rB \mid \rC, \rD}. 
%	\end{align*}
%\end{proof}
%
%\begin{proposition}\label{prop:info-decrease}
%	For random variables $\rA, \rB, \rC,\rD$, if $ \rA \perp \rD \mid \rB,\rC$, then, 
%	\[\mi{\rA}{\rB \mid \rC} \geq \mi{\rA}{\rB \mid \rC, \rD}.\]
%\end{proposition}
%\begin{proof}
%	Since $\rA \perp \rD \mid \rB,\rC$, by \itfacts{cond-reduce}, $\HH(\rA \mid \rB,\rC) = \HH(\rA \mid \rB,\rC,\rD)$. Moreover, since conditioning can only reduce the entropy (again by \itfacts{cond-reduce}), 
%	\begin{align*}
%		\mi{\rA}{\rB \mid  \rC} &= \HH(\rA \mid \rC) - \HH(\rA \mid \rB,\rC) \geq \HH(\rA \mid \rD,\rC) - \HH(\rA \mid \rB,\rC) \\
%		&= \HH(\rA \mid \rD,\rC) - \HH(\rA \mid \rB,\rC,\rD) = \mi{\rA}{\rB \mid \rC,\rD}.   
%	\end{align*}
%	
%\end{proof}

\subsection{Measures of Distance Between Distributions}\label{sec:prob-distance}

We use two main measures of distance (or divergence) between distributions, namely the \emph{Kullback-Leibler divergence} (KL-divergence) and the \emph{total variation distance}. 

\paragraph{KL-divergence.} For two distributions $\mu$ and $\nu$ over the same probability space, the \textbf{Kullback-Leibler (KL) divergence} between $\mu$ and $\nu$ is denoted by $\kl{\mu}{\nu}$ and defined as: 
\begin{align}
	\kl{\mu}{\nu}:= \Ex_{a \sim \mu}\Bracket{\log\frac{\mu(a)}{{\nu}(a)}}. \label{eq:kl}
\end{align}
We have the following relation between mutual information and KL-divergence. 
\begin{fact}\label{fact:kl-info}
	For random variables $\rA,\rB,\rC$, 
	\[\mi{\rA}{\rB \mid \rC} = \Ex_{(B,C) \sim {(\rB,\rC)}}\Bracket{ \kl{\distribution{\rA \mid \rB=B,\rC=C}}{\distribution{\rA \mid \rC=C}}}.\] 
\end{fact}

\paragraph{Total variation distance.} We denote the \textbf{total variation distance} between two distributions $\mu$ and $\nu$ on the same 
support $\Omega$ by $\tvd{\mu}{\nu}$, defined as: 
\begin{align}
	\tvd{\mu}{\nu}:= \max_{\Omega' \subseteq \Omega} \paren{\mu(\Omega')-\nu(\Omega')} = \frac{1}{2} \cdot \sum_{x \in \Omega} \card{\mu(x) - \nu(x)}.  \label{eq:tvd}
\end{align}
\noindent
We use the following basic properties of total variation distance. 
\begin{fact}\label{fact:tvd-small}
	Suppose $\mu$ and $\nu$ are two distributions for $\event$, then, 
	$
	{\mu}(\event) \leq {\nu}(\event) + \tvd{\mu}{\nu}.
	$
\end{fact}

\begin{fact}\label{fact:tvd-small-event}
	For any distribution $\mu$ and event $\event$, \[
		\tvd{\mu}{(\mu \mid \neg \event)} \leq \Pr[\event].
	\]
\end{fact}

%%\begin{fact}\label{fact:tvd-sample}
%%	Suppose $\mu$ and $\nu$ are two distributions with same support $\Omega$; then, given a single sample from either $\mu$ or $\nu$, the best probability of successfully deciding whether $s$ came from $\mu$ or $\nu$ (achieved by the maximum likelihood estimator) is 
%%	\[
%%	\frac12 + \frac12\cdot\tvd{\mu}{\nu}.
%%	\]
%%\end{fact}

We also have the following (chain-rule) bound on the total variation distance of joint variables.

\begin{fact}\label{fact:tvd-chain-rule}
	For any distributions $\mu$ and $\nu$ on $n$-tuples $(X_1,\ldots,X_n)$, 
	\[
	\tvd{\mu}{\nu} \leq \sum_{i=1}^{n} \Exp_{X_{<i} \sim \mu} \tvd{\mu(X_i \mid X_{<i})}{\nu(X_i \mid X_{<i})}. 
	\]
\end{fact}

%We also have the following ``over conditioning'' property. 

%\begin{fact}\label{fact:tvd-over-conditioning}
%	For any random variables $\rX,\rY,\rZ$, 
%	\[
%	\tvd{\rX}{\rY} \leq \tvd{\rX\rZ}{\rY\rZ} = \Exp_{Z} \tvd{(\rX \mid \rZ=Z)}{(\rY \mid \rZ=Z)}.  
%	\]
%\end{fact}
%\begin{proof}
%	First, we prove the equality between the second term and the third term in the statement. 
%	\begin{align*}
%		\tvd{\rX\rZ}{\rY\rZ} &=\frac12 \cdot \sum_{W, Z} \Pr[(W, Z)] \card{\Pr[\rX\rZ = (W, Z)] - \Pr[\rY\rZ=(W, Z)]}  \\
%		&=\frac12 \cdot \sum_{W, Z}  \card{\Pr[\rZ = Z] \cdot (\Pr[\rX = W \mid  \rZ = Z] - \Pr[\rY = W \mid \rZ= Z])} \\
%		&= \sum_{Z} \Pr[\rZ =Z] \cdot \frac12 \sum_{W}  \card{\Pr[\rX = W \mid  \rZ = Z] - \Pr[\rY = W \mid \rZ= Z]} \\
%		&= \sum_{Z} \Pr[\rZ =Z] \cdot\tvd{(\rX \mid \rZ= Z)}{(\rY \mid \rZ = Z)} \\
%		&= \Exp_{Z} \tvd{(\rX \mid \rZ= Z)}{(\rY \mid \rZ = Z)}.
%	\end{align*}
%	
%	Now we prove the inequality between the first term and the third term. 
%	\begin{align*}
%		\tvd{\rX}{\rY} &= \frac12 \cdot \sum_{W} \card{\Pr[\rX = W] - \Pr[\rY = W]} \\
%		&= \frac12 \cdot \sum_{W} \card{\sum_{Z} \Pr[\rZ = Z](\Pr[\rX = W \mid \rZ = Z] - \Pr[\rY = W \mid \rZ = Z])} \\
%		& \leq  \frac12 \cdot \sum_{W}  \sum_{Z} \Pr[\rZ = Z] \card{\Pr[\rX = W \mid \rZ = Z] - \Pr[\rY = W \mid \rZ = Z]} \\
%		&= \sum_{Z}   \Pr[\rZ = Z] \cdot \paren{\frac12 \cdot \sum_{W} \card{\Pr[\rX = W \mid \rZ = Z] - \Pr[\rY = W \mid \rZ = Z]}}  \\
%		&= \Exp_{Z} \tvd{\rX \mid \rZ= Z}{\rY \mid \rZ = Z}. 
%	\end{align*}
%\end{proof}

The following Pinsker's inequality bounds the total variation distance between two distributions based on their KL-divergence.

\begin{fact}[Pinsker's inequality]\label{fact:pinskers}
	For any distributions $\mu$ and $\nu$, 
	$
	\tvd{\mu}{\nu} \leq \sqrt{\frac{1}{2} \cdot \kl{\mu}{\nu}}.
	$ 
\end{fact}

\end{document}